%% file: jmlr.tex
\documentclass[twoside,11pt]{article}
\usepackage{blindtext}

\RequirePackage[OT1]{fontenc}
\RequirePackage{amsthm,amsmath}

\usepackage{jmlr2e}
\hypersetup{ hidelinks }
\usepackage{graphicx}  

\usepackage{times}
\usepackage{booktabs}       
\usepackage{amsfonts}       
\usepackage{nicefrac}       
\usepackage{microtype}      
\usepackage{algorithm}
\usepackage{bbold}
\usepackage{tikz}
\usetikzlibrary{decorations.pathreplacing}
\usetikzlibrary{patterns}
\usepackage{amsmath}
\usepackage{algorithmic}
\usepackage{verbatim}
\usepackage{subcaption}
\usepackage{graphicx}
\usepackage{multirow}
\usepackage{amssymb}
\usepackage{color}
\usepackage{enumitem}
\usepackage{hyphenat}
\usepackage{array}
\usepackage{booktabs}
\usepackage{wrapfig}
\usepackage{relsize}
\usepackage{tikz}

\usepackage{afterpage}
\usetikzlibrary{arrows}

\DeclareGraphicsExtensions{.eps}

\newcommand\nc\newcommand
\nc\bfa{{\boldsymbol a}}\nc\bfA{{\boldsymbol A}}\nc\cA{{\mathcal A}}
\nc\bfb{{\boldsymbol b}}\nc\bfB{{\boldsymbol B}}\nc\cB{{\mathcal B}}
\nc\bfc{{\boldsymbol c}}\nc\bfC{{\boldsymbol C}}\nc\cC{{\mathcal C}}
\nc\sC{{\mathscr C}}
\nc\bfd{{\boldsymbol d}}\nc\bfD{{\boldsymbol D}}\nc\cD{{\mathcal D}}
\nc\bfe{{\boldsymbol e}}\nc\bfE{{\boldsymbol E}}\nc\cE{{\mathcal E}}
\nc\bff{{\boldsymbol f}}\nc\bfF{{\boldsymbol F}}\nc\cF{{\mathcal F}}
\nc\bfg{{\boldsymbol g}}\nc\bfG{{\boldsymbol G}}\nc\cG{{\mathcal G}}
\nc\bfh{{\boldsymbol h}}\nc\bfH{{\boldsymbol H}}\nc\cH{{\mathcal H}}
\nc\bfi{{\boldsymbol i}}\nc\bfI{{\boldsymbol I}}\nc\cI{{\mathcal I}}
\nc\bfj{{\boldsymbol j}}\nc\bfJ{{\boldsymbol J}}\nc\cJ{{\mathcal J}}
\nc\bfk{{\boldsymbol k}}\nc\bfK{{\boldsymbol K}}\nc\cK{{\mathcal K}}
\nc\bfl{{\boldsymbol l}}\nc\bfL{{\boldsymbol L}}\nc\cL{{\mathcal L}}
\nc\bfm{{\boldsymbol m}}\nc\bfM{{\boldsymbol M}}\nc\sM{{\mathscr M}}\nc\cM{{\mathcal M}}
\nc\bfn{{\boldsymbol n}}\nc\bfN{{\boldsymbol N}}\nc\cN{{\mathcal N}}
\nc\bfo{{\boldsymbol o}}\nc\bfO{{\boldsymbol O}}\nc\cO{{\mathcal O}}
\nc\bfp{{\boldsymbol p}}\nc\bfP{{\boldsymbol P}}\nc\cP{{\mathcal P}}
\nc\bfq{{\boldsymbol q}}\nc\bfQ{{\boldsymbol Q}}\nc\cQ{{\mathcal Q}}
\nc\bfr{{\boldsymbol r}}\nc\bfR{{\boldsymbol R}}\nc\cR{{\mathcal R}}
\nc\bfs{{\boldsymbol s}}\nc\bfS{{\boldsymbol S}}\nc\cS{{\mathcal S}}
\nc\bft{{\boldsymbol t}}\nc\bfT{{\boldsymbol T}}\nc\cT{{\mathcal T}}
\nc\bfu{{\boldsymbol u}}\nc\bfU{{\boldsymbol U}}\nc\cU{{\mathcal U}}
\nc\bfv{{\boldsymbol v}}\nc\bfV{{\boldsymbol V}}\nc\cV{{\mathcal V}}
\nc\bfw{{\boldsymbol w}}\nc\bfW{{\boldsymbol W}}\nc\cW{{\mathcal W}}
\nc\bfx{{\boldsymbol x}}\nc\bfX{{\boldsymbol X}}\nc\cX{{\mathcal X}}
\nc\bfy{{\boldsymbol y}}\nc\bfY{{\boldsymbol Y}}\nc\cY{{\mathcal Y}}
\nc\bfz{{\boldsymbol z}}\nc\bfZ{{\boldsymbol Z}}\nc\cZ{{\mathcal Z}}

\nc\diff{{\mathrm d}}
\nc\e{{\mathrm e}}
\nc\calC{{\mathcal C}}

\newcommand{\remove}[1]{}
\newcommand{\correct}[1]{\textcolor{black}{#1}}
\newcommand{\latest}[1]{#1}

\newcommand{\avg}{{\mathbb E}}

\newcommand{\dist}{d_{L}}

\numberwithin{equation}{section}
\theoremstyle{plain}
\newtheorem*{theorem*}{Theorem}
\newtheorem*{lemma*}{Lemma}
\newtheorem*{corollary*}{Corollary}
\newtheorem{observation}[theorem]{Observation}

\newcommand\reals{{\mathbb R}}

\newcommand\integers{{\mathbb Z}}

\newcommand{\norm}[1]{||#1||}

\allowdisplaybreaks
\usepackage{lastpage}

\ShortHeadings{Community Recovery in the Geometric Block Model}{Galhotra, Mazumdar, Pal and Saha}
\firstpageno{1}

\jmlrheading{24}{2023}{1-\pageref{LastPage}}{5/22}{11/23}{22-0572}{Sainyam Galhotra, Arya Mazumdar, Soumyabrata Pal, and Barna Saha}

\begin{document}

\title{Community Recovery in the Geometric Block Model}

\author{\name Sainyam Galhotra \email sg@cs.cornell.edu \\
       \addr Department of Computer Science\\
       Cornell University\\
       Bill and Melinda Gates Hall, \\
       Ithaca, NY 14850 USA
       \AND
       \name Arya Mazumdar \email arya@ucsd.edu \\
       \addr Halicioglu Data Science Institute\\
       University of California, San Diego\\
       9500 Gilman Dr \\ 
       La Jolla, CA 92093 USA
       \AND
       \name Soumyabrata\ Pal \email soumyabrata@google.com \\
       \addr Google Research\\
       Bengaluru 560016, India
       \AND 
       \name Barna Saha \email barnas@ucsd.edu \\
       \addr Department of Computer Science and Halicioglu Data Science Institute\\
       University of California, San Diego\\
       9500 Gilman Dr \\ 
       La Jolla, CA 92093 USA
       }

\editor{Alexandre Proutiere}

\maketitle

\begin{abstract}
To capture the inherent geometric features of many community detection problems, we propose to use a new random graph model of communities that we call a {\em Geometric Block Model}. The geometric block model builds on the {\em random geometric graphs} (Gilbert, 1961), one of the  basic models of random graphs for spatial networks, in the same way that the well-studied stochastic block model builds on the Erd\H{o}s-R\'{en}yi random graphs. It is also a natural extension of random community models inspired by the recent theoretical and practical advancements in community detection. 
%
To analyze the geometric block model, we first provide new connectivity results for  {\em random annulus graphs} which are  generalizations of random geometric graphs. 
The connectivity properties of geometric graphs have been studied since their introduction, and analyzing them has been more difficult than their  Erd\H{o}s-R\'{en}yi counterparts due to correlated edge formation.


We then use the connectivity results of random annulus graphs to provide necessary and sufficient conditions for efficient recovery of communities for  the geometric block model. 
We show that a simple triangle-counting algorithm to detect communities in the geometric block model is near-optimal. For this we consider the following two regimes of graph density.

In the regime where the average degree of the graph grows  logarithmically with the number of vertices, we show that our algorithm performs extremely well, both theoretically and practically. In contrast, the triangle-counting algorithm is far from being optimum for the stochastic block model in the logarithmic degree regime. We simulate our results on both real and synthetic datasets to show superior performance of both the new model as well as our algorithm. 


\end{abstract}

\section{Introduction}

{

Clustering of graphs is a ubiquitous problem  where the objective is to partition the vertices of a graph into disjoint clusters such that each cluster is more densely connected internally than across clusters. Many real-world datasets can be represented in the form of graphs where the vertices (nodes) represent data elements and the edges represent (noisy) interaction between the elements. It is of primary interest to recover latent clusters in the graph either as an end goal or as a means to some other learning problem. As an example, consider the graph induced by the collection of political blogs in internet where each vertex corresponds to a blog website,  and an edge exists between two vertices if one of the corresponding blogs hyperlinks to the other. From this graph, it might be of interest to partition the blogs into two clusters. Achieving this objective allows a recommendation engine to recommend each blog to the correct audience or suggest correct political advertisements and thereby garner more views.

A very simple algorithm to perform this task is the well-known \textit{min-cut} algorithm which intends to partition the graph into clusters such that minimum number of edges go across the clusters. The min-cut algorithm runs in polynomial time but it does not have any guarantee on the sizes of the clusters returned. If we constrain the output clusters to be equally sized, then this problem, also known as the \textit{min-bisection} problem becomes NP-hard. However, most real-world datasets are not pathological and there exist many properties which are satisfied by the graph, for example, sparsity and transitivity. Hence, an approach to resolve this problem is to devise a simple and elegant modeling assumption according to which the observed real-world graphs are generated. Under this assumption, we can design efficient algorithms that can recover the latent clusters. Such modeling assumptions not only allow a rigorous theoretical treatment, but also allow us to benchmark and compare different heuristics for graph partitioning. Finally, note that these models must capture the inherent properties of real world graphs so that the algorithms designed under the corresponding assumption work well on real world datasets as well.

The {\em planted-partition} model or the {\em stochastic block model} (SBM) is an example of such a random graph model for community detection that generalizes the well-known \latest{Erd\H{o}s-R\'{en}yi graphs} \citep{holland1983stochastic,dyer1989solution,decelle2011asymptotic,abbe2015community,abh:16,DBLP:conf/colt/HajekWX15,chin2015stochastic,mossel2015consistency}.
Consider a graph $G(V,E)$, where $V = V_1 \sqcup V_2 \sqcup \dots \sqcup V_k$ is a disjoint union of $k$ clusters denoted by $V_1, \dots, V_k.$ The edges of the graph are drawn randomly: there is an edge between $u \in V_i$ and $v \in V_j$ with probability $q_{i,j}, 1\le i,j \le k.$

This model has been  popular both in theoretical and practical domains of community detection, and the aforementioned references are just a small sample. Recent theoretical works  focus on characterizing sharp threshold of recovering the partition in the SBM. For example, when there are only two communities of exactly equal size, and the inter-cluster edge probability $q_{i,j}=\frac{b\log n}{n}, i \ne j$ and intra-cluster edge probability is $q_{i,i}=\frac{a\log n}{n}$, it is known that perfect recovery is possible if and only if $\sqrt{a} - \sqrt{b} > \sqrt{2}$ 
\citep{abh:16,mossel2015consistency}. The regime of the probabilities being $\Theta\Big(\frac{\log n}{n}\Big)$ has been put forward as one of most interesting ones, \latest{because  in an Erd\H{o}s-R\'{en}yi random graph}, this is the threshold for graph connectivity \citep{bollobas1998random}. This result has been subsequently generalized for $k$ communities \citep{abbe2015community,abbe2015recovering,hajek2016achieving} (for constant $k$ or when $k=o(\log{n})$), and under the assumption that the communities are generated according to a probabilistic generative model (there is a prior probability $p_i$ of an element being in the $i$th community) \citep{abbe2015community}. Note that, the results are not only of theoretical interest, many real-world networks exhibit a ``sparsely connected'' community feature \citep{leskovec2008statistical}, and any efficient recovery algorithm for SBM has many potential applications.  

One aspect that the SBM does not account for is a ``transitivity rule'' (`friends having common friends') inherent to many social and other community structures. To be precise, consider any three vertices $x, y$ and $z$. If $x$ and $y$ are connected by an edge, and $y$ and $z$ are connected by an edge, then it is more likely than not that $x$ and $z$ are connected by an edge. This phenomenon can be seen in many network structures - predominantly in social networks, blog-networks and advertising. SBM, primarily a generalization of \latest{Erd\H{o}s-R\'{en}yi random graph}, does not take into account this characteristic, and in particular, the event that an edge exists between $x$ and $z$ is independent of the events that there exist edges between $x$ and $y$ and $y$ and $z$. However, one needs to be careful such that by allowing such ``transitivity'', the simplicity and elegance of the SBM is not lost.

Further, from an algorithmic point of view, it is well known that triangle based heuristic algorithms work really well for graph clustering tasks (see \citep{tsourakakis2009approximate}, \citep{tsourakakis2017scalable}, \citep{kolountzakis2012efficient}). In particular, in \citep{tsourakakis2017scalable}, the authors proposed a triangle counting based heuristic algorithm \texttt{TECTONIC} that has better performance on Amazon, DBLP and YouTube datasets for graph partitioning than many popular competitors including Spectral Clustering, \latest{Girvan-Newman} algorithm, Louvain method and the Clauset-Newman-Moore (CNM) to name a few. Triangle (and motifs in general) based analytics has been extensively useful in different areas including social networks~\citep{newman2002random}, metabolic networks~\citep{milo2002network}, protein networks~\citep{alon2007network,vinayagam2016controllability}, transportation networks~\citep{rosvall2014memory}, neural networks~\citep{park2013structural}, and food webs~\citep{stouffer2012evolutionary}. Specifically, clustering methods that use triangles capture different interaction patterns. For example, optimizing to minimize the number of triangles across the cut returns similar energy flow patterns across species in the food web data~\citep{benson2016higher}. Further, it was shown in \citep{watts1998collective} that triangles are stronger signals of community structure than edges alone. This begs the following question: ``Is it possible to provide a theoretical justification of why simple triangle based heuristics perform well on real-world data?''. 
It turns out that triangle counting can not recover the latent clusters in SBM in the logarithmic degree regime, which is one of the interesting regimes. This further limits the applicability of algorithms designed using the SBM assumption thus raising the possibility of a different random graph model that captures significantly more properties of real-world graphs.

Inspired by the above questions, we propose a novel random graph community detection model analogous to the stochastic block model, that we call the {\em geometric block model} (GBM). The GBM depends on the basic definition of the well-known class of random graphs called {\em random geometric graph} (RGG) \citep{gilbert1961random} that has found a lot of practical use in wireless networking because of its inclusion of the notion of proximity between nodes \citep{penrose2003random}. The GBM satisfies several desiderata of real networks, such as the degree
associativity property (high degree nodes tend to connect). However, analyzing the GBM is significantly
challenging due to the presence of correlated edges and in particular,  the techniques developed for the
SBM \latest{do not extend} to this setting. 

More precisely, the GBM  is defined as follows.
 Let $V\equiv  V_1 \sqcup V_2 \sqcup \dots \sqcup V_k$ be the set of vertices that is a disjoint union of $k$ clusters, denoted by $V_1, \dots , V_k$. Let, $\beta_{i,j}, i,j \in \{1,\dots,k\}$ be unknown latent variables. 
Given an integer $t\geq 1$, for each vertex $u \in V$, define a random vector $Z_u \in \reals^{t+1}$ that is uniformly distributed in $\cS^{t} \subset \reals^{t+1},$ the $t$-dimensional sphere. 
In this random graph, an edge exists between $v \in V_i$ and $u \in V_j$ if and only if $\langle Z_u, Z_v\rangle \ge \beta_{i,j}$. 
In this special case of $t=1$, the above definition is equivalent to choosing  random variable $\theta_u$ uniformly distributed in  $[0,2\pi]$, for all $u \in V$. Then there will be an edge between two vertices $u\in V_i,v\in V_j$ if and only if 
$\cos \theta_u \cos \theta_v + \sin \theta_u \sin \theta_v = \cos(\theta_u -\theta_v) \ge \beta_{i,j}$ or $\min\{|\theta_u -\theta_v|, 2\pi -|\theta_u - \theta_v|\} \le \arccos \beta_{i,j}$. This in turn, is equivalent to choosing a random variable $X_u$ uniformly distributed in $[0,1]$ for all $u \in V$, and there exists an edge between   two vertices $u\in V_i,v\in V_j$ if and only if 
$$
d_L(X_u,X_v) \equiv \min\{|X_u - X_v|, 1- |X_u- X_v|\} \le r_{i,j},
$$
where $r_{i,j} \in [0,\frac12], 0 \le i,j \le k$, are a set of real numbers. This corresponds to the Geometric block model in 1-dimension which will be referred to as ${\rm GBM_1}$.
 
For the rest of this paper, we concentrate on the case when $r_{i,i} = r_s$ for all $i \in \{1, \dots, k\}$, which we call the ``intra-cluster distance'' and $r_{i,j} = r_d$ for all $i,  j\in  \{1, \dots, k\}, i \ne j$, which we call the ``inter-cluster distance,'' to simply the analysis. To allow for edge density to be higher inside the clusters than across the clusters,  assume $r_s \geq r_d$.

\subsection{{Validation} of GBM}\label{sec:motivation}
 We next give two examples of datasets that {motivate} the geometric block model. In particular, this  datasets validate our hypothesis about geometric block model and the role of distance in the formation of edges. The first one is a product purchase metadata from Amazon, and the second one is a dataset with academic collaboration.

\subsubsection{Amazon Metadata}
\label{subsec:m2}
The first dataset that we use in our experiments is the Amazon product metadata on SNAP (\url{https://snap.stanford.edu/data/amazon-meta.html}), that has 548552 products and each product is one of the following types
\{Books, Music CD's, DVD's, Videos\}. Moreover, each product has a list of attributes, for example, a book may have attributes like $\langle$``General'', ``Sermon'', ``Preaching''$\rangle$. We consider the co-purchase network over these products. We make two observations here: (1) edges get formed (that is items are co-purchased) more frequently if they are similar, where we measure similarity by the number of common attributes between products, and (2) two products that share an edge have more common neighbors (number of items that are bought along with both those products) than two products with no edge in between. 

\begin{figure*}[htbp]
  \centering
  \begin{minipage}[b]{0.4\textwidth}
    \includegraphics[width=0.8\textwidth]{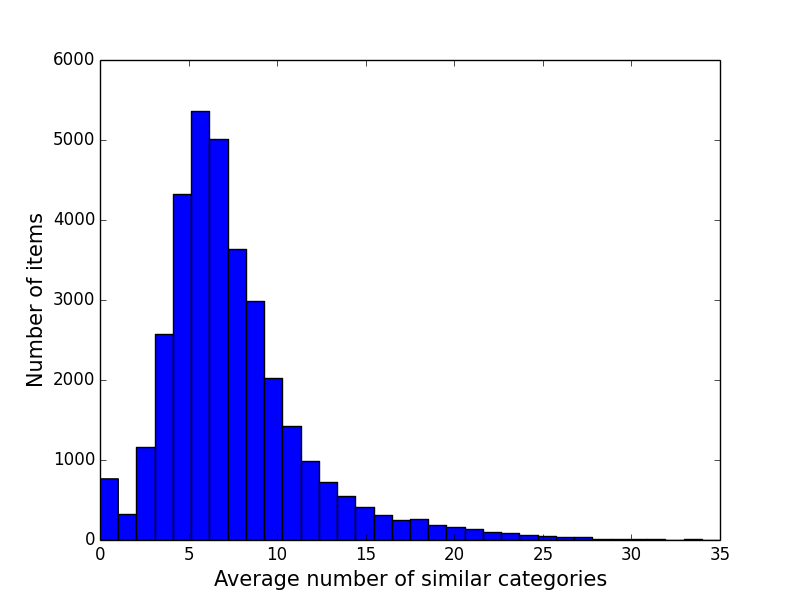}
    \caption{Histogram: similarity of products bought together (mean $\approx 6$)}
    \label{fig:amazon1}
  \end{minipage}
  \quad
  \begin{minipage}[b]{0.4\textwidth}
    \includegraphics[width=0.8\textwidth]{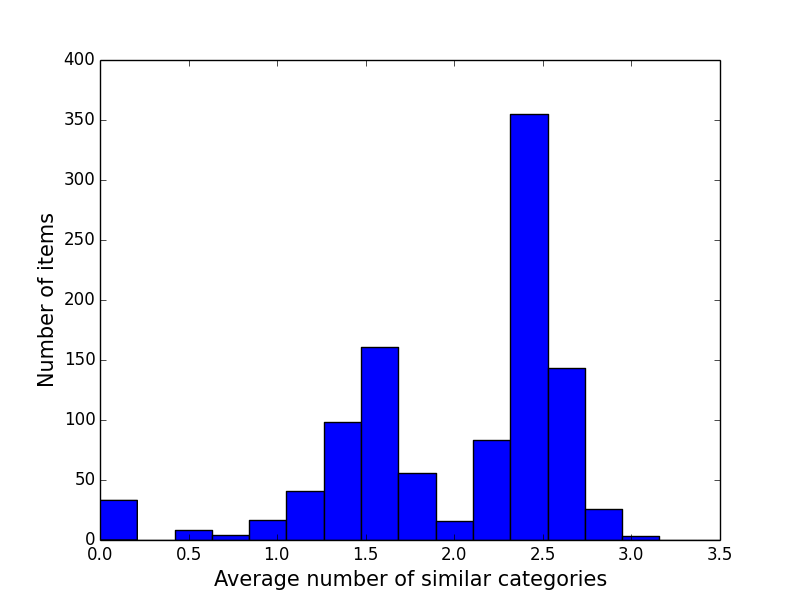}
    \caption{Histogram: similarity of products not bought together (mean$\approx 2$)}
    \label{fig:amazon2}
  \end{minipage}
\end{figure*}
\begin{figure*}[htbp]
  \centering
  \begin{minipage}[b]{0.28\textwidth}
    \includegraphics[width=\textwidth]{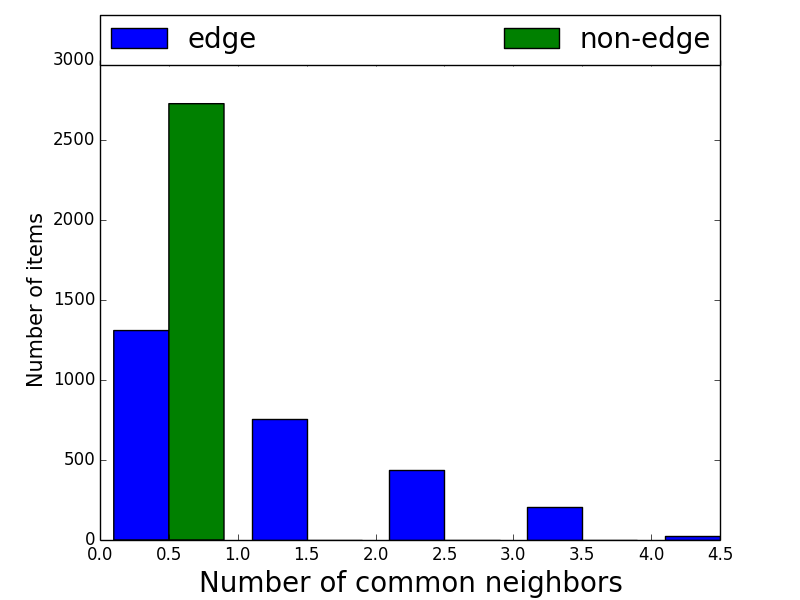}
  \end{minipage}
  \begin{minipage}[b]{0.3\textwidth}
    \includegraphics[width=\textwidth]{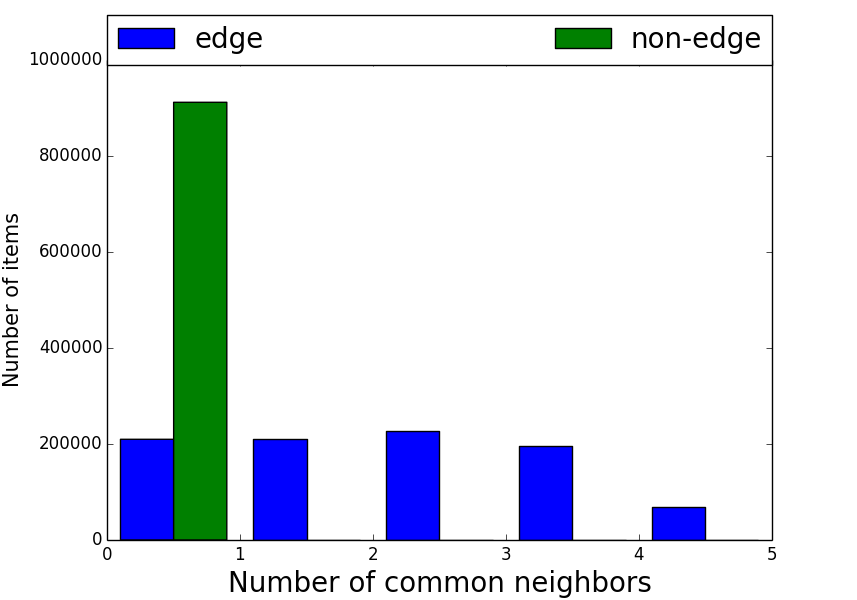}
  \end{minipage}
  \begin{minipage}[b]{0.28\textwidth}
    \includegraphics[width=\textwidth]{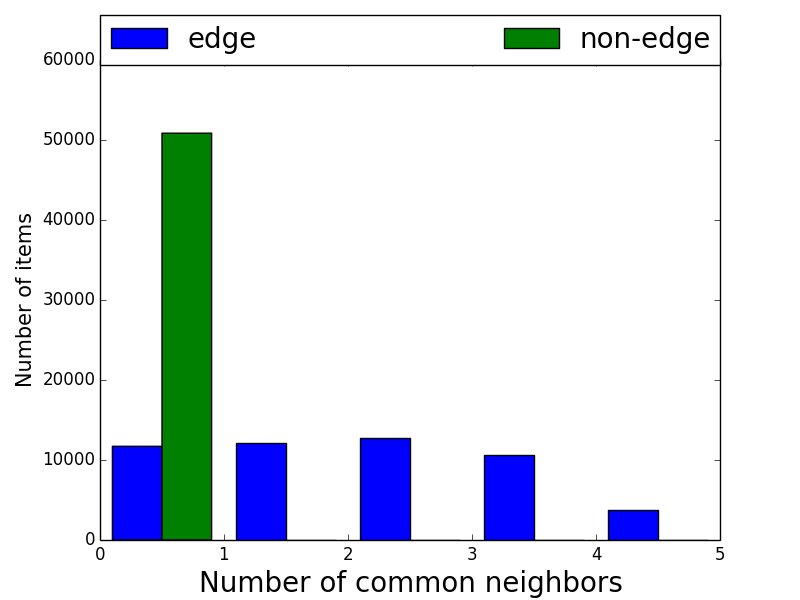}
  \end{minipage}
   \caption{Histogram of common neighbors of edges and non-edges in the co-purchase network, from left to right: Book-DVD, Book-Book, DVD-DVD}
    \label{fig:amazon3}
\end{figure*}

Figures~\ref{fig:amazon1} and \ref{fig:amazon2} show respectively average similarity of products that were bought together, and not bought together. From the distribution, it is quite evident that edges in a co-purchase network gets formed according to distance, a salient feature of random geometric graphs, and the GBM.

We next take equal number of product pairs inside Book (also inside DVD, and across Book and DVD) that have an edge in-between and do not have an edge respectively. Figure~\ref{fig:amazon3} shows that the number of common neighbors when two products share an edge is much higher than when they do not. In fact, almost all product pairs that do not have an edge in between also do not share any common neighbor. 
\latest{This can be explainable in sparse graphs if we imagine the latent features to follow the rules of GBM.
On the other hand, this also suggests that SBM is not a good model for this network, as in sparse SBM, two nodes having common neighbors is extremely rare, irrespective of whether they share an edge or not. Note that,  in practice, additional covariate features like product metadata, reviews etc. are also observed, which could be used for community detection.}

\subsubsection{Academic Collaboration}
 \label{subsec:m1}
We consider the collaboration network of academicians in Computer Science since 2016 (data obtained from \texttt{csrankings.org})\footnote{We use author affiliations from the paper available at \url{https://www.aminer.org/citation}}. According to area of expertise of the authors, we 
consider five different communities: Data Management (MOD), Machine Learning and Data Mining (ML), Artificial Intelligence (AI), Robotics (ROB), Architecture (ARCH). If two authors share the same affiliation, or shared affiliation in the past, we assume that they are geographically close.
We would like to hypothesize that, two authors in the same communities might collaborate even when they are geographically far. However, two authors in different communities are more likely to collaborate only if they share the same affiliation (or are geographically close).
Table \ref{table:collab} describes the number of edges across  the communities. 
 It is evident that the authors from same community are likely to collaborate irrespective of the affiliations and the authors of different communities collaborate much frequently when they share affiliations or are close geographically. This clearly indicates that the inter cluster edges are likely to form if the distance between the nodes is quite small, motivating the fact $r_d < r_s$ in the GBM.  

\begin{table*}[htbp]
\centering
\begin{tabular}{|p{1.5cm}|p{1.5cm}|p{1.5cm}|p{1.5cm}|} 
 \hline
 Area 1 & Area 2  & same  & different \\ 
 \hline
ML & AI  & 65& 43 \\
\hline
AI & ROB & 21 &8\\
\hline
ML & MOD & 22&14\\
\hline
ARCH & MOD & 7& 1\\
\hline
ROB & ARCH  &5& 1\\
\hline
\end{tabular}
\quad
\begin{tabular}{|p{1.5cm}|p{1.5cm}|p{1.5cm}|} 
 \hline
 Area  &  same & different   \\ 
 \hline
MOD & 88&91 \\
\hline
ARCH & 51 & 45\\
\hline
ROB & 63 &30\\
\hline
AI & 194& 110\\
\hline
ML &152 &141\\
\hline
\end{tabular}
\caption{On the left we count the number of inter-cluster edges when authors shared same affiliation and different affiliations. On the right, we count the same for intra-cluster edges.\label{table:collab}}
\end{table*}


}

\section{Description of Our Results}
In Section \ref{sec:rgg_intro}, we describe a new random graph model that we call the {\em Random Annulus Graphs} (RAG) that we study in this paper. The RAG is a variant of the RGG  and interestingly, the connectivity properties of the RAG will allow us to design algorithms with provable guarantees for community detection in the GBM. Subsequently, in section \ref{sec:gbm_intro}, we introduce the Geometric Block Model more formally. In section \ref{sec:cluster_recovery}, we provide a brief discussion on the algorithms that we design for recovery of clusters in a random graph sampled according to GBM.
\subsection{Random graph models}\label{sec:rgg_intro}
Models of random graphs are ubiquitous with Erd\H{o}s-R\'{en}yi graphs \citep{erdos1959random,gilbert1959random} at the forefront. Studies of the properties of random graphs have led to many fundamental theoretical observations as well as many engineering applications.  In an Erd\H{o}s-R\'{en}yi graph $G(n,p), n \in \integers_+, p \in [0,1]$, the randomness lies in how the edges are chosen: each possible pair of vertices forms an edge independently with probability $p$. It is also possible to consider models of graphs where randomness lies in the vertices. 

Keeping up with the simplicity of the  Erd\H{o}s-R\'{en}yi model, let us define a random annulus graph in one dimension (${\rm RAG_1}$) in the following way.
\begin{definition}[Random annulus graph]\label{defn:high}
Let the $t$-dimensional unit sphere be denoted by $S^t\equiv \{x \in \reals^{t+1} \mid \|x\|_2=1\}$ where $\|\cdot\|_2$ denote the $\ell_2$ norm.
For $t\ge 1$, a $t$-dimensional random annulus graph ${\rm RAG}_t(n,[r_1,r_2])$ on $n$ vertices has parameters $n,t \in \integers_+$,  and $r_1, r_2 \in [0,2], r_1 \le r_2$. It is defined by assigning $X_i \in S^t$ to vertex $i, 1 \le i \le n,$ where $X_i$'s are independent and identical random vectors uniformly distributed in $S^t$. There will be an edge between vertices $i$ and $j, i \ne j,$ if and only if $r_1 \le \|X_i-X_j\|_2  \le r_2$. 
\end{definition}
When from the context if it is clear that we are in high dimensions ($t >1$), we use $d(u,v)$ to denote $\|X_u-X_v\|_2$ or just the $\ell_2$ distance between the arguments.


 We give the name random annulus graph (RAG) because two vertices are connected if one is within an `annulus' centered at the other.
 One can think of the random variables $X_i, 1\le i \le n$, to be uniformly distributed on the perimeter of a circle with radius $\frac1{2\pi}$ and the distance $d_L(\cdot,\cdot)$ to be the geodesic distance (the length of the smaller arc between the two points).
For clarity in the calculations, it will be helpful to consider the vertices as just random points on $[0,1]$. Note that every point has a natural left direction (if we think of them as points on a circle then this is the counterclockwise direction) and a right direction.

This definition is by no means new. For the case of $r_1=0, t=1$, this is the random geometric graphs (RGG) in one dimension. 
Random Geometric graphs were defined first  by \citep{gilbert1961random} and constitute the first and simplest model of spatial networks. The definition of ${\rm RAG_1}$ has been previously mentioned in \citep{dettmann2016random}. The interval $[r_1,r_2]$ is called the connectivity interval in ${\rm RAG_1}$.
Random geometric graphs have several desirable properties that model real human social networks, such as vertices with the degree associativity property (high degree nodes tend to connect). This has led RGGs to be used as models of disease outbreak in social network \citep{eubank2004modelling} and flow of opinions \citep{zhang2014opinion}. RGGs are an excellent model for wireless (ad-hoc) communication networks \citep{dettmann2016random,haenggi2009stochastic}. From a more mathematical stand-point, RGGs act as a bridge between the theory of classical random graphs and that of percolation \citep{b:01,b:06}. Recent works on RGGs also include hypothesis testing between an Erd\H{o}s-R\'{en}yi graph and a random geometric graph \citep{bubecktriangle}.

Threshold properties of   Erd\H{o}s-R\'{en}yi graphs have been at the center of much theoretical interest, and in particular it is known  that many graph properties exhibit sharp phase transition phenomena \citep{friedgut1996every}. Random geometric graphs also exhibit similar threshold properties  \citep{penrose2003random}. 

\paragraph{Random Annulus Graphs in one dimension}  
We describe the case of $t=1$ separately from $t>1$ because in the former case we have exact connectivity result.
The base of the logarithm is $e$ here and everywhere else in the paper unless otherwise mentioned. 


\begin{theorem}[Connectivity threshold of one dimensional random annulus graphs]\label{thm:rag}
The ${\rm RAG}_1(n,[2\sin{\frac{\pi b\log n}{n}},2\sin{\frac{\pi a\log n}{n}}])$  is \latest{ i) connected with  probability $1-o(1)$ if $a >1$ and $a - b >0.5$ ii) disconnected with  probability $1-o(1)$ if $a<1$ or $a - b < 0.5$  }
\end{theorem}

For the 1-dimensional case, the definition of random annulus graphs can be  simplified, and that simplification explains the parameterization in the above theorem. The ${\rm RAG}_1(n,[2\sin{\frac{\pi b\log n}{n}},2\sin{\frac{\pi a\log n}{n}}])$ is equivalent  to ${\rm RAG}^\ast_1(n,[\frac{b\log n}{n},\frac{a\log n}{n}])$ (i.e., the two random graphs follow the same distribution), where the later graph is defined below.

\begin{definition}\label{defn:one}
A random annulus graph ${\rm RAG}^\ast_1(n,[r_1,r_2])$ on $n$ vertices has parameters $n$,  and  a pair of  real numbers $r_1, r_2 \in [0,1/2], r_1 \le r_2$. It is defined by assigning a number $X_i \in \reals$ to vertex $i, 1 \le i \le n,$ where $X_i$s are independent and identical random variables uniformly distributed in $[0,1]$. There will be an edge between vertices $i$ and $j, i \ne j,$ if and only if $r_1 \le d_L(X_i,X_j)  \le r_2$ where $d_L(X_i,X_j) \equiv \min\{|X_i - X_j|, 1 - |X_i - X_j|\} $. 
\end{definition}
For the 1-dimensional case, for any two vertices $u,v$,  let $d(u,v)$ denote $d_L(X_u,X_v)$ where $X_u,X_v$ are corresponding random values to the vertices respectively. We can extend this notion to denote the distance $d(u,v)$ between a vertex $u$ (or the embedding of that vertex in $[0,1]$) and a number $v \in [0,1]$ naturally. 

\begin{remark}
For the analysis and proofs of the $1$-dimensional case, we will use the definition of ${\rm RAG}_1^{\ast}$ (instead of ${\rm RAG}_1$). However, due to the equivalence between the two random graph models, our results for ${\rm RAG}_1^{\ast}$ also hold for ${\rm RAG}_1$ with appropriately modified parameters. This is because we can think of the vertices in ${\rm RAG}^{\ast}_1(n,[r_1,r_2])$ being distributed randomly on the circumference of circle of radius $1/2\pi$ and distance between two vertices in ${\rm RAG}^{\ast}_1(n,[r_1,r_2])$ is geodesic (measured along the circumference of the circle). It is  easier  to reason about the geodesic distance over Euclidean distance in $1$-dimension; however, the advantage  is no longer there in higher dimensions. Since handling geodesic distances is more cumbersome in the higher dimensions, we resorted to Euclidean distance for $t>1$ in order to retain simplicity.
\end{remark}

Consider a ${\rm RAG}^\ast_1(n,[0,r])$ defined above with $r = \frac{a\log n}{n}$. It is known that ${\rm RAG}^\ast_1(n,[0,r])$ is connected with high probability if and only if $a > 1$\footnote{That is, ${\rm RAG}^\ast_1(n,[0,\frac{(1+\epsilon)\log n}{n}])$ is connected for any $\epsilon >0$. We will ignore this $\epsilon$ and just mention connectivity threshold as $\frac{\log{n}}{n}$.}~(\citep{muthukrishnan2005bin,penrose2003random} See also \citep{penrose2016}). Now let us consider the graph ${\rm RAG}^\ast_1(n,[\frac{\delta\log n}{n},\frac{\log n}{n}]), \delta>0$. Clearly this graph has less edges than ${\rm RAG}^\ast_1(n,[0,\frac{\log n}{n}])$.  {\bf Is this graph still connected?}
Surprisingly, we show that the above modified graph remains connected as long as $\delta \le 0.5$. Note that, on the other hand,  ${\rm RAG}^\ast_1(n,[0,\frac{(1-\epsilon)\log n}{n}])$ is not connected for any $\epsilon >0$.




This means the graphs ${\rm RAG^\ast}_1(n, [0, \frac{0.99 \log n}{n}])$ and ${\rm RAG}^\ast_1(n, [\frac{0.49 \log n}{n}, \frac{0.99 \log n}{n}])$ are not connected with high probability, whereas ${\rm RAG}^\ast_1(n, [\frac{0.50 \log n}{n}, \frac{\log n}{n}])$ is connected. 
For a depiction of the connectivity regime for the random annulus graph ${\rm RAG}^\ast_1(n,[\frac{b\log n}{n},\frac{a\log n}{n}])$ see Figure~\ref{fig:region}.
\vspace{-0.1in}
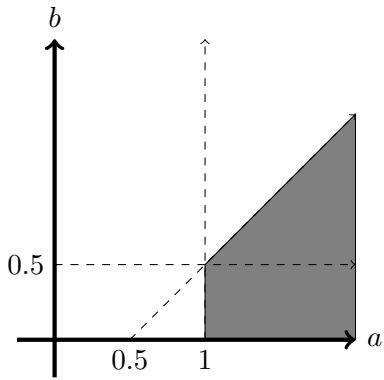
\begin{figure*}[htbp]
\centering
\begin{tikzpicture}

\draw[fill=gray]  (2,1) -- (4,3) -- (4,0) -- (2,0) -- cycle;
\draw[->,ultra thick] (-0.5,0)--(4,0) node[right]{$a$};
\draw[->,ultra thick] (0,-0.5)--(0,4) node[above]{$b$};

\draw[->,dashed] (0,1)--(4,1); 
\draw[->,dashed] (2,0)--(2,4); 
\draw[->,dashed] (1,0)--(4,3); 

\node[below] at (2,0) {$1$};
\node[below] at (1,0) {$0.5$};
\node[left] at (0,1) {$0.5$};

\draw[fill=gray,transparent]  (2,1) -- (4,3) -- (4,0) -- (2,0) -- cycle;
\end{tikzpicture}
\caption{The shaded area in  the $a$-$b$ plot shows the regime where an ${\rm RAG}_1^\ast(n,[\frac{b \log n}{n}, \frac{a\log n}{n}])$ is connected with high probability. \label{fig:region}}
\end{figure*}

Can we explain this seemingly curious shift in connectivity interval, when one goes from $b=0$ to $b >0$? Compare the ${\rm RAG}_1^\ast(n, [\frac{0.50 \log n}{n}, \frac{\log n}{n}])$ with the ${\rm RAG}^\ast_1(n, [0, \frac{\log n}{n}])$. The former one can  be thought of being obtained by deleting all the `short-distance' edges from the latter. It turns out  the `long-distance' edges are  sufficient to maintain connectivity, because they can connect points over multiple hops in the graph. Another possible explanation is that connectivity threshold for ${\rm RAG}^\ast_1$ is not dictated by isolated nodes as is the case in Erd\H{o}s-R\'{en}yi graphs. Thus, after the connectivity threshold has been achieved, removing certain short edges still retains connectivity. 

\paragraph{Random Annulus Graphs in Higher Dimension} It is natural to ask similar question of connectivity for random annulus graphs in higher dimension. In a random annulus graph at dimension $t$, we may assign $t$-dimensional random vectors to each of the vertices, and use a standard metric such as the Euclidean distance to decide whether there should be an edge between two vertices. 

The ${\rm RAG}_t(n, [0,r])$ gives the standard definition of random geometric graphs in $t$ dimensions (for example, see \citep{bubecktriangle} or \citep{penrose2003random}).
Our main result here gives a condition that guarantees connectivity.

\begin{theorem}\label{thm:highdem1}
Let $t>1$ and
$$
\psi(t)\equiv\frac{\sqrt{\pi}(t+1)\Gamma(\frac{t+2}{2})}{\Gamma(\frac{t+3}{2})},$$ where $\Gamma(x) = \int_0^\infty y^{x-1}e^{-y}dy$ is the gamma function.
If 
$(a/2)^t-b^t \ge {8(t+1)\psi(t)}\text{  and  }  a>2b$, then a $t-$dimensional random annulus graph ${\rm RAG}_t(n,[b\Big(\frac{\log n}{n}\Big)^{1/t},a \Big(\frac{ \log n}{n}\Big)^{1/t}])$ is connected with  probability $1-o(1)$.
\end{theorem}

Computing the connectivity threshold of RAG in high dimensions exactly is highly challenging, and we have to use several approximations of high dimensional geometry. Our arguments crucially rely on VC dimensions of sets of geometric objects such as intersections of high dimensional annuluses and hyperplanes. 

In the process, we also derived the following results about the existence of isolated vertices in random annulus graphs.


\begin{theorem}[Zero-One law for Isolated Vertices in ${\rm RAG_t}$]\label{th:lb}
Let $t>1$. For a $t$-dimensional random annulus graph ${\rm RAG}_t(n,[r_1,r_2])$ where $r_2=a \Big(\frac{ \log n}{n}\Big)^{\frac{1}{t}}$ and $r_1=b\Big(\frac{\log n}{n}\Big)^{\frac{1}{t}}$, \latest{i) there exists isolated nodes with  probability $1-o(1)$ if 
$a^t -b^t < \psi(t)$, and ii) the graph is connected with probability $1-o(1)$ if 
$a^t -b^t > \psi(t)$.}
\end{theorem}

An obvious deduction from this theorem is that an ${\rm RAG}_t(n,[b\Big(\frac{\log n}{n}\Big)^{\frac{1}{t}},a \Big(\frac{ \log n}{n}\Big)^{\frac{1}{t}}])$ is not connected with probability $1-o(1)$ if $a^t -b^t < \psi(t)$.

All these connectivity results find immediate application in analyzing the algorithm that we propose for the geometric block model (GBM). 
A GBM is a generative model for networks (graphs) with underlying community structure.





\subsection{Geometric Block Model}\label{sec:gbm_intro}


Since for random geometric graphs in 1-dimensional case the distance $d_L(\cdot,\cdot)$ defined in Def.~\ref{defn:one} gives a simpler expression, we use this for the 1-dimensional case. For higher dimensions we resort back to the Euclidean distance.

 
\begin{definition}[Geometric Block Model in 1-dimension]
Given $V = V_1\sqcup V_2, |V_1|=|V_2| = \frac{n}2$,  choose a random variable $X_u$ uniformly distributed in $[0,1]$ for all $u \in V$.
The geometric block model  ${\rm GBM_1}(r_s, r_d)$ with parameters $r_s> r_d$ is a random graph where an edge exists between vertices $u$ and $v$  if and only if,
\begin{align*}
d_L(X_u, X_v) \le r_s & \text{ when } u, v \in V_1 \text{ or } u,v \in V_2\\
 d_L(X_u, X_v) \le r_d & \text{ when } u \in V_1, v \in V_2 \text{ or } u\in V_2,v \in V_1.
\end{align*}
 \end{definition}

As a consequence of the connectivity lower bound on  ${\rm RAG}_1^\ast$, we are able to show that  recovery of the partition in  ${\rm GBM_1}(\frac{a\log n}{n}, \frac{b \log n}{n})$ is not possible with high probability  by any means whenever $a - b <0.5$ or $a<1$ (see, Theorem~\ref{thm:rag}). Another consequence of the random annulus graph results is that we show that if in  addition to a ${\rm GBM_1}$ graph, all the locations of the vertices are also provided, then recovery is possible if and only if $a - b >0.5$ or $a>1$ (formal statement in Theorem~\ref{thm:gbmplus}).

Coming back to the actual recovery problem, our main contribution for ${\rm GBM_1}$ is to provide a simple and efficient algorithm that performs  well in the sparse regime (see, Algorithm~\ref{alg:alg1}). 

\begin{theorem}[Recovery algorithm for ${\rm GBM_1}$]\label{gbm:upper}
Suppose we have the graph $G(V,E)$ generated according to ${\rm GBM_1}(r_s \equiv \frac{a\log n}{n},r_d\equiv \frac{b\log n}{n}), a \ge 2b$. 
Define
\begin{align*}
f_1&=\min\{f: (2b+f)\log \frac{2b+f}{2b}-f > 1\},~~~
f_2=\min\{f : (2b-f)\log \frac{2b-f}{2b}+f > 1\}\\
\theta_1 &= \max\{\theta:\frac{1}{2}\Big((4b+2f_1)\log \frac{4b+2f_1}{2a-\theta}+2a-\theta-4b-2f_1\Big) > 1  \text{ and } 0 \le \theta \le 2a-4b-2f_1\}\\
\theta_2 &= \min\{\theta: \frac{1}{2}\Big((4b-2f_2\log \frac{4b-2f_2}{2a-\theta}+2a-\theta-4b+2f_2\Big) > 1
 \text{ and }  
a \ge \theta \ge \max\{2b,2a-4b+2f_2\}\}.
\end{align*}
Then, if   $a-\theta_2+\theta_1>2$ or $a> \max(1+\theta_2,2)$, there exists an efficient algorithm which will recover the correct partition in $G$ with  probability $1-o(1)$. Moreover, if $a-b < 0.5$ or $a < 1$, any algorithm to recover the partition in ${\rm GBM_1}(\frac{a \log n}{n},\frac{b \log n}{n})$ will give incorrect output with probability $1-o(1)$ .
\end{theorem}

Some examples of the parameters when the proposed algorithm (Algorithm~\ref{alg:alg1}) can successfully recover are given  in Table~\ref{tab:per} \latest{and Figure~\ref{fig:comparison_ab} compares them with the lower bound.}

 \begin{figure}
     \centering
     \includegraphics[scale=0.5]{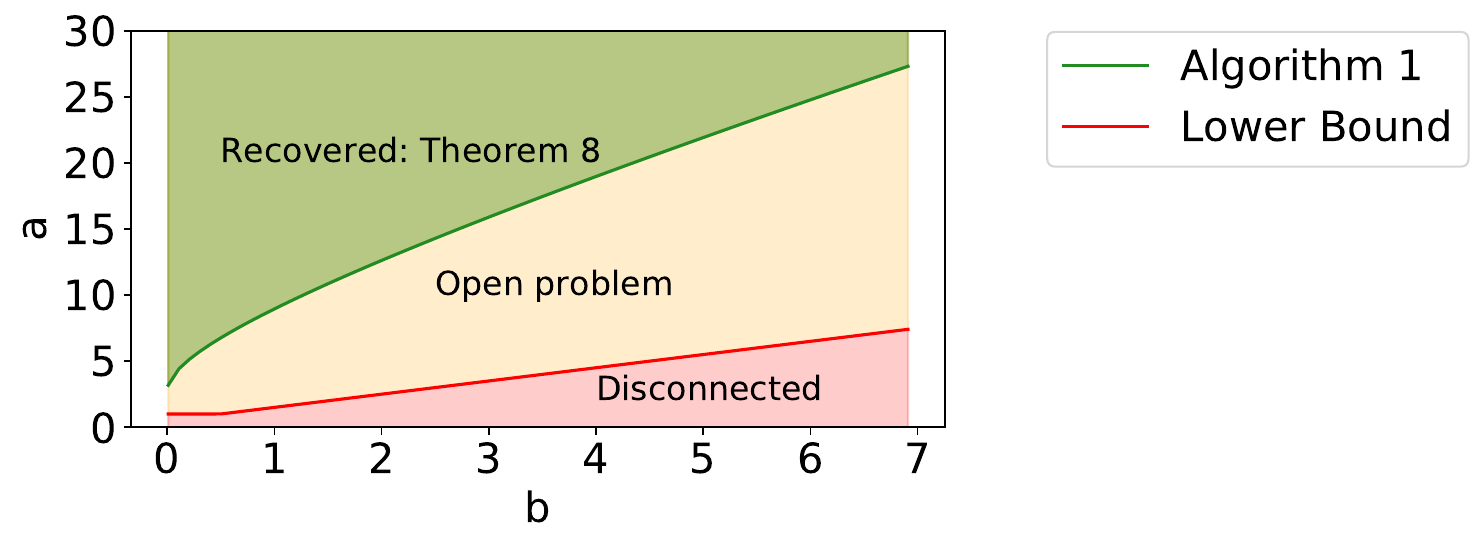}
     \caption{\latest{Comparison of $a$ and $b$ in ${\rm GBM_1}(\frac{a\log n}{n},\frac{b\log n}{n})$. Red region denotes the values where the graph is disconnected, green region is recoverable by Algorithm~\ref{alg:alg1} and yellow region denotes values of $a$ and $b$ for which there is no algorithm to recover the clusters.}}
     \label{fig:comparison_ab}
 \end{figure}

\begin{table}[h!]\label{tab:per}
\centering
\begin{tabular}{ |c|c|c|c|c|c|c|c|c| } 
 \hline
$b$ & 0.01 & 1 & 2 & 3 & 4 & 5 & 6 & 7 \\ 
\hline
Minimum value of $a$ & 3.18 & 8.96 & 12.63 & 15.9 & 18.98 & 21.93 & 24.78 & 27.57 \\
 \hline 
\end{tabular}
\caption{Minimum value of $a$, given $b$ for which Algorithm~\ref{alg:alg1} resolves clusters correctly in ${\rm GBM_1}(\frac{a\log n}{n},\frac{b\log n}{n})$.\label{tab:per}}
\end{table}

As can be anticipated, the connectivity results for RAG applies to the `high dimensional' geometric block model. 
 In many applications, the latent feature space of nodes  are high-dimensional. For example, road networks are two-dimensional whereas the number of features used in a social network may have much higher dimensions.  
 \begin{definition}[The GBM in High Dimensions] Let $t>1$.
Given $V = V_1\sqcup V_2, |V_1|=|V_2| = \frac{n}2$,  choose a random vector $X_u$ independently uniformly distributed in $S^t \subset \reals^{t+1}$ for all $u \in V$.
The geometric block model  ${\rm GBM}_t(r_s, r_d)$ with parameters $r_s> r_d$ is a random graph where an edge exists between vertices $u$ and $v$  if and only if,
\begin{align*}
\norm{X_u-X_v}_2 \le r_s & \text{ when } u, v \in V_1 \text{ or } u,v \in V_2\\
\norm{X_u-X_v}_2 \le r_d & \text{ when } u \in V_1, v \in V_2 \text{ or } u\in V_2,v \in V_1.
\end{align*}
\end{definition}
We extend the algorithmic results to high dimensions (details in Section~\ref{sec:sparse-high}).
\begin{theorem}
\label{theorem:intro-1}
Let $t>1$. If $r_s=\Theta((\frac{\log{n}}{n})^{\frac{1}{t}})$ and \latest{$r_s-r_d=\Theta( (\frac{\log n}{n})^{\frac{1}{t}})$}, there exists a polynomial time efficient algorithm that recovers the partition from ${\rm GBM}_t(r_s, r_d)$ with probability $1-o(1)$ . Moreover, if $r_s-r_d=o( (\frac{\log n}{n})^{\frac{1}{t}})$ or $r_s=o((\frac{\log{n}}{n})^{\frac{1}{t}})$, any  algorithm fails to recover the partition with probability at least $1/2$.
\end{theorem}

In the following section, we present the high level ideas of our algorithms to recover the clusters with sub-quadratic running time.

\subsection{Clustering Algorithm}\label{sec:cluster_recovery}

As we have mentioned already, it is well known that triangle based heuristics have good performance on real world graphs for cluster recovery. It turns out that for random graphs generated according to the GBM, simple triangle based algorithms are order optimal. A very simple and intuitive algorithm to cluster the vertices of a graph is as follows: 
\begin{enumerate}
    \item For a particular edge say $(u,v)$, count the number of triangles that contain the edge $(u,v)$ and from this observed statistics, infer based on a threshold if the endpoints $u$ and $v$ belong to the same cluster or different cluster. 
    \item We perform this inference step on all edges of the graph and using the results, we can partition the nodes into their respective clusters. \latest{The analysis of the first step guarantees that our algorithm labels an edge as intra-cluster or inter-cluster with a probability of $1-O(\frac{1}{n\log^2 n})$. On applying union bound, it guarantees that all edges are labelled correctly with probability $1-O(\frac{1}{\log n})$.}
\end{enumerate}
This algorithm was analyzed in a preliminary conference version of this paper \citep{galhotra2017geometric} and although the algorithm was order optimal, the guarantees for cluster recovery in the sparse logarithmic regime i.e. $r_s,r_d = \Theta(\frac{\log n}{n})$ hold only when $r_s \ge 4r_d$. When $r_s \le 4r_d$, it is not possible to infer if the nodes forming an edge belong to the same cluster or not, based on a single threshold on the number of triangles the edge is part of. Therefore,  we proceed to devise an improved algorithm whose main ingredient is still based on triangle counting but uses two thresholds on the triangle count instead of one. At a high level, if the number of triangles formed by an edge is between these two thresholds, we delete the edge because it is not possible to infer correctly whether the endpoints belong to the same or different cluster. We then show that all the remaining edges belong to same cluster and identify the condition when the surviving edges suffice to recover the clusters correctly. The connectivity of surviving edges is shown with the help of corresponding results for the Random Annulus graphs. A  preliminary version of these results were published in \citep{galhotra2019connectivity}.  
\latest{
\begin{remark}
    When the degree of a graph is logarithmic in the number of vertices (i.e., about $n \log n$ edges), a stochastic block model implies the existence of about $\log^3 n$ triangles, whereas in GBM there will be about $n \log^2 n$ triangles. As a result, in that regime, most vertices in SBM are not part of any triangles, and the above triangle counting algorithm will fail. However, in the same regime, the algorithm provably recover the communities in GBM, and also in practice. 
\end{remark}
}

It is possible to generalize the GBM to include different distributions and different metric spaces. It is also possible to construct other type of spatial block models such as the one  put forward in a parallel work \citep{sankararaman2018community} which rely on the random dot product graphs \citep{young2007random}. In \citep{sankararaman2018community}, edges are drawn between vertices randomly and independently as a function  of the distance between the corresponding vertex random variables.  \citep{sankararaman2018community} also considers the recovery scenario where in addition to the graph, values of the vertex random variables are provided. In GBM, we only observe the graph.

The rest of the paper is organized as follows. First, 
in Section~\ref{sec:rag}, the sharp connectivity phase transition results for random annulus graphs are proven (details in Appendix~\ref{sec:VRG-detail}). In Section~\ref{sec:hrag}, the connectivity results are proven for high dimensional random annulus graphs (details in Appendix~\ref{sec:hrag-detail}). In Section~\ref{sec:gbm}, a lower bound for the geometric block model as well as the main recovery algorithm are presented. In Section~\ref{sec:sparse-high}, we present the extension of our recovery algorithm to high dimensional GBM. Finally, Section~\ref{sec:exp} empirically evaluates our techniques on various real world and synthetic datasets.

\section{Connectivity of random annulus Graphs}
\label{sec:rag}
In this section we give a sketch of the proof of  sufficient condition for connectivity of ${\rm RAG}^{\ast}_1$ (as part of proving Theorem~\ref{thm:rag}) since connectivity guarantees in ${\rm RAG}_1^{\ast}$ can be mapped to connectivity properties in ${\rm RAG}_1$ (recall that the ${\rm RAG}_1(n,[2\sin{\frac{\pi b\log n}{n}},2\sin{\frac{\pi a\log n}{n}}])$ is equivalent to ${\rm RAG}^\ast_1(n,[\frac{b\log n}{n},\frac{a\log n}{n}])$). The full details along with the proof of the necessary condition in Theorem \ref{thm:rag} are given in Appendix~\ref{sec:VRG-detail}. 
\subsection{Sufficient condition for connectivity of ${\rm RAG}^\ast_1$}

\begin{theorem*}(Sufficient Condition in Theorem \ref{thm:rag} for ${\rm RAG}_1^{\ast}$)
If $a >1$ and $a - b >0.5$, the random annulus graph ${\rm RAG}^\ast_1(n,[\frac{b\log n}{n},\frac{a\log n}{n}])$  is connected with  probability $1-o(1)$ .
\end{theorem*}
To prove this theorem we use two main technical lemmas that show two different events happen with high probability simultaneously. 

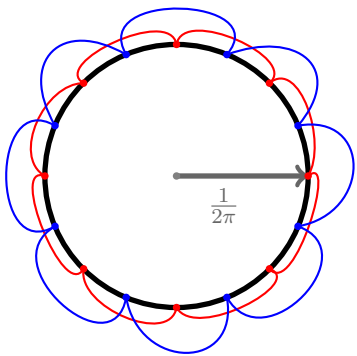
\begin{figure*}[htbp]
\centering
\begin{tikzpicture}[thick, scale=0.5]
 \filldraw[color=black, fill=red!0,  line width=2pt](0,0) circle (3.5);
 \draw [line width=2pt, color=black!60,->](0,0) -- (3.5,0)node[anchor=north east] at (2,0) {$\frac{1 }{2\pi}$};
 \filldraw [gray] (0,0) circle (2pt);
\filldraw [red] (0,3.5) circle (2pt);
\filldraw [red] (3.5,0) circle (2pt);
\filldraw [red] (0,-3.5) circle (2pt);
\filldraw [red] (-3.5,0) circle (2pt);
\filldraw [red] (2.47,2.47) circle (2pt);
\filldraw [red] (2.47,-2.47) circle (2pt);
\filldraw [red] (-2.47,2.47) circle (2pt);
\filldraw [red] (-2.47,-2.47) circle (2pt);
\draw[red]    (0,3.5) to[out=80,in=40] (2.47,2.47);
\draw[red]    (2.47,2.47) to[out=40,in=60](3.5,0) ;
\draw[red]    (3.5,0) to[out=40,in=-50](2.47,-2.47) ;
\draw[red]    (2.47,-2.47) to[out=-20,in=-60](0,-3.5) ;
\draw [red]   (0,-3.5) to[out=-100,in=-100](-2.47,-2.47) ;
\draw [red]   (-2.47,-2.47) to[out=-120,in=-200](-3.5,0) ;
\draw [red]   (-3.5,0) to[out=150,in=-180](-2.47,2.47) ;
\draw [red]   (-2.47,2.47) to[out=150,in=100](0,3.5) ;
\filldraw [blue] (3.23,1.339) circle (2pt);
\filldraw [blue] (3.23,-1.339) circle (2pt);
\filldraw [blue] (-3.23,1.339) circle (2pt);
\filldraw [blue] (-3.23,-1.339) circle (2pt);
\filldraw [blue] (1.339,3.23) circle (2pt);
\filldraw [blue] (1.339,-3.23) circle (2pt);
\filldraw [blue] (-1.339,3.23) circle (2pt);
\filldraw [blue] (-1.339,-3.23) circle (2pt);
\draw[blue]   (3.23,1.339) to[out=10,in=10, , distance=2cm](3.23,-1.339);
\draw[blue]    (3.23,-1.339) to[out=-20,in=-60, distance=2cm](1.339,-3.23) ;
\draw[blue]    (1.339,-3.23) to[out=-80,in=-80,distance=2cm](-1.339,-3.23) ;
\draw[blue]    (-1.339,-3.23) to[out=-100,in=-150,distance=2cm](-3.23,-1.339) ;
\draw[blue]    (-3.23,-1.339) to[out=-150,in=150,distance=2cm](-3.23,1.339) ;
\draw[blue]    (-3.23,1.339) to[out=120,in=150,distance=2cm](-1.339,3.23) ;
\draw[blue]    (-1.339,3.23) to[out=120,in=50,distance=2cm](1.339,3.23) ;
\draw[blue]    (1.339,3.23) to[out=20,in=50,distance=2cm](3.23,1.339) ;
\end{tikzpicture}
\caption{Each vertex having two neighbors on either direction implies the graph is a union of cycles. The cycles can be interleaving in $[0,1]$.\label{fig:arr0}}
\vspace{5pt}
\end{figure*}

\begin{lemma}\label{lem:lemma1}
A set of vertices $\cC \subseteq V$ is called a cover of $[0,1]$, if for any point $y$ in $[0,1]$ there exists a vertex $v\in \cC$ such that $d(v,y) \le \frac{a\log n}{2n}$.  If $a-b >0.5$ and $a >1$,  the random annulus graph ${\rm RAG}^\ast_1(n, [\frac{b\log n}{n}, \frac{a\log n}{n}])$ is a union of cycles  
such that every cycle forms a cover of $[0,1]$  (see Figure~\ref{fig:arr0}) with  probability $1-o(1)$.
\end{lemma}
This lemma also shows effectively the fact that `long-edges' are able to connect vertices over multiple hops.
Note that, the statement of Lemma~\ref{lem:lemma1} would be easier to prove if the condition were $a-b>1$. 
In that case what we prove is that every vertex  has neighbors (in the ${\rm RAG}_1^\ast$) on both of the left and right directions. To  see this
for each vertex $u$ , assign two indicator $\{0,1\}$-random variables $A_{u}^{l}$ and $A_{u}^{r}$, with $A_{u}^{l}=1$ if and only if there is no node $x$ to the left of node $u$ such that $d(u,x) \in [\frac{b\log n}{n},\frac{a\log n}{n}]$. Similarly, let $A_{u}^{r}=1$ if and only if there is no node $x$ to the right of node $u$ such that $d(u,x)\in [\frac{b\log n}{n},\frac{a\log n}{n}]$. Now define $A=\sum_{u}( A_{u}^{l}+A_{u}^{r})$.  We have, 
$$\Pr(A_{u}^{l}=1)=\Pr(A_{u}^{r}=1)=(1-\frac{(a-b)\log n}{n})^{n-1},$$ 
and, 
$$\avg[A]=2n(1-\frac{(a-b)\log n}{n})^{n-1} \le 2n^{1-(a-b)}.$$
 If $a-b >1$ then $\avg[A]=o(1)$ which implies, by invoking Markov inequality, that with high probability every node will have neighbors (connected by an edge in the ${\rm RAG}_1^\ast$) on either side. This results in the interesting conclusion that every vertex will lie in  a cycle that covers  $[0,1]$. This is true for every vertex, hence the graph is simply a union of cycles each of which is a cover of $[0,1]$.
The main technical challenge is to show that this conclusion remains valid even when $a-b >0.5$, which is proved in Lemma~\ref{lem:lemma1} in Appendix~\ref{sec:VRG-detail}. 

\begin{lemma}\label{lem:lemma2}
Set two real numbers $k\equiv \lceil b/(a-b)\rceil+1$ and $\epsilon < \frac1{2k}$. In an ${\rm RAG}^\ast_1(n, [\frac{b\log n}{n}, \frac{a\log n}{n}]), 0 <b <a$, with  probability $1-o(1)$ there exists a vertex $u_0$ and  $k$ nodes $\{u_1, u_2 ,\ldots, u_k\}$ to the right of $u_0$ such that $d(u_0,u_i) \in [\frac{(i(a-b)-2i\epsilon)\log n}{n},\frac{(i(a-b)-(2i-1)\epsilon) \log n}{n}]$ and  $k$ nodes $\{v_1, v_2 ,\ldots, v_k\}$ to the right of $u_0$ such that $d(u_0,v_i) \in [\frac{((i(a-b)+b-(2i-1)\epsilon)\log n}{n},\frac{(i(a-b)+b -(2i-2)\epsilon)\log n}{n}]$, for $i =1,2,\ldots,k$. The arrangement of the vertices is shown  in Figure~\ref{fig:arr} (pg. 18).
\end{lemma}

With the help of these two lemmas, we are in a position to  prove the sufficient condition in Theorem~\ref{thm:rag}. The proofs of  the two lemmas are given in Appendix~\ref{sec:VRG-detail} and contain the technical essence of this section.

\begin{proof}[Proof of Sufficient condition in Theorem~\ref{thm:rag}]
We have shown that the two events mentioned in Lemmas \ref{lem:lemma1} and \ref{lem:lemma2} happen with high probability. Therefore they simultaneously happen under the condition $a>1$ and $a-b >0.5$.
Now we will show that these events together imply that the graph is connected. To see this, consider the vertices $u_0,\{u_1, u_2, \ldots, u_k\}$ and $\{v_1, v_2, \ldots, v_k\}$ that satisfy the conditions of Lemma \ref{lem:lemma2}. We can observe that each vertex $v_i$ has an edge with $u_i$ and $u_{i-1}$, $i =1, \ldots,k$. This is because (see Figure~\ref{fig:arr} for a depiction)
$$
d(u_i,v_i) 
 \ge \frac{((i(a-b)+b-(2i-1)\epsilon)\log n}{n} - \frac{i(a-b)-(2i-1)\epsilon) \log n}{n}= \frac{b \log n}{n} \quad \text{and}
$$ 
\begin{align*}
 d(u_i,v_i) &\le 
   \frac{i(a-b)+b -(2i-2)\epsilon\log n}{n} - \frac{(i(a-b)-2i\epsilon)\log n}{n} = \frac{(b + 2\epsilon)\log n}{n}.
\end{align*}

Similarly, 
\vspace{-10pt}
\begin{align*}
d(u_{i-1},v_i)& 
\ge \frac{((i(a-b)+b-(2i-1)\epsilon)\log n}{n} - \frac{(i-1)(a-b)-(2i-3)\epsilon) \log n}{n}\\
& = \frac{(a-2\epsilon)\log n}{n} \quad \text{and}
\end{align*}
\begin{align*}
d(u_{i-1},v_i) 
&\le \frac{i(a-b)+b -(2i-2)\epsilon\log n}{n} - \frac{((i-1)(a-b)-2(i-1)\epsilon)\log n}{n} = \frac{a\log n}{n}.
\end{align*}
 This implies that $u_0$ is connected to $u_i$ and $v_i$ for all $i=1,\dots,k$. Using Lemma \ref{lem:lemma1}, the first event implies  that 
 the connected components are cycles spanning the entire line $[0,1]$.
 Now consider two such disconnected components, one of which consists of the nodes $u_0,\{u_1, u_2, \ldots, u_k\}$ and $\{v_1, v_2, \ldots, v_k\}$. There must exist a node $t$ in the other component (cycle) such that $t$ is on the right of $u_0$ and $d(u_0,t) \equiv \frac{x \log n}{n}  \le \frac{a \log n}{n}$. If $x\leq b$, $\exists i \mid  i\leq k \text{ and } i(a-b)+b - a-(2i-2)\epsilon \le x\leq i(a-b)-(2i-1)\epsilon$ (see Figure~\ref{fig:arr1d}). 
When $x\leq b$, we can calculate the distance between $t$ and $v_i$ as 
\begin{align*}
d(t,v_i)
&\ge \frac{(i(a-b)+b-(2i-1)\epsilon)\log n}{n} - \frac{(i(a-b)-(2i-1)\epsilon) \log n}{n} = \frac{b \log n}{n}
\end{align*}
and 
\begin{align*}
d(t,v_i)
&\le \frac{(i(a-b)+b -(2i-2)\epsilon)\log n}{n} - \frac{(i(a-b)+b-a-(2i-2)\epsilon)\log n}{n} = \frac{a\log n}{n}.
\end{align*}
Therefore  $t$ is connected to $v_i$ when  $x\leq b.$ If $x >b$ then $t$ is already connected to  $u_0$. Therefore the two components (cycles) in question are connected.
This is true for all cycles and hence there is only a single component in the entire graph. Indeed, if we consider the cycles to be disjoint super-nodes, then we have shown  that there must be a star configuration. 
\end{proof}

The following result is an immediate corollary of the connectivity upper bound.
\begin{corollary}\label{cor:patch}
Consider a random  graph $G(V,E)$  is being  generated as a variant of the ${\rm RAG}^{\ast}_1$ where each  $u,v\in V$ forms an edge if and only if  $d(u,v)\in \left[0,c\frac{\log n}{n}\right]\cup \left[b\frac{\log n}{n},a\frac{\log n}{n}\right], 0 <c<b<a$.  If $a-b + c>1$ or if $a-b > 0.5, a>1$, the graph $G$ is connected with  probability $1-o(1)$. \label{cor:patches}
\end{corollary}

The above corollary can be further improved for some regimes of $a,b,c$. In particular, we can get the following result (proof delegated to  the appendix).

\begin{corollary}\label{cor:extra}
Consider a random  graph $G(V,E)$  is being  generated as a variant of the ${\rm RAG}^{\ast}_1$ where each  $u,v\in V$ forms an edge if and only if  $d(u,v)\in \left[0,c\frac{\log n}{n}\right]\cup \left[b\frac{\log n}{n},a\frac{\log n}{n}\right], 0 <c<b<a$. If
any of the following conditions are true:
\begin{enumerate}
\item $2(a-b)+ c/2 > 1 \text{ when } a-b<c \text{ and }b>3c/2$
\item $b-c > 1 \text{ when } a-b<c \text{ and }b\le 3c/2$
\item $a > 1 \text{ when } a-b\ge c \text{ and } b\le 3c/2 $
\item $(a-b)+ 3c/2  > 1 \text{ when } a-b\ge c \text{ and } b>3c/2$,
\end{enumerate}
the graph $G$ is connected with  probability $1-o(1)$.
\end{corollary}
\section{Connectivity of High Dimensional Random Annulus Graphs: Proof of Theorem \ref{thm:highdem1}}
\label{sec:hrag}
In this section we show a proof sketch of Theorem \ref{thm:highdem1} to establish the sufficient condition of connectivity of random annulus graphs. The details of the proof and the necessary conditions are provided in Appendix~\ref{sec:hrag-detail}.


 
Note, here $r_1\equiv   b\left(\frac{\log n}{n}\right)^{1/t}$ and $r_2 \equiv a\left(\frac{\log n}{n}\right)^{1/t}$.
We show the upper bound for connectivity of a Random Annulus Graphs in $t$ dimension as shown in Theorem \ref{thm:highdem1}. 
For this we first define a \emph{pole} as a vertex which is connected to all vertices within a distance of $r_2$ from itself. 
In order to prove  Theorem \ref{thm:highdem1}, we first show the existence of a pole with high probability in Lemma \ref{lem:pole}.
\begin{lemma}\label{lem:pole}
In a ${\rm RAG}_t\left(n, \left[b\left(\frac{\log n}{n}\right)^{1/t},a\left(\frac{\log n}{n}\right)^{1/t}\right]\right), 0 <b <a$, with  probability $1-o(1)$ there exists a pole.
\end{lemma}
 Next, Lemma \ref{lem:high_stuff1} shows that for every vertex $u$ and every hyperplane $L$ passing through $u$ and not too close to the tangent hyperplane at $u$, there will be a neighbor of $u$ on either side of the plane. Therefore, there should be a neighbor towards the direction of the pole. In order to formalize this, let us define a few regions associated with a node $u$ and a hyperplane $L:w^{T}x=\beta$ passing through $u$.
\begin{align*}
\mathcal{R}_{L}^1 &\equiv \{x \in S^t \mid r_1 \le  d(u,x) \le r_2, w^{T}x \le \beta \} \\
\mathcal{R}_{L}^2 &\equiv \{x \in S^t \mid r_1 \le  d(u,x) \le r_2, w^{T}x \ge \beta \} \\
\mathcal{A}_{L} & \equiv \{x \mid x \in \mathcal{S}^t, \quad w^{T}x=\beta \}.
\end{align*}
Informally, $\mathcal{R}_{L}^1$ and $\mathcal{R}_{L}^2$ represent the partition of the annulus on either side of the hyperplane $L$ and $\mathcal{A}_L$ represents the region on the sphere lying on $L$. 
\begin{lemma}\label{lem:high_stuff1}
If we sample $n$ nodes from $S^t$ according to ${\rm RAG}_t\left(n,\left[b\left(\frac{\log n}{n}\right)^{1/t},a\left(\frac{\log n}{n}\right)^{1/t}\right]\right)$, then for every node $u$ and every hyperplane $L$ passing through $u$ such that $\mathcal{A}_L$ is not all within  distance $r_2$ of $u$, node $u$ has a neighbor on both sides of the hyperplane $L$ with probability at least $1-\frac{1}{n}$ provided 
$
(a/2)^t-b^t \ge   \frac{8\sqrt{\pi}(t+1)^2\Gamma(\frac{t+2}{2})}{\Gamma(\frac{t+3}{2})}
$ and $a>2b$.
\end{lemma}
We found the proof of this lemma to be challenging. Since, we do not know the location of the pole, we need to show that every point has a neighbor on both sides of the plane $L$ no matter what the orientation of the plane. Since the number of possible orientations is uncountably infinite, we cannot use a union-bound type argument. To show this we have to rely on the VC Dimension of the family of sets $\{x \in S^t \mid r_1 \le  \|u-x\|_2 \le r_2, w^{T}x \ge \beta , \mathcal{A}_{L:w^{T}x=\beta}  \text{ not all within } r_2 \text{ of } u\}$ for all hyperplanes $L$ (which can be shown to be less than $t+1$). We rely on the celebrated result of \citep{haussler1987} (we derived a continuous version of it), see Theorem~\ref{thm:VC_dim}, to deduce our conclusion.

For a node $u$ {and its corresponding location $X_u=(u_1,u_2,\dots,u_{t+1})$}, define the particular hyperplane $L^{\ast}_u : x_1=u_1$ which is normal to the line joining $u_0 \equiv (1,0,\dots,0)$ and the origin and passes through $u$. We now need one more lemma that will help us prove Theorem \ref{thm:highdem1}.
\begin{lemma}\label{lem:cut_twice}
For a particular node $u$ and corresponding hyperplane $L^{\ast}_u$, if every point in $\mathcal{A}_{L^{\ast}_u}$ is within distance $r_2$ from $u$,  then $u$ must be within $r_2$ of $u_0$.
\end{lemma}

For now, we assume that the Lemmas \ref{lem:pole}, \ref{lem:high_stuff1} and \ref{lem:cut_twice} are true and show why these lemmas together imply the proof of Theorem \ref{thm:highdem1}.
\begin{proof}[Proof of Theorem \ref{thm:highdem1}] 
We consider an alternate (rotated but not shifted) coordinate system by multiplying every vector by an orthonormal matrix such that the new position of the pole is the $t+1$-dimensional vector $(1,0,\dots,0)$ where only the first co-ordinate is non-zero. Let the $t+1$ dimensional vector describing any node $u$ in this new coordinate system be $\hat{u}=(\hat{u}_1,\hat{u}_2,\dots,\hat{u}_{t+1})$. Now consider the hyperplane $L: x_1=\hat{u}_1$ and if $u$ is not connected to the pole already, then by Lemma \ref{lem:high_stuff1} and Lemma \ref{lem:cut_twice}, the node $u$ has a neighbor $u_2$ which has a higher first coordinate ($\hat{u}_2 >\hat{u}_1$). The same analysis applies for $u_2$ and hence we have a path where the first coordinate of every node is higher than the previous node. Since the number of nodes is finite, this path cannot go on indefinitely and at some point, one of the nodes is going to be within $r_2$ of the pole and will be connected to the pole. Therefore every node is going to be connected to the pole and hence our theorem is proved. 
\end{proof}
\section{The Geometric Block Model in one dimension (${\rm GBM_1}$)}
\label{sec:gbm}
In this section, we prove the necessary condition for exact cluster recovery of ${\rm GBM_1}$ and give an efficient algorithm that matches that within a constant factor. 
\subsection{Immediate consequence of ${\rm RAG}_1^\ast$ connectivity}
The following lower bound for ${\rm GBM_1}$ can be obtained as a  consequence of Theorem~\ref{thm:rag}. 

\begin{theorem*}[Impossibility in ${\rm GBM_1}$ in Theorem \ref{gbm:upper}]
 If $a-b < 0.5$ or $a < 1$, then any algorithm to recover the partition in ${\rm GBM_1}(\frac{a \log n}{n},\frac{b \log n}{n})$ will give incorrect output with probability $1-o(1)$.
\end{theorem*}
\begin{proof}
Consider the scenario that not only the geometric block model graph  ${\rm GBM_1}(\frac{a \log n}{n},\frac{b \log n}{n})$ was provided to us, but also the random values $X_u \in [0,1]$ for all vertex $u$ in the graph were provided. We will show that we will still not be able to recover the correct partition of the vertex set $V$ with probability at least $0.5$ (with respect to choices of $X_u,~u, v\in V$ and any randomness in the algorithm).

In this situation, the edge $(u,v)$ where $d_L(X_u,X_v) \le \frac{b \log n}{n}$ does not give any new information than $X_u,X_v$. However the edges $(u,v)$ where $\frac{b \log n}{n} \le d_L(X_u,X_v) \le \frac{a \log n}{n}$ are informative, as existence of such an edge will imply that $u$ and $v$ are in the same part. These edges constitute a one-dimensional random annulus graph ${\rm RAG}^\ast_1(n, [\frac{b \log n}{n},\frac{a \log n}{n}])$. But if there are more than two components in this random annulus graph, then it is impossible to separate out the vertices into the correct two parts, as the connected components can be assigned to any of the two parts and the ${\rm RAG}_1^\ast$ along with the location values ($X_u, u \in V$) will still be consistent. 

What remains to be seen that ${\rm RAG}^\ast_1(n, [\frac{b \log n}{n},\frac{a \log n}{n}])$ will have $\omega(1)$ components with high probability  if $a-b < 0.5$ or $a < 1$. This is certainly true when $a-b < 0.5$ as we have seen in Theorem~\ref{gbm:upper}, there can indeed be $\omega(1)$ isolated nodes with high probability. On the other hand, when $a<1$, just by using an analogous argument it is possible to show that there are $\omega(1)$ vertices that do not have any neighbors on the left direction (counterclockwise). We delegate the proof of this claim as Lemma \ref{lem:disc} in the appendix.  
If there are $k$ such vertices, there must be at least $k-1$ disjoint candidates. This completes the proof.
\end{proof}

Indeed, when the locations $X_u$ associated with every vertex $u$ is provided, it is also possible to recover the partition exactly when $a-b > 0.5$ and $a > 1$, matching the above lower bound exactly (see Theorem~\ref{thm:gbmplus}).

\begin{corollary}[${\rm GBM_1}$ with known vertex locations]\label{thm:gbmplus}
Suppose a geometric block model graph ${\rm GBM_1}(\frac{a \log n}{n},\frac{b \log n}{n})$ is provided along with the associated  values of the locations $X_u$ for every vertex $u$. Any algorithm to recover the partition in ${\rm GBM_1}(\frac{a \log n}{n},\frac{b \log n}{n})$ will give \latest{ i) incorrect output with probability $1-o(1)$ if  $a-b < 0.5$ or $a < 1$ and ii) correct output probability $1-o(1)$ if  $a-b > 0.5$ and $a > 1$}. 
\end{corollary}
\begin{proof}
We need to only prove that  it is possible to recover the partition exactly with probability $1-o(1)$ when $a-b > 0.5$ and $a > 1$, since the other part is immediate from the impossibility guarantee in ${\rm GBM_1}$ in Theorem \ref{gbm:upper} (\latest{When $a-b<0.5$ or $a<1$, there will be an isolated node with high probability. It is impossible to identify the cluster for this node.}). For any pair of vertices $u,v$, we can verify if $d(u,v) \in [\frac{b \log n}{n},\frac{a \log n}{n}]$. If that is the case then by just checking in the ${\rm GBM_1}$ graph whether they are connected by an edge or not we can decide whether they belong to the same cluster or not respectively. What remains to be shown that  all vertices can be covered by this procedure. However that will certainly be the case since ${\rm RAG}^\ast_1(n , [\frac{b \log n}{n},\frac{a \log n}{n}])$ is connected with high probability.
\end{proof}



\subsection{A recovery algorithm for ${\rm GBM_1}$}
We now turn our attention to an efficient recovery algorithm for ${\rm GBM_1}$. Intriguingly, we show a simple triangle counting algorithm works well for ${\rm GBM_1}$ and recovers the clusters in the sparsity regime. Triangle counting algorithms are popular heuristics applied to social networks for clustering \citep{easley2012networks}, however they fail in SBM. Hence, this serves as another validation why ${\rm GBM_1}$ are well-suited to model community structures in social networks.

The algorithm is as follows. Suppose we are given a graph $G=(V:|V|=n,E)$ with two disjoint parts, $V_1, V_2 \subseteq V$ generated according to ${\rm GBM_1}(r_s, r_d)$. The algorithm (Algorithm~\ref{alg:alg1}) goes over  all edges  $(u,v)\in E$. It counts the number of triangles  containing the edge  $(u,v)$ by calling the \texttt{process} function that  
counts the number of common neighbors of $u$ and $v$. 


%
 
    \texttt{process} outputs `true' if it is confident that the nodes $u$ and $v$ belong to the same cluster and `false' otherwise. More precisely, if the count is within some prescribed values $E_S$ and $E_D$, it returns `false'\footnote{Note that, the thresholds $E_S$ and $E_D$ refer to the maximum and minimum value of triangle-count for an `inter cluster' edge.}.The algorithm removes the edge on getting a `false' from \texttt{process} function. After processing all the edges of the network, the algorithm is left with a reduced graphs (with certain edges deleted from the original). It then finds the connected components in the graph and returns them  as the parts $V_1$ and $V_2$. 

 \begin{remark}
The algorithm can iteratively maintain the connected components over the processed edges (the pairs for which process function has been called and it returned true) like the union-find algorithm. This reduces the number of queries as the algorithm does not need to  call the \texttt{process} function for the edges which are present in the same connected component. 
 \end{remark}
\captionof{algorithm}{Cluster recovery in ${\rm GBM_1}$}
\begin{algorithmic}[1]\label{alg:alg1}
{
\REQUIRE ${\rm GBM_1}$ $G = (V,E)$, $r_s, r_d$
\FOR {$(u,v)\in E$}
\IF{{\rm process}($u,v,r_s,r_d$)}
\STATE continue
\ELSE
\STATE $E.remove((u,v))$
\ENDIF
\ENDFOR
\RETURN {\rm connectedComponent}$(V,E)$
}
\end{algorithmic}
 
\captionof{algorithm}{\texttt{process}}
\begin{algorithmic}[1]\label{alg:process}
{
\REQUIRE $u$,$v$, $r_s$, $r_d$
\ENSURE  true/false\\
\COMMENT{Comment: When $a>2b$, $f_1=\min\{f: (2b+f)\log \frac{2b+f}{2b}-f > 1\}, f_2=\min\{f : (2b-f)\log \frac{2b-f}{2b}+f > 1$ and $E_S  = (2b +f_1)\frac{\log n}{n}$ and $E_D  = (2b - f_2)\frac{\log n}{n}$}
\STATE count $\leftarrow |\{z: (z,u)\in E, (z,v)\in E\}|$
\IF{$\frac{\text{count}}{n} \ge E_S(r_d,r_s)$ or $\frac{\text{count}}{n} \le E_D(r_d,r_s)$}
\RETURN true
\ENDIF
\RETURN false
}
\end{algorithmic}

\subsection{Analysis of Algorithm~\ref{alg:alg1}}
~\label{sec:theory}
Given a ${\rm GBM_1}$ graph $G(V,E)$ with two clusters $V = V_1 \sqcup V_2$, and  a pair of vertices $u,v \in V$, the events $\cE^{u,v}_z, z \in V$ of any other vertex $z$ being a common neighbor of both $u$ and $v$ given $(u,v) \in E$ are dependent ; however given the distance between the corresponding random variables  \correct{$d_L(X_u,X_v) =x$}, the events 
are independent. This is a crucial observation which helps us to overcome the difficulty of handling correlated edge formation.       

Moreover, given the distance between two nodes $u$ and $v$ are the same, the probabilities of $\cE^{u,v}_z\mid (u,v) \in E$ are different when $u$ and $v$ are in the same cluster and when they are in different clusters. Therefore the count of the common neighbors are going to be different, and substantially separated with high probability 
 for two vertices in cases when they are from the same cluster or from different clusters. However, this may not be the case, if we do not restrict the distance to be the same and look at the entire range of possible distances. 
 
 First, we quote two simple lemmas about the expected value of the commons neighbors.
 
 \begin{lemma}\label{lem:sep}
For any two vertices $u,v \in V_i: (u,v) \in E, i =1,2$ belonging to the same cluster with  $d_L(X_u,X_v) = x$, the count of common neighbors $C_{u,v} \equiv |\{z\in V: (z,u), (z,v) \in E\}|$ is a random variable distributed according to ${\rm Bin}(\frac{n}{2}-2, 2r_s-x)$ when $r_s \geq x> 2r_d$ and according to ${\rm Bin}(\frac{n}{2}-2,2r_s-x)+{\rm Bin}(\frac{n}{2},2r_d-x)$ when $x \leq \min(2r_d,r_s)$, where ${\rm Bin}(n,p)$ is a binomial random variable with mean $np$.
\end{lemma}
 
\begin{proof}
Without loss of generality, assume $u,v \in V_1$. For any vertex $z \in V$, let $\cE^{u,v}_z \equiv \{(u,z), (v,z) \in E\}$ be the event that $z$ is a common neighbor.
For $z\in V_1$,
\begin{align*}
&\Pr(\cE^{u,v}_z) =\Pr( (z,u) \in E, (z,v) \in E) \\
& = 2r_s - x,
\end{align*}
since $ \dist(X_u,X_v) =x$.
For $z \in V_2$, we have, 
\begin{align*}
&\Pr(\cE^{u,v}_z) =  \Pr( (z,u), (z,v) \in E ) \\
& =  \begin{cases}2r_d-x & \text{ if } x<2r_d \\ 0 & \text{ otherwise} \end{cases} .
\end{align*}
Now since there are $\frac{n}2-2$ points in $V_1 \setminus \{u,v\}$ and $\frac{n}2$ points in $V_2$, we have the statement of the lemma. 
\end{proof} 

In a similar way, we can prove.
 \begin{lemma}\label{lem:sep2}
For any two vertices $u\in V_1,v \in V_2: (u,v) \in E$ belonging to different clusters with $d_L(X_u,X_v) = x$ , the count of common neighbors $C_{u,v} \equiv |\{z\in V: (z,u), (z,v) \in E\}|$ is a random variable distributed according to ${\rm Bin}(n-2,2r_{d})$ when $r_s > 2r_d$ and according to ${\rm Bin}(n-2,\min(r_{s} + r_d -x,2r_d))$ when $r_s \leq 2r_d$ and $x\leq r_d$.
\end{lemma}
 


The distribution of the number of common neighbors given $(u,v)\in E$ and $d(u,v)=x$ is given in Table ~\ref{tab:tab1}. 
As throughout this paper, we have assumed that there are only two clusters of equal size. The functions change when the cluster sizes are different. 
 In the table, $u\sim v$ means $u$ and $v$ are in the same cluster. 

\begin{table*}[htbp]
\begin{center}
\resizebox{\textwidth}{!}{
\begin{tabular}{|c|p{3.2cm}|p{3cm}|p{3cm}|p{3.6cm}|} 
 \hline
$(u,v) \in E$ & \multicolumn{2}{c}{Distribution of count ($r_s>2r_d$)}&    \multicolumn{2}{c|}{Distribution of count ($r_s\le 2r_d$)} \\
 $d(u,v) =x$      & $u \sim v,  x\leq r_s$ & $u \nsim v, x\leq r_d$ &  $u \sim v, x\leq r_s$ & $u \nsim v, x\leq r_d$\\
 \hline
 $z\mid (z,u)\in E, (z,v)\in E$   &  {\correct{${\rm Bin}(\frac{n}{2}-2,2r_s-x)+\mathbb{1}\{x\leq 2r_d\}{\rm Bin}(\frac{n}{2},2r_d-x)$}} & ${\rm Bin}(n-2,2r_{d})$ &  {\correct{${\rm Bin}(\frac{n}{2}-2,2r_s-x)+{\rm Bin}(\frac{n}{2},2r_d-x)$}}  & \correct{${\rm Bin}(n-2, \min(r_{s}+r_d-x,2r_d))$}\\
\hline
\end{tabular}}
\end{center}
\caption{Distribution of triangle count for an edge $(u,v)$ conditioned on the distance between them $d(u,v) = d_L(X_u, X_v) = x$, when there are two equal sized clusters. Here ${\rm Bin}(n,p)$ denotes a binomial random variable with mean $np$.\label{tab:tab1}}
\end{table*}

At this point note  that, in a ${\rm GBM_1}(r_s,r_d)$ for any edge $u,v$ that do not belong to the same part, the expected total number of common neighbors of  $u$ and $v$ does not depend on  their distance. We will next show that in this case the normalized total number of common neighbors is concentrated around $2r_d$. Therefore, when Algorithm~\ref{alg:alg1} finished removing all the edges, with high probability all the `inter-cluster' edges are removed. However, some of the `in-cluster' edges will also be removed in the process. This is similar to the case when from an ${\rm RAG}^\ast_1(n, [0,r_s])$, all the edges that correspond to a distance close to $2r_d$ has been removed. 
This situation is shown for the case when $r_s \ge 2r_d$ in Figure~\ref{fig:gbm}.
Finally we show that the edge-reduced  ${\rm RAG}^\ast_1(n, [0,r_s])$ is still connected under certain condition. In  what follows we will assume the ${\rm GBM_1}(r_s,r_d)$ with $r_s\ge 2r_d$. .The other case of $r_s < 2r_d$ is similar.

In the next lemma, we show a concentration result for the count made in \texttt{process}.

\begin{figure}
\centering
\begin{tikzpicture}


\draw[->] (0,0)--(5,0); 
\draw[->] (0,0)--(0,5); 
\draw[blue, ultra thick] (0,4.2)--(2.4,1.8); 
\draw[blue, ultra thick] (2.4,1.8)--(3,1.5); 
\draw[red, ultra thick, dashed] (0,2.4)--(1.2,2.4); 
\draw[red,<->] (0,-0.5)--(1.2,-0.5); 
\draw[red,<->] (1.2,-0.5)--(2.4,-0.5); 
\draw[violet, <->] (0,-0.7)--(3,-0.7); 
\draw[red,<->] (-0.5,0)--(-0.5,1.2); 
\draw[red,<->] (-0.5,1.2)--(-0.5,2.4); 
\draw[violet,<->] (-1,0)--(-1,3); 
\draw[red,<->] (-1,3)--(-1,4.2); 
\draw[red,<->] (-1.5,3)--(-1.5,1.8); 
\draw (1.2,0.1)--(1.2,-0.1)node at (1.2,-0.25) {$r_d$}; 
\draw (2.4,0.1)--(2.4,-0.1)node at (2.4,-0.25){$2r_d$}; 
\draw (3,0.1)--(3,-0.1)node at (3,-0.25){$r_s$}; 
\draw[gray,dashed] (3,0)--(3,1.5); 
\draw[gray,dashed] (1.2,0)--(1.2,2.4); 
\draw[gray,dashed] (2.4,0)--(2.4,1.8); 

\draw[gray,dashed] (-1.5,1.8)--(2.4,1.8); 
\draw[gray,dashed] (-1,3)--(1.2,3); 
\node at (-0.3,2.6){$2r_d$}; 
\node at (5,-0.2){$x$}; 
\draw (0.1,4.2)--(-0.1,4.2);
\node at (-0.7,4.4){$r_s+r_d$}; 
\draw[<->, red ] (3,3.5)--(3.5,3.5)node at (3.8,3.5) {$r_d$}; 
\draw[<->, violet ] (3,4)--(3.5,4)node at (3.8,4) {$r_s$}; 
\draw[blue, ultra thick ] (3,3)--(3.5,3)node at (7.5,3) {Intra-cluster edge: $ \begin{cases}
r_s+r_d-x, x\leq 2r_d\\
r_s-x/2, 2r_d< x\le r_s
\end{cases}$}; 
\draw[red, ultra thick, dashed ] (3,2)--(3.5,2)node at (6.5,2) {Inter-cluster edge: $2r_d, 0\le x\le r_d$}; 

\end{tikzpicture}
\caption{Average number of common neighbors of $(u,v) \in E$  for varying values of $d(u,v) =x$ when $r_s\ge2r_d$.\label{fig:gbm}}
\end{figure}
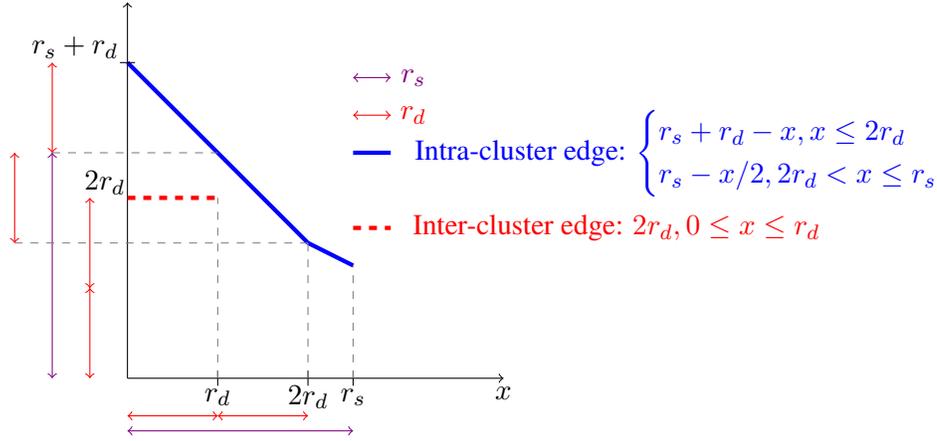

\begin{lemma}\label{lem:defn}
Suppose we are given the graph $G(V,E)$ generated according to ${\rm GBM_1}(r_s\equiv \frac{a\log n}{n},r_d\equiv \frac{b\log n}{n}), a \ge 2b.$ 
Our algorithm with $E_S = (2b + f_1)\frac{\log n}{n}$ and $E_D = (2b - f_2)\frac{\log n}{n}$, removes
all the edges $(u,v)\in E$ such that   $u$ and $v$ are in different parts with probability at least $1-o(1)$, where
\begin{align*}
f_1&=\min\{f: (2b+f)\log \frac{2b+f}{2b}-f > 1\} \\
f_2&=\min\{f : (2b-f)\log \frac{2b-f}{2b}+f > 1\}.
\end{align*}
\end{lemma} 

\begin{proof}
Here we will use the fact that for $a \geq 1$, the number of edges in ${\rm GBM_1}(r_s\equiv \frac{a\log n}{n},r_d\equiv \frac{b\log n}{n})$ is $O(n\log{n})$ with probability $1-\frac{1}{n^{\Theta(1)}}$. Consider any vertex $u \in V_1$ (symmetrically for $u \in V_2$), since the vertices are thrown uniformly at random in $[0,1]$, the probability that a $v \in V_1$, $v \neq u$, is a neighbor of $u$ is $\frac{a\log{n}}{n}$, and for $v \in V_2$, the corresponding probability is $\frac{b\log{n}}{n}$. Therefore, the expected degree of $u$ is $\frac{(a+b)}{2}\log{n}$. By a simple Chernoff bound argument, the degree of $u$ is therefore $O(\log{n})$ with probability $1-\frac{1}{n^c}$ for $c \geq 2$. By union bound over all the vertices, the total number of edges is $O(n\log{n})$ with probability $1-\frac{1}{n}$.

Let $Z$ denote the random variable that equals the number of common neighbors of two nodes $u,v \in V: (u,v) \in E$ such that $u,v$ are from different parts of the GBM. Using Lemma \ref{lem:sep2}, we know that $Z$ is sampled from the distribution ${\rm Bin}(n-2,2r_d)$, where $r_d = \frac{b\log n}{n}$. Therefore,
\begin{align*}
\Pr(Z \ge nE_S) \le \sum_{i=nE_S}^{n} {n \choose i}(2r_d)^{i}(n-2r_d)^{n-i} \le \exp\Big(-nD\Big((2b + f_1)\frac{\log n}{n}\|\frac{2b\log n}{n}\Big)\Big),
\end{align*}
where $D(p\|q)\equiv p\log\frac{p}{q}+ (1-p)\log \frac{1-p}{1-q}$ is the KL divergence between Bernoulli($p$) and Bernoulli($q$) distributions.
It is easy to see that,
\begin{align*}
nD(\frac{\alpha \log n}{n}||\frac{\beta \log n}{n})= \Big(\alpha\log\frac{\alpha}{\beta}+(\alpha-\beta)\Big)\log n-o(\log n).
\end{align*}
Therefore $\Pr(Z \ge nE_S) \le \frac{1}{n(\log n)^2}$ because 
$
(2b+f_1)\log\frac{2b+f_1}{2b}-f_1 > 1.
$
Similarly, we have that 
\begin{align*}
\Pr(Z \le nE_D) \le \sum_{i=0}^{nE_D} {n \choose i}(2r_d)^{i}(n-2r_d)^{n-i} \le \exp(-nD((2b-t)\frac{\log n}{n}\|\frac{2b\log n}{n}))\le \frac{1}{n(\log n)^2}.
\end{align*}
So all of the inter-cluster edges will be removed by Algorithm~\ref{alg:alg1}   with probability $1 - O(\frac{n \log n}{n(\log n)^2}) =1 -o(1)$, as with probability $1-o(1)$ the total number of edges in the graph is $O(n \log n)$.
\end{proof}
 
After Algorithm~\ref{alg:alg1}  finishes, in the edge-reduced ${\rm GBM_1}(\frac{a\log n}{n},\frac{b\log n}{n})$, all the edges are `in-cluster' edges with high probability. However some of the  
`in-cluster' edges are also deleted, namely, those that has a count of common neighbors between $E_S$ and $E_D$. In the next two lemmas, we show the necessary condition on  the `in-cluster' edges such that they do not get removed by Algorithm~\ref{alg:alg1}.

\begin{lemma}
Suppose we have the graph $G(V,E)$ generated according to ${\rm GBM_1}(r_s \equiv \frac{a\log n}{n},r_d\equiv \frac{b\log n}{n}), a \ge 2b$. Define $f_1,f_2, E_D,E_S$ as in  Lemma~\ref{lem:defn}. Consider an edge $(u,v) \in E$ where $u,v$ belong to the same part of the GBM and let  $d(u,v)\equiv x \equiv \frac{\theta \log n}{n}$. Suppose $\theta$ satisfies {\bf either} of the following conditions:
\begin{enumerate}
\item $
\frac{1}{2}\Big((4b+2f_1)\log \frac{4b+2f_1}{2a-\theta}+2a-\theta-4b-2f_1\Big) > 1 \quad \text{ and } \theta \le 2a-4b-2f_1
$
\item $
\frac{1}{2}\Big((4b-2f_2)\log \frac{4b-2f_2}{2a-\theta}+2a-\theta-4b+2f_2\Big) > 1
\quad \text{and} \quad 
a \ge \theta \ge \max\{2b,2a-4b+2f_2\}.
$
\end{enumerate}
Then Algorithm~\ref{alg:alg1} with $E_S  = (2b +f_1)\frac{\log n}{n}$ and $E_D  = (2b - f_2)\frac{\log n}{n}$ will not remove this edge with probability at least $1- O(\frac1{n (\log n)^2})$.
%
%
\label{lem:esed}
\end{lemma}
\begin{proof}
Let $Z$ be the number of common neighbors of $u,v$. 
Recall that, $u$ and $v$ are in the same cluster. 
We know from Lemma \ref{lem:sep2} that $Z$ is sampled from the distribution ${\rm Bin}(\frac{n}{2}-2,2r_s-x)+{\rm Bin}(\frac{n}{2},2r_d-x)$ when $x \le 2r_d$, and from the distribution ${\rm Bin}(\frac{n}{2}-2,2r_s-x)$ when $x \ge 2r_d$. 
We have,
\begin{align*}
&\Pr(Z \le nE_S)\\
&=
\begin{cases}
\sum_{i=0}^{nE_S}{\frac{n}{2}-2 \choose i} (2r_s-x)^{i}(1-2r_s+x)^{\frac{n}{2}-i-2}\sum_{j=0}^{nE_S-i}{\frac{n}{2} \choose j} (2r_d-x)^{j}(1-2r_d+x)^{\frac{n}{2}-j} \text{ if $x\le 2r_d$ }\\
\sum_{i=0}^{nE_s}{\frac{n}{2}-2 \choose i} (2r_s-x)^{i}(1-2r_s+x)^{\frac{n}{2}-i} \text{ otherwise} \\
\end{cases} \\
& \le e^{-\frac{n}{2}D(2E_S|| \frac{(2a-\theta)\log n}{n})} \text{ since  } 2a-\theta \ge 4b+2f_1\\
& \le e^{-\frac{n}{2}D(\frac{(4b+2f_1)\log n}{n}|| \frac{(2a-\theta)\log n}{n})} \le \frac{1}{n \log^2 n},
\end{align*}
because of condition 1 of this lemma. Therefore, this edge will not be deleted with high probability.

Similarly, let us find the probability of $Z \ge n E_D = (2b-f_2) \log n.$
Let us just assume the worst case when  $\theta \le 2b$: that the edge is being deleted (see condition 2, this is prohibited if that condition is satisfied).
%
Otherwise, $\theta > 2b$ and,
%
\begin{align*}
\Pr(Z \ge nE_D)&=\sum_{i=nE_D}^{n}{\frac{n}{2}-2 \choose i} (2r_s-x)^{i}(1-2r_s+x)^{\frac{n}{2}-i-2} \\
&\le e^{-\frac{n}{2}D(2E_D\|\frac{(2a-\theta)\log n}{n})} \text{   if   } 2a-\theta \le 4b-2f_2\\
&= e^{-\frac{n}{2}D(\frac{(4b-2f_2)\log n}{n}\|\frac{(2a-\theta)\log n}{n})} \le \frac{1}{n \log^2 n},
\end{align*}
because of condition 2 of this lemma.
\end{proof}

Now we are in a position to prove our main theorem from this part. Let us restate this theorem.

\begin{theorem*}(Sufficient Condition in Theorem \ref{gbm:upper})~
Suppose we have the graph $G(V,E)$ generated according to ${\rm GBM_1}(r_s \equiv \frac{a\log n}{n},r_d\equiv \frac{b\log n}{n}), a \ge 2b$. 
Define,
\begin{align*}
f_1&=\min\{f: (2b+f)\log \frac{2b+f}{2b}-f > 1\} \\
f_2&=\min\{f : (2b-f)\log \frac{2b-f}{2b}+f > 1\}\\
\theta_1 &= \max\{\theta:\frac{1}{2}\Big((4b+2f_1)\log \frac{4b+2f_1}{2a-\theta}+2a-\theta-4b-2f_1\Big) > 1  \text{ and } 0 \le \theta \le 2a-4b-2f_1\}\\
\theta_2 &= \min\{\theta: \frac{1}{2}\Big((4b-2f_2\log \frac{4b-2f_2}{2a-\theta}+2a-\theta-4b+2f_2\Big) > 1
 \text{ and }  
a \ge \theta \ge \max\{2b,2a-4b+2f_2\}\}.
\end{align*}
Then, if   $a-\theta_2+\theta_1>2$ or $a-\theta_2 > 1, a>2$, Algorithm~\ref{alg:alg1} with $E_S  = (2b +f_1)\frac{\log n}{n}$ and $E_D  = (2b - f_2)\frac{\log n}{n}$ will recover the correct partition  in the GBM with  probability $1-o(1)$  . 
\end{theorem*}
\begin{proof}
From Lemma~\ref{lem:defn}, we know that after Algorithm~\ref{alg:alg1} goes over all the edges, the edges with end-points being in different parts of the GBM are all removed with probability $1-o(1)$. 
There are $O(n \log n)$ edges in the GBM with  probability $1-o(1).$
From Lemma~\ref{lem:esed}, 
we can say that no edge with both ends at the same part is deleted with probability at least $1-o(1)$ (by simply applying a union bound).

After Algorithm~\ref{alg:alg1} goes over all the edges, the remaining edges from a disjoint union of two random annulus graphs of $\frac{n}{2}$ vertices each. For any two vertices $u,v$ in the same part, there will be an edge if $d(u,v) \in [0,\theta_1] \cup [\theta_2, a]$. From Corollary \ref{cor:patch}, it is evident that each of these two parts (each part  is of size $\frac{n}{2}$) will be connected if either $a-\theta_2+\theta_1>2$ or $a-\theta_2 > 1, a>2$.
\end{proof}


 It is also possible to incorporate the result of Corollary \ref{cor:extra} as well to get somewhat stronger recovery guarantee for our algorithm.

\sloppy
\section{High Dimensional GBM: Proof of Sufficient Condition in Theorem \ref{theorem:intro-1}}
\label{sec:sparse-high}
In this section, we show that our algorithm for recovery of clusters in GBM, i.e., Algorithm \ref{alg:alg1} extends to higher dimensions. Recall the precise definition of the high-dimensional GBM:
\begin{definition}[The GBM in High Dimensions]
Given $V = V_1\sqcup V_2, |V_1|=|V_2| = \frac{n}2$,  choose a random vector $X_u$ independently uniformly distributed in $S^t$ for all $u \in V$.
The geometric block model  ${\rm GBM}_t(r_s, r_d)$ with parameters $r_s> r_d$ is a random graph where an edge exists between vertices $u$ and $v$  if and only if,
\begin{align*}
\norm{X_u-X_v}_2 \le r_s & \text{ when } u, v \in V_1 \text{ or } u,v \in V_2\\
\norm{X_u-X_v}_2 \le r_d & \text{ when } u \in V_1, v \in V_2 \text{ or } u\in V_2,v \in V_1.
\end{align*}
\end{definition}

Indeed, for the higher dimensional case the algorithm remains exactly the same, except the value of $E_D$ and $E_S$ in the subroutine \texttt{process} needs to be changed.
Recall that the algorithm proceeds by checking each edge and counting the number of triangle the edge is part of. 
If the count is between $E_D$ and $E_S$ the edge is removed. In this process  we claim to remove all inter-cluster edges with high probability. The main difficulty lies in proving that the original communities remain connected in the redacted graph. For that we crucially use the connectivity results of the high dimensional random annulus graphs (from Section~\ref{sec:hrag}) in somewhat different way that what we do for the one dimensional case.

%


\subsection{Analysis of Algorithm \ref{alg:alg1} in High Dimension}

Let us define a few more terminologies to simplify the expressions for high dimensional space. The volume of a $t$-sphere with unit radius is $ a_t=\frac{2\pi^{t+1/2}}{\Gamma(\frac{t+1}{2})}$. Let the spherical cap ${B}_t(O,r) \subset S^t$ define a region on the surface of this $t$-sphere $S^t$ such that every point $u
\in {B_t}(O,r)$ satisfies $\|u-O\|_2 \le  r$. Let us denote the volume of the spherical cap ${B_t}(O,r)$ normalized with $a_t$  by
$B_t(r)$. Similarly ${B_t}(O,[r_1,r_2])$ refers to a region on the $t$-sphere such that every point $u \in {B}_t(O,[r_1,r_2])$ satisfies $r_1 \le \|u-O\|_2 \le r_2$ and ${B_t}(r_1,r_2)$ refers to the volume normalized by $a_t$.  Now consider two such spherical caps ${B_t}(O_1,r_1)$ and ${B}_t(O_2,r_2)$ such that $d(O_1,O_2)=\ell$. In that case let us define the volume of the intersection of the two aforementioned spherical caps (again normalized by $a_t$) by $\mathcal{V}_t(r_1,r_2,\ell)$.

Let us use $u \sim v$ ($u \nsim v$) to denote $u$ and $v$ belong to the same cluster (different clusters). Let $\cE^{u,v}_z$ denote the event that $z$ is a common neighbor of $u$ and $v$ and $e(u,v)$ denote the event that there is an edge between $u$ and $v$.  Following are some simple observations.

\begin{observation}
\label{obs:1}
$\Pr(e(u,v) \mid u\sim v)=B_t(r_s)$ and $\Pr(e(u,v) \mid u\nsim v)=B_t(r_d)$.
\end{observation}
\begin{observation}
\label{obs:2}
$\Pr(\cE^{u,v}_z \mid z\sim u, u\sim v \text{ and } \|u-v\|_2=\ell)=\cV_t(r_s,r_s,\ell)$ and $\Pr(\cE^{u,v}_z \mid z\nsim u, u\sim v \text{ and } \|u-v\|_2=\ell)=\cV_t(r_d,r_d,\ell)$.
\end{observation}

In the following proof, we assume $r_s \leq 2r_d$. The other situation where the gap between $r_s$ and $r_d$ is higher is only easier to handle.

\begin{lemma}\label{lem:sepd}
For any two vertices $u,v \in V_i: (u,v) \in E, i =1,2$ such that $d(u,v)=\ell$ belonging to the same cluster, the count of common neighbors $C_{u,v} \equiv |\{z\in V: (z,u), (z,v) \in E\}|$ is a random variable distributed according to ${\rm Bin}(\frac{n}{2}-2,\mathcal{V}_t(r_s,r_s,\ell))$  when $r_s \geq \ell> 2r_d$ and according to ${\rm Bin}(\frac{n}{2}-2,\mathcal{V}_t(r_s,r_s,\ell))+{\rm Bin}(\frac{n}{2},\mathcal{V}_t(r_d,r_d,\ell)$ when $\ell \le 2r_d$.
\end{lemma}
\begin{lemma}\label{lem:sep2d}
For any two vertices $u\in V_1,v \in V_2: (u,v) \in E$ such that $\|u-v\|_2=\ell$ belonging to different clusters, the count of common neighbors $C_{u,v} \equiv |\{z\in V: (z,u), (z,v) \in E\}|$ is a random variable distributed according to ${\rm Bin}(n-2,B_t(r_d))$ when $r_s > 2r_d$ and according to ${\rm Bin}(n-2,\min(\mathcal{V}_t(r_s,r_d,\ell),B_t(r_d)))$ when $r_s \leq 2r_d$ and $\ell \leq r_d$.
\end{lemma}
\begin{proof}[Proof of Lemma \ref{lem:sepd}]
Without loss of generality, assume $u,v \in V_1$. In order for $(u,v) \in E$, we must have $r_s \geq \ell$. Now there are two cases to consider, $\ell > 2r_d$ and $\ell \leq 2r_d$. In case 1, for $z$ to be a common neighbor of $u$ and $v$, $z$ must be in $V_1$ by triangle inequality. Since, there are $\frac{n}{2}-2$ points in $V_1 \setminus \{u,v\}$, from Observation~\ref{obs:2}, $\cE^{u,v}_z) \sim {\rm Bin}(\frac{n}{2}-2,\mathcal{V}_t(r_s,r_s,\ell))$. In case 2, $z$ can also be part of $V_2$ and there are $\frac{n}{2}$ points in $V_2$, thus again from Observation~\ref{obs:2}, $\cE^{u,v}_z) \sim
{\rm Bin}(\frac{n}{2}-2,\mathcal{V}_t(r_s,r_s,\ell))+{\rm Bin}(\frac{n}{2},\mathcal{V}_t(r_d,r_d,\ell)$. 
\end{proof}
The proof of Lemma \ref{lem:sep2d} is similar. We now use the following version of the Chernoff bound to estimate the deviation on the number of common neighbors in the two cases: $u \sim v$ and $u \nsim v$.

\begin{lemma}[Chernoff Bound]
Let $X_1, \ldots, X_n$ be iid random variables in $\{0,1\}$. Let $X$ denote the sum of these $n$ random variables. Then for any $\delta>0$,
\[
 \begin{cases}
\Pr(X > (1+\delta) \avg(X)) \le e^{-\delta^2\avg(X)/3}  = \frac{1}{n\log^2 n}, \text{ when } \delta = \sqrt{\frac{3(\log{n}+2\log\log n)}{\avg(X)}},\\
\Pr(X < (1-\delta)\avg(X)) \leq e^{-\delta^2\avg(X)/2} = \frac{1}{n\log^2 n},\text{ when } \delta = \sqrt{\frac{2(\log{n}+2\log\log n)}{\avg(X)}}.
\end{cases}
\]
\end{lemma}

We take $E_S = c^{(t)}_s\cdot (B_t(r_d)n+\sqrt{6B_t(r_d) n\log n})$ and $E_D=c^{(t)}_d \cdot( n\mathcal{V}_t(r_s,r_d,r_d)-\sqrt{2nB_t(r_d)\log n })$ where $c^{(t)}_s \geq 1$ and $c^{(t)}_d \leq 1$ are suitable constants that depend on $t$.

\begin{lemma}
For any pair of nodes $(u,v)=e\in E,\  u\nsim v$, the \texttt{process} algorithm removes the edge $e$ with a probability of $1-O\left(\frac{1}{n\log^2 n}\right)$ when $E_S \geq B_t(r_d)n+\sqrt{6B_t(r_d) n\log n}$ and $E_D \leq n\mathcal{V}_t(r_s,r_d,r_d)-\sqrt{2nB_t(r_d)\log n }$. \label{lem:diff}
\end{lemma}

\begin{proof}
Let $Z$ denote the random variable for the number of common neighbors of two nodes $u,v \in V: (u,v) \in E, \|u-v\|_2 = \ell, u \nsim v$. From Lemma~\ref{lem:sep2d}, $E[Z]\leq (n -2)B_t(r_d)$. Using the Chernoff bound we know that with a probability of at least $1-\frac{1}{n\log^2 n}$
\begin{align*}
Z\le F_{\nsim}=&\  (n -2)B_t(r_d) + \sqrt{3(\log n+ 2\log \log n)(n-2)B_t(r_d) }\\
=& \  B_t(r_d)n+\sqrt{3B_t(r_d) n\log n}+o(1)\\
\le& \  E_S.
\end{align*}

Moreover again from Lemma~\ref{lem:sep2d}, $E[Z]=(n-2)\min(\mathcal{V}_t(r_s,r_d,\ell),B_t(r_d))$ as we assume $r_s \leq 2r_d$. Hence, with probability of at least  $1-\frac{1}{n\log^2 n}$

\begin{align*}
Z\ge f_{\nsim}= &\ \min_{\ell: \ell \le r_d,r_s \le 2r_d} ( ( n -2)\min(\mathcal{V}_t(r_s,r_d,\ell),B_t(r_d)) - \\
  & \sqrt{2(\log n+ 2\log \log n)(n-2)\min(\mathcal{V}_t(r_s,r_d,\ell),B_t(r_d)) }) \\
&\geq \ \min_{\ell: \ell \le r_d,r_s \le 2r_d} ( ( n -2)\min(\mathcal{V}_t(r_s,r_d,\ell),B_t(r_d)) - \\
  & \sqrt{2(\log n+ 2\log \log n)(n-2)B_t(r_d) }) \,\text{since } \cV_t(r_s,r_d,\ell) \subseteq B_t(r_d) \\
>&\  n\mathcal{V}_t(r_s,r_d,r_d)-\sqrt{2nB_t(r_d)\log n }\, \text{since } \cV_t(r_s,r_d,\ell) \text{ is a decreasing function of } \ell \\
\geq & \  E_D.
\end{align*}

Hence, $E_S \le Z\le E_D$ with a probability of $1-\frac{2}{n\log^2 n}$ for $(u,v)\in E, u \nsim v$. Hence $(u,v)$ gets removed with high probability by the algorithm.
\end{proof}
Applying a union bound, we therefore can assume all inter-cluster edges are removed with probability $1-o(1)$ as there is $O(n\log{n})$ edges. 

In the next two lemmas, we provide two different conditions on $\|u-v\|_2$ when $u \sim v$ such that our algorithm does not remove the edge $(u,v)$.  Then we obtain a sufficient condition for the two communities to remain connected by the edges that are not removed. 


\begin{lemma}
\label{lem:essame-1}
Given a pair of nodes $u,v$ belonging to the same cluster such that $(u,v) \in E$, the \texttt{process} algorithm does not remove the edge $e$ with probability of $1-O\left( \frac{1}{n\log^2 n}\right)$ when $\|u-v\|_2=\ell$ (say) satisfies the following:
\begin{align*}
\frac{n}{2}\Big(\mathcal{V}_t(r_s,r_s,\ell)+\mathcal{V}_t(r_d,r_d,\ell)\Big)- \sqrt{2n\log n}\Big(\sqrt{B_t(r_s)} + \sqrt{B_t(r_d)} \Big) > E_S. 
\end{align*}  
\end{lemma}
\begin{proof}
Let $Z$ denote the random variable corresponding to the number of common neighbors of $u,v$. Let $\mu_s(\ell)= \avg(Z|u \sim v, d(u,v)=\ell)$. From Lemma~\ref{lem:sepd}, $\mu_s(\ell)=(\frac{n}{2}-2)\cV_t(r_s,r_s,l)+\frac{n}{2}\cV_t(r_d,r_d,l)$.

Using the Chernoff bound, with a probability of $1-O\Big(\frac{1}{n\log^2 n}\Big)$

\begin{align*}
Z &>  
 (n/2-2)\mathcal{V}_t(r_s,r_s,\ell)+n/2\mathcal{V}_t(r_d,r_d,\ell)  - \sqrt{2(\log n+2\log\log n)\mathcal{V}_t(r_s,r_s,\ell)n/2}    \\&- \sqrt{2(\log n+2\log\log n)\mathcal{V}_t(r_d,r_d,\ell)n/2}\\
&\ge n/2\mathcal{V}_t(r_s,r_s,\ell)+n/2\mathcal{V}_t(r_d,r_d,\ell) -\Big(\sqrt{B_t(r_s)} + \sqrt{B_t(r_d)} \Big)\sqrt{2n\log n}\\
&=\frac{n}{2}\Big(\mathcal{V}_t(r_s,r_s,\ell)+\mathcal{V}_t(r_d,r_d,\ell)\Big)- \sqrt{2n\log n}\Big(\sqrt{B_t(r_s)} + \sqrt{B_t(r_d)} \Big).
\end{align*}

Therefore,  Algorithm \ref{alg:alg1} will not delete $e$ if 
\begin{align*}
\frac{n}{2}\Big(\mathcal{V}_t(r_s,r_s,\ell)+\mathcal{V}_t(r_d,r_d,\ell)\Big)- \sqrt{2n\log n}\Big(\sqrt{B_t(r_s)} + \sqrt{B_t(r_d)} \Big) > E_S.
\end{align*}

Note that there exists a maximum value of distance (referred to as $\ell_1$) such that whenever $\|u-v\| \leq \ell_1$, the condition will be satisfied.
\end{proof}

%


\begin{lemma}
\label{lem:edsamed-1}
Given a pair of nodes $u,v$ belonging to the same cluster such that $(u,v) \in E$, the \texttt{process} algorithm does not remove the edge $e$ with probability of $1-O\left( \frac{1}{n\log^2 n}\right)$ when $\ell \equiv \|u-v\|_2$ (say) satisfies the following:
$$\frac{n}{2}\Big(\mathcal{V}_t(r_s,r_s,\ell+\mathcal{V}_t(r_d,r_d,\ell)\Big))+ \sqrt{n\log{n}}\sqrt{[\mathcal{V}_t(r_s,r_s,\ell)+\mathcal{V}_t(r_d,r_d,\ell)]} \leq E_D.$$\end{lemma}
\begin{proof}

Let $Z$ denote the random variable corresponding to the number of common neighbors of $u,v$. Let $\mu_s(\ell)= \avg(Z|u \sim v, \|u-v\|_2=\ell)$. 
From Lemma~\ref{lem:sepd}, $\mu_s(\ell)=(\frac{n}{2}-2)\cV_t(r_s,r_s,l)+\frac{n}{2}\cV_t(r_d,r_d,l)$.

Using the Chernoff bound, with a probability of $1-O\Big(\frac{1}{n\log^2 n}\Big)$
\begin{align*}
Z &<  n/2\mathcal{V}_t(r_s,r_s,\ell)+n/2\mathcal{V}_t(r_d,r_d,\ell)  + \sqrt{2(\log n+2\log\log n)[\mathcal{V}_t(r_s,r_s,\ell)+\mathcal{V}_t(r_d,r_d,\ell)](n/2)}   \\
 &\le \frac{n}{2}\Big(\mathcal{V}_t(r_s,r_s,\ell+\mathcal{V}_t(r_d,r_d,\ell)\Big))+ \sqrt{n\log{n}}\sqrt{[\mathcal{V}_t(r_s,r_s,\ell)+\mathcal{V}_t(r_d,r_d,\ell)]}.
\end{align*}
The  \texttt{process} algorithm will not remove $e$ if 
$$\frac{n}{2}\Big(\mathcal{V}_t(r_s,r_s,\ell+\mathcal{V}_t(r_d,r_d,\ell)\Big))+ \sqrt{n\log{n}}\sqrt{[\mathcal{V}_t(r_s,r_s,\ell)+\mathcal{V}_t(r_d,r_d,\ell)]} \leq E_D.$$ 
Note that there exists a minimum value of distance (referred to as $\ell_2$) such that whenever $\|u-v\|_2 \geq \ell_2$, the condition will be satisfied.
\end{proof}

%

\begin{lemma}\label{lem:main1d}
Algorithm \ref{alg:alg1} can identify all edges $(u,v)$ correctly for which $\ell\equiv \|u-v\|_2$ satisfies either of  the following:
\begin{align*}
{\rm Cond. 1:} & \quad \frac{n}{2}\Big(\mathcal{V}_t(r_s,r_s,\ell)+\mathcal{V}_t(r_d,r_d,\ell)\Big)- \sqrt{2n\log n}\Big(\sqrt{B_t(r_s)} + \sqrt{B_t(r_d)} \Big) > E_S   
\end{align*}
or
\begin{align*}
{\rm Cond. 2:} & \quad \frac{n}{2}\Big(\mathcal{V}_t(r_s,r_s,\ell+\mathcal{V}_t(r_d,r_d,\ell)\Big))+ \sqrt{n\log{n}}\sqrt{[\mathcal{V}_t(r_s,r_s,\ell)+\mathcal{V}_t(r_d,r_d,\ell)]} \leq E_D. 
\end{align*}
with probability at least $1-O\Big(\frac{1}{\log n}\Big)$.
\end{lemma}

\begin{proof}
Follows from combining Lemma~\ref{lem:essame-1} and Lemma \ref{lem:edsamed-1}, and noting that in the connectivity regime, the number of edges is $O(n\log{n})$. 
\end{proof}

Let $\ell_1$ be the maximum value of $\|u-v\|_2$ such that Cond 1 is satisfied and $\ell_2$ is the minimum value of $\|u-v\|_2$ such that Cond 2 is satisfied. Also note that $\ell_1 \leq \ell_2$. We now give a condition on $\ell_1$ and $\ell_2$ such the two communities are each connected by the edges $(u,v)$ that satisfy either $\|u-v\|_2 \leq \ell_1$ or $\ell_2 \leq \|u-v\|_2 \leq r_s$.

\begin{lemma}
\label{RAG:cond}
If $(\ell_1/2)^t >  {8(t+1)\psi(t)} \frac{ \log n}{ n}$ 
 then the edges $e$ that satisfy $\|u-v\|_2 \leq \ell_1$ 
  constitute two disjoint connected components corresponding to the two original communities.
\end{lemma}
\begin{proof}
Proof of this lemma follows from the result of connectivity of random annulus graphs (RAG) in dimension $t$, i.e., Theorem~\ref{thm:highdem1}. 
\end{proof}

We now find out the values of $r_s$ and $r_d$ such that $\ell_1$ and $\ell_2$ satisfy the condition of Lemma~\ref{RAG:cond} as well as Cond. 1 and Cond 2. respectively.

\begin{theorem*}(Sufficient Condition in Theorem \ref{theorem:intro-1})
If $r_s=\Theta((\frac{\log{n}}{n})^{\frac{1}{t}})$ and $r_s-r_d=\Omega( (\frac{\log n}{n})^{\frac{1}{t}})$, Algorithm \ref{alg:alg1} recovers the clusters with probability $1-o(1)$. 
\end{theorem*}
\begin{proof}
Let us take $r_s=a_t \Big(\frac{\log n}{n}\Big)^{\frac{1}{t}}$, $r_d \leq b_t \Big(\frac{\log n}{n}\Big)^{\frac{1}{t}}$ for some large constants $a_t$ and $b_t$ that depends on $t$. Then to satisfy Lemma~\ref{RAG:cond}, we can take $\ell_1=a'_t \Big(\frac{\log n}{n}\Big)^{\frac{1}{t}}$ and $\ell_2=b'_t \Big(\frac{\log n}{n}\Big)^{\frac{1}{t}}$ again for suitable constants $a'_t$ and $b'_t$. While it is possible to concisely compute 
$B_t(r)$ and $\cV_t(r_1,r_2,x)$ \citep{li2011concise,ellis2007random}, for the purpose of analysis it is sufficient to know $B_t(r)=\Theta(r^t)$ for fixed $t$. Moreover, $\cV_t(r_1,r_2,x)=\Theta((r_1+r_2-x)^t)$ if $r_1+r_2\geq x \geq \max(r_1,r_2)$, $\Theta(\min\{r_1,r_2\}^t)$ if $x \leq \max(r_1,r_2)$ and $0$ otherwise.

Then, Cond 1. requires 
$$\frac{n}{2}\Big(\mathcal{V}_t(r_s,r_s,\ell_1)+\mathcal{V}_t(r_d,r_d,\ell_1)\Big)- \sqrt{2n\log n}\Big(\sqrt{B_t(r_s)} + \sqrt{B_t(r_d)} \Big) > E_S $$
and Cond 2. requires
$$\frac{n}{2}\Big(\mathcal{V}_t(r_s,r_s,\ell_2+\mathcal{V}_t(r_d,r_d,\ell_2)\Big))+ \sqrt{n\log{n}}\sqrt{[\mathcal{V}_t(r_s,r_s,\ell)+\mathcal{V}_t(r_d,r_d,\ell)]} \leq E_D.$$

By selecting the constants $c^{(t)}_s$ and $c^{t}_d$ involving $E_S$ and $E_D$ suitably, both the conditions are satisfied.
\end{proof}


\section{Experimental Results}\label{sec:exp}
In addition to validation experiments in Section~\ref{subsec:m1}, we also conducted an in-depth experimentation of our proposed model and techniques over a set of synthetic and real world networks. Additionally, we compared the efficacy and efficiency of our triangle-counting algorithm with the popular spectral clustering algorithm using normalized cuts\footnote{http://scikit-learn.org/stable/modules/clustering.html\#spectral-clustering} and the correlation clustering algorithm~\citep{bansal2004correlation}.
\begin{table*}[htbp]
\centering
\small{
 \resizebox{\textwidth}{!}{\begin{tabular}{||c| c| c| c| c| c| c|c|c||} 
 \hline
 Dataset &Total no.  & $T_1$ & $T_2$ & $T_3$ &\multicolumn{2}{c|}{Accuracy}  & \multicolumn{2}{c||}{Running Time (sec)}\\ 
   &of nodes &  &  &  &Triangle-Counting & Spectral clustering & Triangle-Counting&Spectral clustering\\ [0.5ex] 
\hline\hline
Political Blogs & 1222 & 20 & 2 & 1&  \textbf{0.788}& 0.53 & 1.62 & \textbf{0.29} \\
DBLP & 12138    & 10 & 1 & 2& \textbf{0.675}&0.63 & \textbf{3.93}&18.077\\
LiveJournal & 2366    & 20 & 1 & 1& \textbf{0.7768}&0.64 & \textbf{0.49} & 1.54\\ 
 \hline
\end{tabular}}}
\caption{\small Performance on real world networks}
\label{tab:realtab}
\end{table*}

\noindent{\bf Real Datasets.} We use three real datasets described below.

\begin{itemize}[noitemsep,leftmargin=5pt]
\item \textbf{Political Blogs.}~\citep{politicalblog} It contains a list of political blogs from 2004 US Election classified as liberal or conservative, and links between the blogs. The clusters are of roughly the same size with a total of 1200 nodes and 20K edges.

\item \textbf{DBLP.}~\citep{communityreal} The DBLP dataset is a collaboration network where the ground truth communities are defined by the research community. The original graph consists of roughly 0.3 million nodes. We process it to extract the top two communities of size $\sim$ 4500 and 7500 respectively. This is given as input to our algorithm.

\item \textbf{LiveJournal.}~\citep{livejournal} The LiveJournal dataset is a free online blogging social network of around 4 million users. Similar to DBLP, we extract the top two clusters of sizes 930 and 1400 which consist of around 11.5K edges.
\end{itemize}
  We have not used the academic collaboration (Section~\ref{subsec:m1}) dataset here because it is quite sparse and below the connectivity threshold regime of both GBM and SBM.

\noindent{\bf Synthetic Datasets.} We generate synthetic datasets of different sizes according to the GBM with $t=2, k=2$ and for a wide spectrum of  values of $r_s$ and $r_d$, specifically we focus on the sparse region where $r_s = \frac{a\log{n}}{n}$ and $r_d = \frac{b\log{n}}{n}$ with variable values of $a$ and $b$.

\noindent{\bf Experimental Setting.} For real networks,  it is difficult to calculate an exact threshold as the exact values of $r_s$ and $r_d$ are not known. Hence, we follow a three step approach. Using a somewhat large threshold $T_1$ we sample a subgraph $S$ such that $u,v$ will be in $S$ if there is an edge between $u$ and $v$, and they have at least $T_1$ common neighbors. We now attempt to recover the subclusters inside this subgraph by following our algorithm with a small threshold $T_2$. Finally, for nodes that are not part of $S$, say $x \in V \setminus S$, we select each $u \in S$ that  $x$ has an edge with and use a threshold of $T_3$ to decide if $u$ and $x$ should be in the same cluster. The final decision is made by taking a majority vote. We can employ sophisticated methods over this algorithm to improve the results further, which is beyond the scope of this work. 

We use the popular f-score metric which is the harmonic mean of precision (fraction of number of pairs correctly classified to total number of pairs classified into clusters) and recall (fraction of number of pairs correctly classified to the total number of pairs in the same cluster for ground truth), as well as the node error rate for performance evaluation. A node is said to be misclassified if it belongs to a cluster where the majority comes from a different ground truth cluster (breaking ties arbitrarily). Following this, we use the above described metrics to compare the performance of different techniques on various datasets.

\begin{figure*}[t!]
  \centering
 \begin{subfigure}[t]{0.3 \textwidth}
   \includegraphics[height=1.6in]{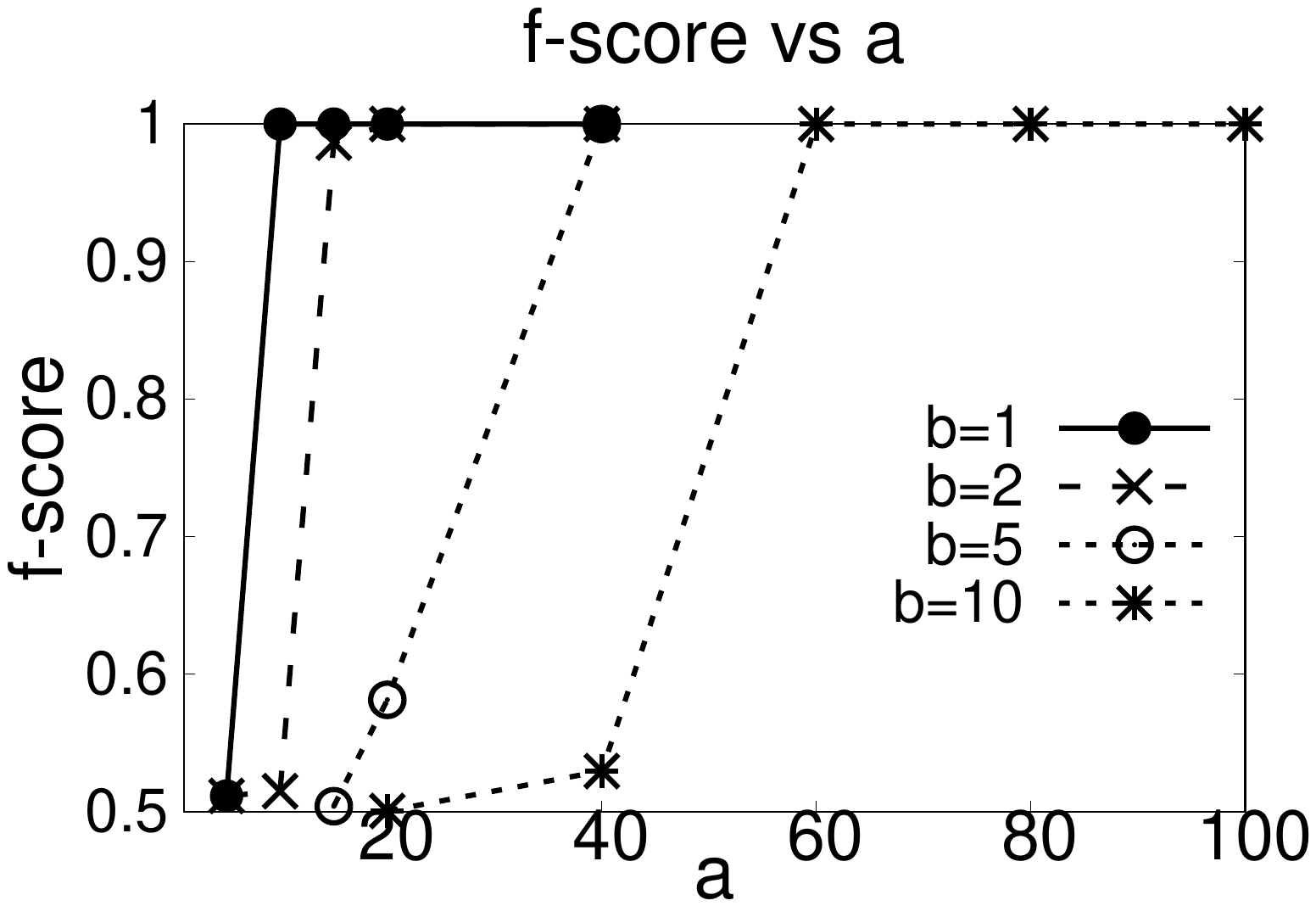}
   \caption{\small f-score with varying $a$, fixing $b$.}
       ~\label{fig:varya}
 \end{subfigure}%
 \begin{subfigure}[t]{0.3\textwidth}
    \includegraphics[height=1.6in]{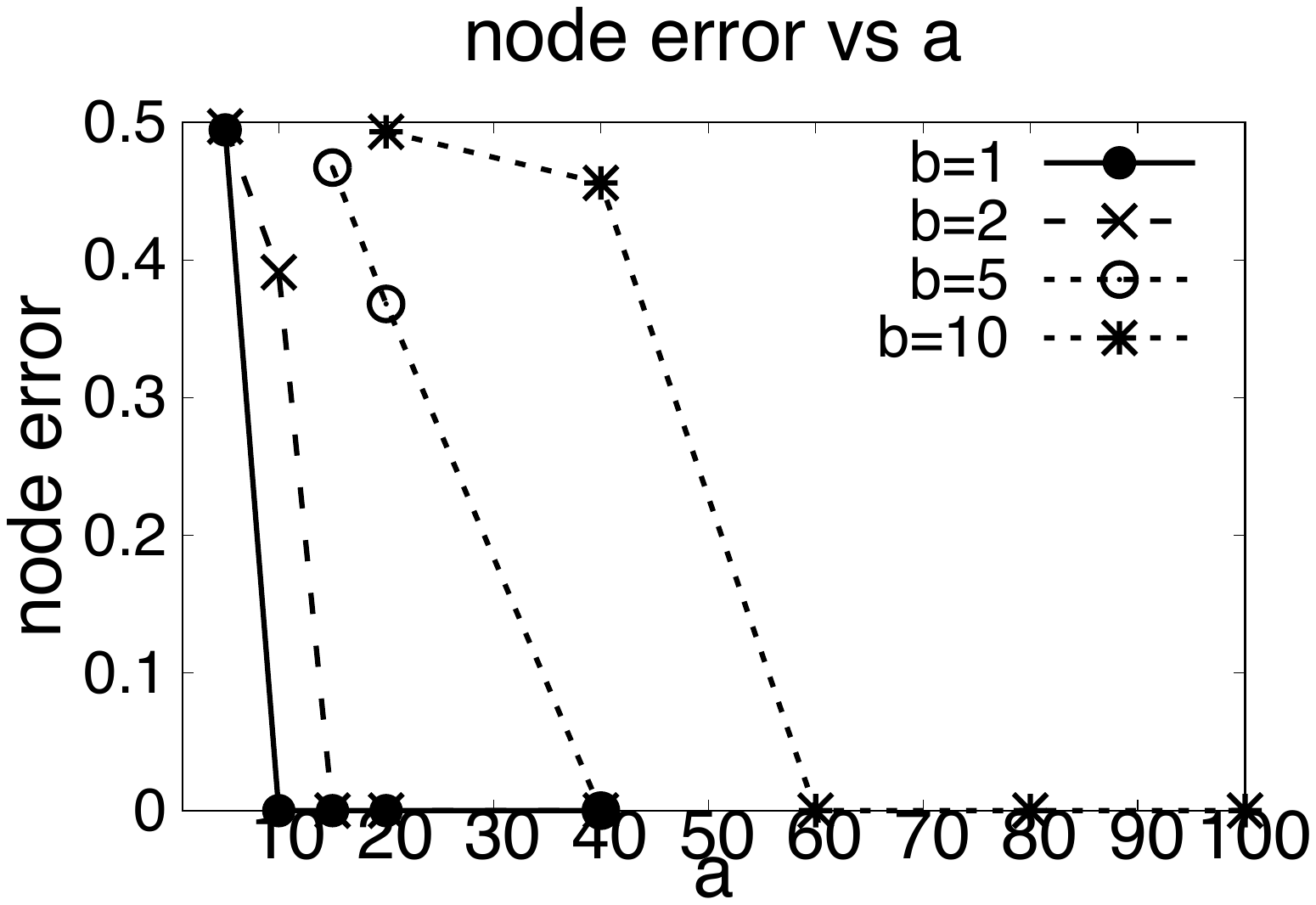}
    \caption{Fraction of nodes misclassified.}
          ~\label{fig:node}
  \end{subfigure}
 \caption{\small Results of the triangle-counting algorithm on a synthetic dataset with $5000$ nodes.}
\end{figure*}

\noindent{\bf Results.} We  compared our algorithm with the spectral clustering algorithm where we extracted two eigenvectors in order to extract two communities. Table~\ref{tab:realtab} shows that our algorithm gives an accuracy as high as $78\%$. The spectral clustering  performed worse  compared to our algorithm for all real world datasets. It obtained the worst accuracy of 53\% on political blogs dataset. The correlation clustering algorithm generates various small sized clusters leading to a very low recall, performing much worse than the triangle-counting algorithm for the whole spectrum of parameter values.

We can observe in Table~\ref{tab:realtab} that our algorithm is much faster than the spectral clustering algorithm for larger datasets (LiveJournal and DBLP). This  confirms that triangle-counting algorithm is more scalable than the spectral clustering algorithm.  The spectral clustering algorithm also works very well on synthetically generated SBM networks even in the sparse regime \citep{lei2015consistency,rohe2011spectral}. The superior performance of the simple triangle clustering algorithm over the real networks provide a further validation of GBM over SBM. Correlation clustering  takes 8-10 times longer as compared to triangle-counting algorithm for the various range of its parameters.  We also compared our algorithm with the Newman algorithm \citep{newman} that performs really well for the LiveJournal dataset (98\% accuracy). But it is extremely slow and performs much worse on other datasets. This is because the LiveJournal dataset has two well defined subsets of vertices with very few intercluster edges. The reason for the worse performance of our algorithm is the sparseness of the graph. If we create a subgraph by removing all nodes of degrees $1$ and $2$, we get $100\%$ accuracy with our algorithm. Finally, our algorithm is easily parallelizable to achieve better improvements. This clearly establishes the efficiency and effectiveness of triangle-counting.

\begin{figure}[t]
  \centering
   \begin{subfigure}[c]{0.3 \textwidth}
   \includegraphics[width=\textwidth]{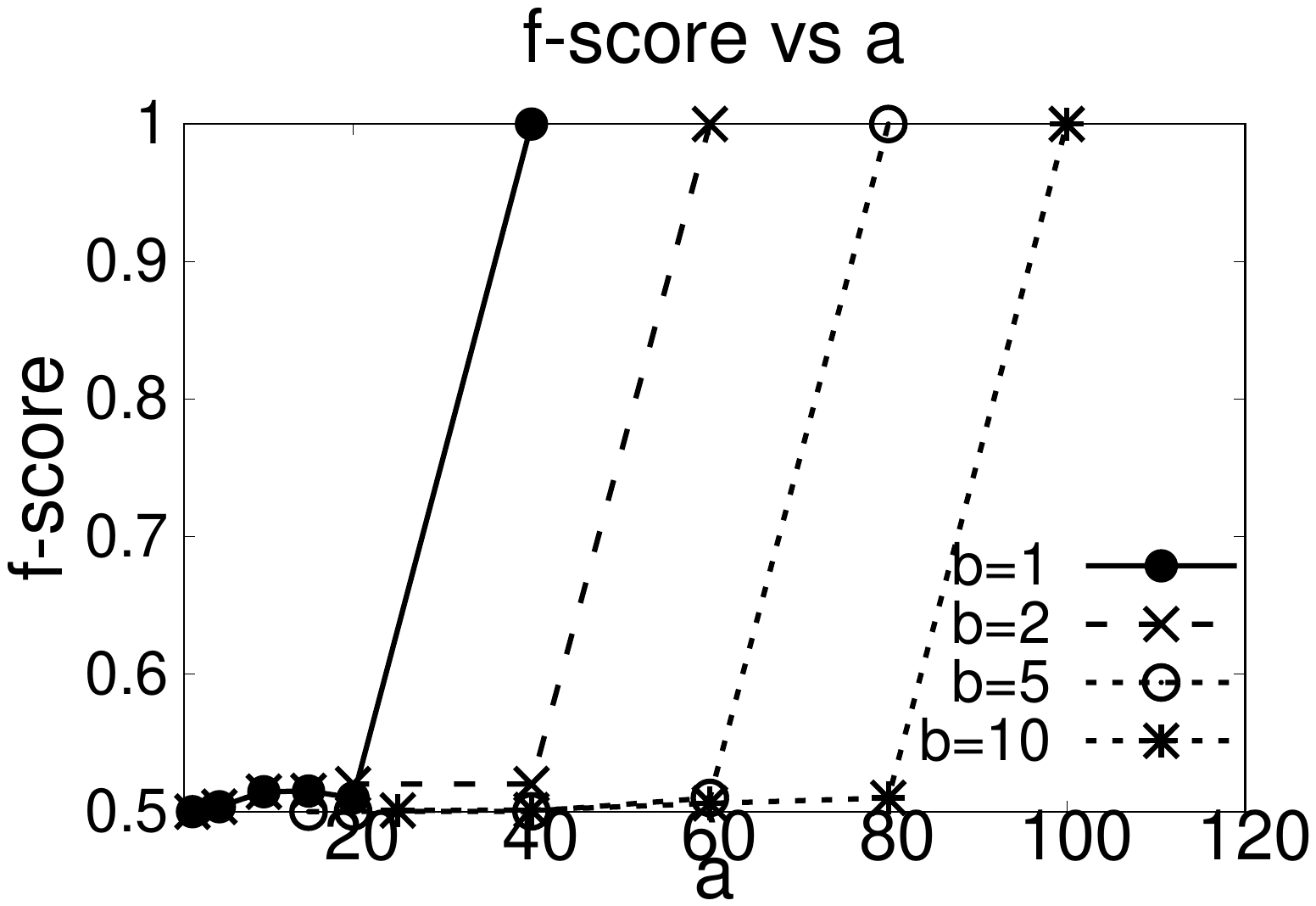}
    \caption{f-score with varying $a$, fixed $b$.}
      ~\label{fig:varya1}
 \end{subfigure}%
 \begin{subfigure}[c]{0.3\textwidth}
    \includegraphics[width=\textwidth]{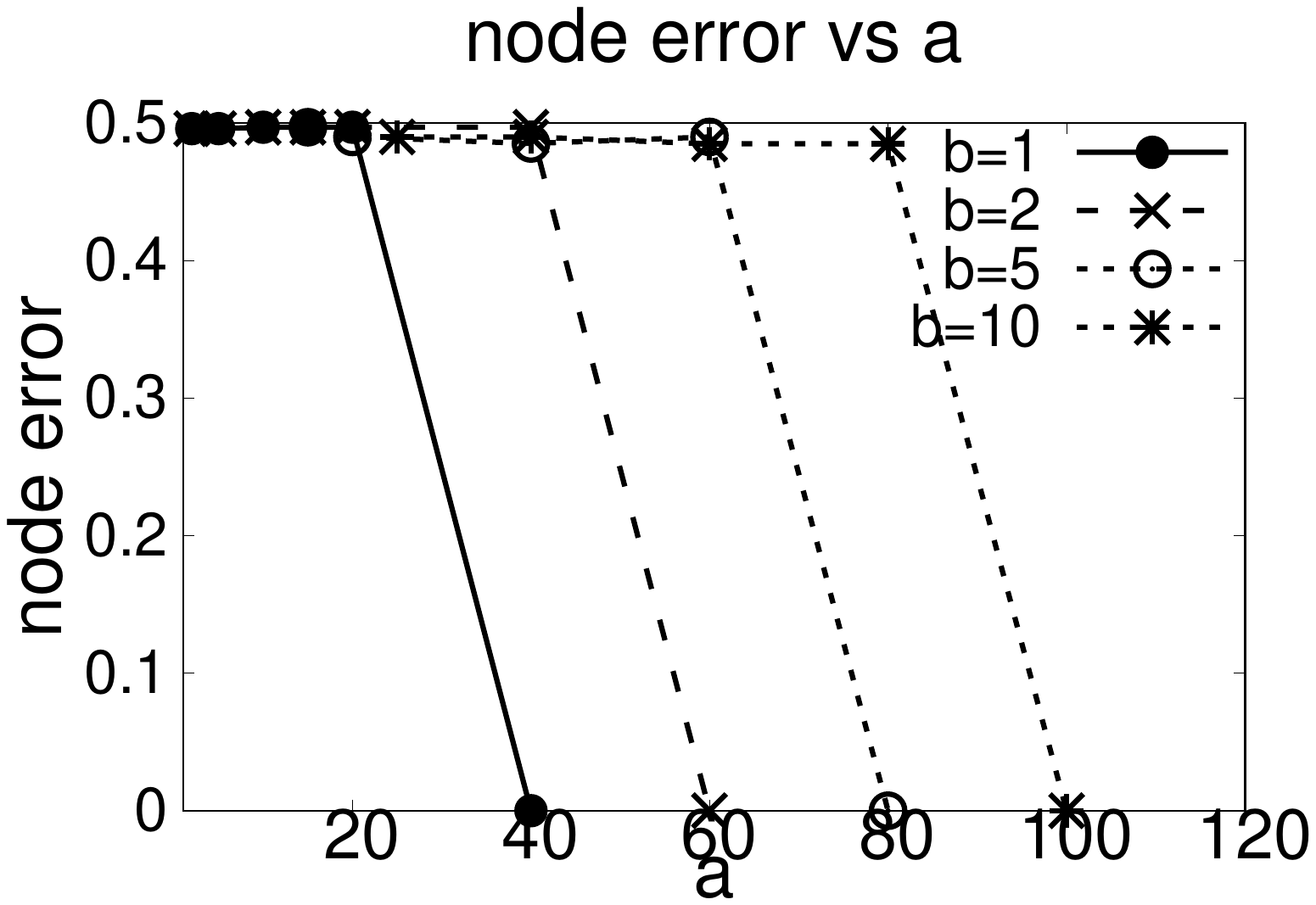}
     \caption{Fraction of  nodes misclassified.}
     ~\label{fig:fscore-1}
  \end{subfigure}
   \vspace{-0.1in}
 \caption{\small Results of the spectral clustering on a synthetic dataset with $5000$ nodes.}
\end{figure}

We observe similar gains on synthetic datasets. 
Figures~\ref{fig:varya} and \ref{fig:node} report results on the synthetic datasets with $5000$ nodes.  Empirically, our results demonstrate \correct{much superior performance of our algorithm as compared to theoretical guarantees. The empirical results are much better than the theoretical bounds because the concentration inequalities assume the worst value of the distance between the pair of vertices that are under consideration.}  We also see a clear threshold behavior on both f-score and node error rate in Figures~\ref{fig:varya} and \ref{fig:node}. We have also performed spectral clustering on this 5000-node synthetic dataset (Figures~\ref{fig:varya1} and \ref{fig:fscore-1}). Compared to the plots of figures \ref{fig:varya} and \ref{fig:node}, they show suboptimal performance, indicating the relative ineffectiveness of spectral clustering in GBM compared to the triangle counting algorithm.
\begin{figure}
    \centering
    \includegraphics[width=0.35\textwidth]{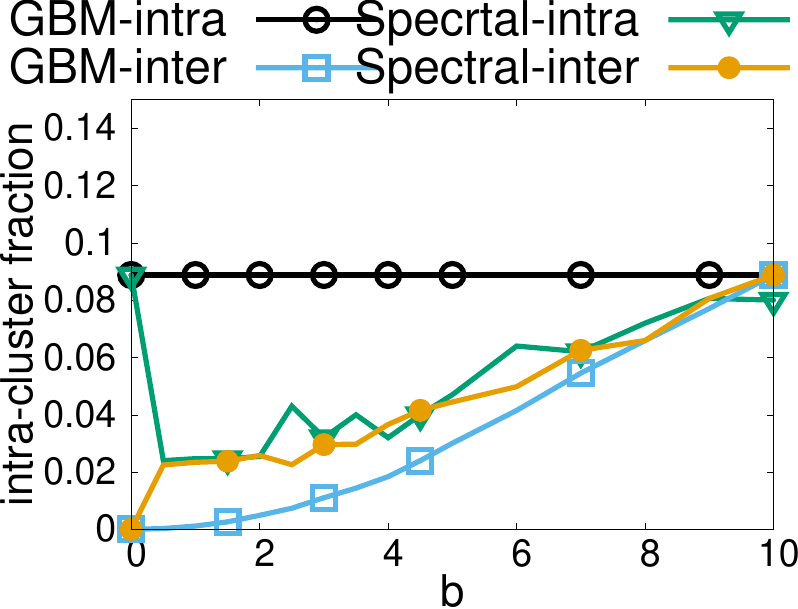}
    \caption{Fraction of triangles (intra-cluster and inter-cluster) in the ground truth of a graph generated according to ${\rm GBM_1}$ and the clusters returned by spectral clustering for varying $r_d = \frac{b\log n}{n}$. The number of nodes is $1000$ and $r_s=10\log n/n$ is fixed.
    Since $r_s$ is fixed, the fraction of intra-cluster triangles in the ground truth remains unchanged.
    }
    \label{fig:spectral}
\end{figure}

We ran another experiment to test whether spectral clustering techniques capture varied number of intra-cluster and inter-cluster triangles and if it can be used for datasets where higher order structures are more prevalent.
We tested the fraction of intra-cluster triangles returned by spectral clustering for datasets generated by GBM with varying inter-cluster connectivity threshold ($r_d$) and compared it against the fraction of intra-cluster triangles in the ground truth clustering. Figure~\ref{fig:spectral} shows that the clusters returned by spectral clustering have a very low fraction of intra-cluster triangles and thus it is not designed for clustering based on triangles (or higher-order structures in general). In fact the clusters returned by spectral clustering have similar fraction of intra-cluster and itner-cluster triangles. This evaluation also shows that GBM indeed has a large disparity between inter and intra cluster triangles and it may be useful to model datasets that contain higher order structures. 



\paragraph*{Acknowledgements:} This work is supported by NSF CCF awards 2133484, 2217058, 2223282.


\appendix
\input{rgg-one}

 \input{rgg-high}


\section{Proof of Lemma~\ref{lem:disc}}
\begin{lemma}\label{lem:disc}
A random geometric graph $G(n,\frac{a \log n}{n})$ will have $\omega(1)$ disconnected components for $a <1.$
\end{lemma}
\begin{proof}
Define an indicator random variable $A_u$ for a node $u$ which is $1$ if it does not have a neighbor on its left. We must have that 
\begin{align*}
\Pr(A_u)=\Big(1-\frac{a \log n}{n}\Big)^{n-1}.
\end{align*}
Therefore we must have that $\sum_u \avg A_u =n^{1-a}=\Omega(1)$ if $a<1$. This statement also holds true with high probability. To show this we need to prove that the variance of $\sum_u \avg A_u$ is bounded. We have that 
\begin{align*}
{\rm Var}(A) <\avg[A]+\sum_{u \neq v} {\rm Cov}(A_u,A_v)=\avg[A]+\sum_{u \neq v} \Pr(A_u=1 \cap A_v=1)-\Pr(A_u=1)\Pr(A_v=1)
\end{align*}
Now, consider the scenario when the  vertices $u$ and $v$ are at a distance more than $\frac{2a \log n}{n}$ apart (happens with probability at least $1-\frac{4a\log n}{n})$. Then the region in $[0,1]$ that is within distance $\frac{a\log n}{n}$ from both of the vertices is empty and therefore $\Pr(A_u=1 \cap A_v=1) = \Pr(A_u=1)\Pr(A_v=1|A_u=1) \leq  \Pr(A_u=1)\Pr(A_v=1) = (\Pr(A_u=1))^2$.  When the vertices  are within distance $\frac{2a \log n}{n}$ of one another, then
$
\Pr(A_u=1 \cap A_v=1) \le \Pr(A_u=1).
$
Therefore,
\begin{align*}
\Pr(A_u=1 \cap A_v=1) \le (1-\frac{4a\log n}{n}) (\Pr(A_u=1))^2 + \frac{4a\log n}{n}\Pr(A_u=1).
\end{align*}
Consequently,
\begin{align*}
\Pr(A_u=1 \cap A_v=1)-\Pr(A_u=1)\Pr(A_v=1) &\le (1-\frac{4a\log n}{n}) (\Pr(A_u=1))^2 \\+ \frac{4a\log n}{n}\Pr(A_u=1) -& (\Pr(A_u=1))^2 
\le  \frac{4a\log n}{n}\Pr(A_u=1).
\end{align*}
Now,
$$
{\rm Var}(A) \le \avg[A] + \binom{n}{2}\frac{4a\log n}{n}\Pr(A_u=1) \le \avg[A](1+ 2a\log n).
$$
By using Chebyshev bound, with probability at least $1-\frac{1}{\log n}$, 
$$A >n^{1-a}-\sqrt{n^{1-a}(1+2a\log n)\log n},$$

Now, observe that if there exists $k$ nodes which do not have a neighbor on one side, then there must exist $k-1$ disconnected components. Hence the number of disconnected components in $G(n,\frac{a \log n}{n})$ is $\omega(1)$.
\end{proof}


\end{document}

%% file: rgg-one.tex
\section{Connectivity of one dimensional Random Annulus graphs: Details}
\label{sec:VRG-detail}
In this section, we prove the necessary and sufficient condition for connectivity of ${\rm RAG}^\ast_1$ in full details. We provide the necessary conditions for connectivity in Theorem \ref{thm:rag} for ${\rm RAG}^\ast_1$ because of the equivalence between ${\rm RAG}^\ast_1$ and ${\rm RAG}_1$.
\subsection{Necessary condition for connectivity of ${\rm RAG}^\ast_1$}

\begin{theorem*}[${\rm RAG}^\ast_1$ connectivity lower bound in Theorem \ref{thm:rag}]
If $a <	1$ or $a-b < 0.5$, the random annulus graph ${\rm RAG}^\ast_1(n,[\frac{b\log n}{n},\frac{a\log n}{n}])$  is not connected
with  probability $1-o(1)$ . 
\end{theorem*}
\begin{proof}
First of all,  it is known that ${\rm RAG}^\ast_1(n,[0,\frac{a\log n}{n}])$ is not connected  with  high probability when $a <1$ \citep{muthukrishnan2005bin,penrose2003random}. Therefore  ${\rm RAG}^\ast_1(n,[\frac{b\log n}{n},\frac{a\log n}{n}])$ must not be connected with high probability when $a<1$ as the connectivity interval is a strict subset of the previous case, and ${\rm RAG}^\ast_1(n,[\frac{b\log n}{n},\frac{a\log n}{n}])$ can be obtained from ${\rm RAG}^\ast_1(n,[0,\frac{a\log n}{n}])$ by deleting all the edges that has the two corresponding random variables separated by distance less than $\frac{b\log n}{n}$.

Next we will show that if $a-b < 0.5$ then there exists an isolated vertex with high probability. It would be easier to think of each vertex as a uniform random point in $[0,1]$. 
Define an indicator variable $A_{u}$ for every node $u$ which is 1 when node $u$ is isolated and $0$ otherwise. We have,
$$\Pr(A_{u}=1)=\bigg(1-\frac{2(a-b)\log n}{n}\bigg)^{n-1}.$$ 
Define $A=\sum_{u} A_{u}$, and hence
$$\avg[A]=n\Big(1-\frac{2(a-b)\log n}{n} \Big)^{n-1}= n^{1-2(a-b)-o(1)}.$$ Therefore, when $a-b < 0.5$, $\avg[A]=\Omega(1)$. To prove this statement with high probability we can show that the variance of $A$ is bounded. Since $A$ is a sum of indicator random variables, we have that 
$${\rm Var}(A) \le \avg[A]+\sum_{u \neq v} {\rm Cov}(A_u,A_v)=\avg[A]+\sum_{u \neq v} (\Pr(A_u=1 \cap A_v=1)-\Pr(A_u=1)\Pr(A_v=1)).$$ 
Now, consider the scenario when the  vertices $u$ and $v$ are at a distance more than $\frac{2a \log n}{n}$ apart (happens with probability $1-\frac{4a\log n}{n})$. Then the region in $[0,1]$ that is between  distances $\frac{b \log n}{n}$ and $\frac{a\log n}{n}$ from both of the vertices is empty and therefore $\Pr(A_u=1 \cap A_v=1) =  
\Big(1- \frac{4(a-b) \log n}{n}\Big)^{n-2}.$
When the vertices  are within distance $\frac{2a \log n}{n}$ of one another, then
$
\Pr(A_u=1 \cap A_v=1) \le \Pr(A_u=1).
$
Therefore,
\begin{align*}
\Pr(A_u=1 \cap A_v=1) \le &(1-\frac{4a\log n}{n}) \Big(1- \frac{4(a-b) \log n}{n}\Big)^{n-2} + \frac{4a\log n}{n}\Pr(A_u=1)\\ 
&\le  (1-\frac{4a\log n}{n}) n^{-4(a-b)+o(1)}+ \frac{4a\log n}{n} n^{-2(a-b)+o(1)}.
\end{align*}
Consequently for large enough $n$,
\begin{align*}
\Pr(A_u=1 \cap A_v=1)-\Pr(A_u=1)\Pr(A_v=1) &\le (1-\frac{4a\log n}{n}) n^{-4(a-b)+o(1)} \\+ \frac{4a\log n}{n} n^{-2(a-b)+o(1)}  -& n^{-4(a-b) +o(1)} 
\le  \frac{8a\log n}{n}\Pr(A_u=1).
\end{align*}
Now,
$$
{\rm Var}(A) \le \avg[A] + \binom{n}{2}\frac{8a\log n}{n}\Pr(A_u=1) \le \avg[A](1+ 4a\log n).
$$
By using Chebyshev bound, with probability at least $1-\frac{1}{\log n}$, 
$$A >n^{1-2(a-b)}-\sqrt{n^{1-2(a-b)}(1+4a\log n)\log n},$$
which imply for $a-b <0.5$, there will exist isolated nodes with high probability.
\end{proof}

\subsection{Proof of Lemma \ref{lem:lemma1} and Lemma \ref{lem:lemma2}}

\begin{lemma*}[Restatement of Lemma \ref{lem:lemma1}]
A set of vertices $\cC \subseteq V$ is called a cover of $[0,1]$, if for any point $y$ in $[0,1]$ there exists a vertex $v\in \cC$ such that $d(v,y) \le \frac{a\log n}{2n}$. If $a-b >0.5$ and $a >1$, then a random annulus graph ${\rm RAG}^\ast_1(n, [\frac{b\log n}{n}, \frac{a\log n}{n}])$ is a union of cycles  
such that every cycle forms a cover of $[0,1]$  (see Figure~\ref{fig:arr0}) with  probability $1-o(1)$.
\end{lemma*}
 
 \begin{proof}[Proof of Lemma~\ref{lem:lemma1}]
 The proof of this lemma is somewhat easily explained if we consider a weaker result (a stronger condition) with $a -b >2/3$. Let us first briefly describe this case.

Consider a node $u$ and assume without loss of generality that the position of $u$ is $0$ (i.e. $X_u=0$). Associate four indicator $\{0,1\}$-random variables $A_u^i, i =1,2,3,4$ which take the value of $1$ if and only if there does not exist any node $x$ such that 
\begin{enumerate}
\item $d(u,x) \in [b\frac{\log n}{n},a\frac{\log n}{n}] \cup [0,\frac{a-b}{2}\frac{\log n}{n}]\} \text{  for  } i=1$
\item $d(u,x) \in [b\frac{\log n}{n},a\frac{\log n}{n}] \cup [\frac{-a-b}{2}\frac{\log n}{n},-b\frac{\log n}{n}]\}  \text{  for  } i=2$
\item $d(u,x) \in [-a\frac{\log n}{n},-b\frac{\log n}{n}] \cup [\frac{-a+b}{2}\frac{\log n}{n},0]\}  \text{  for  } i=3$
\item $d(u,x) \in [-a\frac{\log n}{n},-b\frac{\log n}{n}] \cup [b\frac{\log n}{n},\frac{a+b}{2}\frac{\log n}{n}]\}  \text{  for  } i=4.$
\end{enumerate}
The intervals representing these random variables are shown in Figure~\ref{fig:four}.

Notice that $\Pr(A^i_u=1)=\max\{\Big(1-1.5(a-b)\frac{\log n}{n}\Big)^{n-1},\Big(1-a\frac{\log n}{n}\Big)^{n-1}\}$ and therefore $\sum_{i,u} \avg A^i_u \approx 4\max\{n^{1-1.5(a-b)},n^{1-a}\}$.  This means that for $a-b \ge 0.67$ and $a \ge 1$, $\sum_{i,u} \avg A^i_u=o(1)$.  Hence there exist vertices in all the regions described above for every node $u$ with high  probability.
 
Now, $A_u^1$ and $A_u^2$ being zero implies that either there is a vertex in $ [b\frac{\log n}{n},a\frac{\log n}{n}] $ or there exists two vertices $v_1,v_2$ in $ [0,\frac{a-b}{2}\frac{\log n}{n}]$ and $[\frac{-a-b}{2}\frac{\log n}{n},-b\frac{\log n}{n}]$ respectively (see, Figure~\ref{fig:four}). In the second case, $u$ is connected to $v_2$ and $v_2$ is connected to $v_1$. Therefore $u$ has nodes on left ($v_2$) and right ($v_1$) and $u$ is connected to both of them through one hop in the graph. 

Similarly, $A_u^3$ and $A_u^4$ being zero implies that either there exists a vertex in $[-a\frac{\log n}{n},-b\frac{\log n}{n}]$ or again $u$ will have vertices on left and right and will be connected to them. So, when all the four $A_u^i, i=1,2,3,4$ are zero together:
{
\begin{itemize}
    \item $A_u^1=A_u^2=0$ implies there is a neighbor of $u$ on either sides or there is a single node in $[b\frac{\log n}{n},a\frac{\log n}{n}]$
    \item $A_u^3=A_u^4=0$ implies there is a neighbor of $u$ on either sides or there is a single node in $[-a\frac{\log n}{n},-b\frac{\log n}{n}]$
\end{itemize}
This shows that when $A_u^1=A_u^2=0$ and  $A_u^3=A_u^4=0$ guarantee a node on only one side of $u$, there are nodes in $[b\frac{\log n}{n},a\frac{\log n}{n}] $ and $[-a\frac{\log n}{n},-b\frac{\log n}{n}]$.} But in that case $u$ has direct neighbors on both its left and right. 
We can conclude that every vertex $u$ is connected to  a vertex  $v$ on  its right and a vertex  $w$ on its left  such that $d(u,v) \in [0,a \frac{\log n}{n}]$ and  $d(u,w) \in [-a\frac{\log n}{n},0]$; therefore every vertex is part of a cycle that covers $[0,1]$.

 \begin{figure}
   \centering
\begin{tikzpicture}[thick, scale=0.9]
\draw [gray,<->] (0,1.5) -- (0.5,1.5)node at (1,1.5) {$\frac{a-b}{2}$};;
\draw[pattern=north west lines,pattern color=red][preaction={fill=blue!20}] (3.5,1.9) -- (4.5,1.9) -- (4.5,1.7) -- (3.5,1.7) -- cycle node at (4.8,1.8) {$A_u^3$};;
\draw[pattern=vertical lines,pattern color=red][preaction={fill=blue!50}]  (3.5,1.4) -- (4.5,1.4) -- (4.5,1.2) -- (3.5,1.2) -- cycle node at (4.8,1.2) {$A_u^4$};;
\draw[pattern=dots,pattern color=blue] [preaction={fill=gray}]  (1.5,1.9) -- (2.5,1.9) -- (2.5,1.7) -- (1.5,1.7) -- cycle node at (2.8,1.8) {$A_u^1$};;
\draw[pattern=grid,pattern color=violet][preaction={fill=red!50}]   (1.5,1.4) -- (2.5,1.4) -- (2.5,1.2) -- (1.5,1.2) -- cycle node at (2.8,1.2) {$A_u^2$};;

\draw[pattern=dots,pattern color=blue][preaction={fill=gray}]  (3.5,0) -- (4.5,0) -- (4.5,-0.3) -- (3.5,-0.3) -- cycle;
\draw[pattern=grid,pattern color=violet] [preaction={fill=red!50}]  (3.5,-0.3) -- (4.5,-0.3) -- (4.5,-0.6) -- (3.5,-0.6) -- cycle;
\draw[pattern=north west lines,pattern color=red][preaction={fill=blue!20}]  (-3.5,0) -- (-4.5,0) -- (-4.5,-0.3) -- (-3.5,-0.3) -- cycle;
\draw[pattern=vertical lines,pattern color=red][preaction={fill=blue!50}]   (-3.5,-0.3) -- (-4.5,-0.3) -- (-4.5,-0.6) -- (-3.5,-0.6) -- cycle;
\draw[pattern=dots,pattern color=blue][preaction={fill=gray}]  (0,0) -- (.5,0) -- (.5,-0.3) -- (0,-0.3) -- cycle;
\draw[pattern=north west lines,pattern color=red][preaction={fill=blue!20}]  (0,0) -- (-.5,0) -- (-.5,-0.3) -- (0,-0.3) -- cycle;
\draw[pattern=vertical lines,pattern color=red][preaction={fill=blue!50}]   (3.5,-0.6) -- (4,-0.6) -- (4,-0.9) -- (3.5,-0.9) -- cycle;
\draw[pattern=grid,pattern color=violet] [preaction={fill=red!50}]  (-3.5,0) -- (-3,0) -- (-3,-0.3) -- (-3.5,-0.3) -- cycle;

\draw [gray,<->] (3.5,-1.3) -- (4,-1.3);
\draw [gray,<->] (0,-1.3) -- (0.5,-1.3);
\draw [gray,<->] (0,-1.3) -- (-.5,-1.3);
\draw [gray,<->] (-3.5,-1.3) -- (-3,-1.3);

\draw [gray,<->] (-5,0) -- (5,0);
\draw [gray,dashed] (3.5,0) -- (3.5,-1.4);
\draw [gray,dashed] (0.5,0) -- (0.5,-1.4);
\draw [gray,dashed] (0,0) -- (0,-1.4);
\draw [gray,dashed] (-0.5,0) -- (-0.5,-1.4);
\draw [gray,dashed] (-3.5,0) -- (-3.5,-1.4);
\draw [gray,dashed] (-3,0) -- (-3,-1.4);
\draw [gray,dashed] (4,0) -- (4,-1.4);
\draw node[anchor=south] at (3.5,0) {$b$};
\draw node[anchor=south] at (4.5,0) {$a$};
\draw node[anchor=south] at (-3.5,0) {$-b$};
\draw node[anchor=south] at (-4.5,0) {$-a$};

%
%
%
%
\filldraw [black] (0,0) circle (3pt)node[anchor=south] at (0,0) {$u$};
\end{tikzpicture}
\caption{Representation of four different random variables for Lemma \ref{lem:lemma1}. \label{fig:four}}
\end{figure}
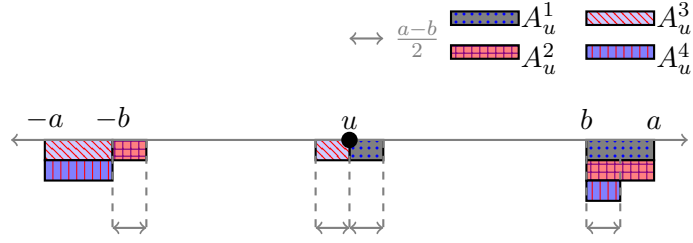

We can  now extend this proof to the case when $a-b> 0.5.$

Let $c$ be large number to  be chosen specifically later.
Consider a node $u$ and assume that the position of $u$ is $0$. Now consider the {four different regions $[-a\frac{\log n}{n},-b\frac{\log n}{n}]$, $[-(a-b)\frac{\log n}{n},0]$, $[b\frac{\log n}{n},a\frac{\log n}{n}]$ and $[0,a-b\frac{\log n}{n}]$ around $u$ each divided into $L\equiv 2^{c}$ patches (intervals) of size $\theta=\frac{a-b}{2^c}$  in  the following way}:
\begin{enumerate}
\item $I_u^i=[\frac{(-a + (i-1)\theta )\log n}{n},\frac{(-a + i\theta)  \log n}{n}]$
\item $J_u^i=[\frac{(-(a-b) + (i-1)\theta )\log n}{n},\frac{(-(a-b) + i\theta)  \log n}{n}]$
\item $K_u^i= [\frac{(b + (i-1)\theta )\log n}{n},\frac{(b + i\theta)  \log n}{n}]$
\item $M_u^i= [\frac{( (i-1)\theta )\log n}{n},\frac{ i\theta  \log n}{n}]$
\end{enumerate}
where $i=1,2,3,\ldots,L$. Note that any vertex in $\cup I_u^i \cup K_u^i$ is connected to $u$. See, Figure~\ref{fig:patch} for a depiction.

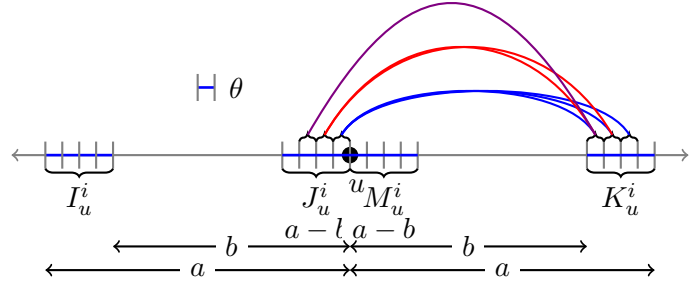
\begin{figure}
\centering
\begin{tikzpicture}[thick, scale=0.9]
\draw [gray,<->] (-5,0) -- (5,0);\filldraw [black] (0,0) circle (3pt)node[anchor=south] at (0.1,-0.7) {$u$};
\draw[<->] (0,-1.1) -- node[midway,fill=white] {$a-b$} (-1,-1.1);
\draw[<->] (0,-1.1) -- node[midway,fill=white] {$a-b$} (1,-1.1);
\draw[<->] (0,-1.35) -- node[midway,fill=white] {$b$} (-3.5,-1.35);
\draw[<->] (0,-1.7) -- node[midway,fill=white] {$a$} (-4.5,-1.7);
\draw[<->] (0,-1.35) -- node[midway,fill=white] {$b$} (3.5,-1.35);
\draw[<->] (0,-1.7) -- node[midway,fill=white] {$a$} (4.5,-1.7);

\draw [blue, line width=1pt] (-4.5,0) -- (-4.25,0);
\draw [blue, line width=1pt] (-4.25,0) -- (-4,0);
\draw [blue, line width=1pt] (-4,0) -- (-3.75,0);
\draw [blue, line width=1pt] (-3.75,0) -- (-3.5,0);

\draw [blue, line width=1pt] (-2.00,1) -- (-2.25,1) node [black,midway,xshift=0.4cm,yshift=0cm] {$\theta$};;
\draw [gray] (-2.00,1.2) -- (-2.00,0.8);
\draw [gray] (-2.25,1.2) -- (-2.25,0.8);

\draw [gray] (4.5,0.2) -- (4.5,-.2);
\draw [gray] (4.25,0.2) -- (4.25,-.2);
\draw [gray] (4,0.2) -- (4,-.2);
\draw [gray] (3.75,0.2) -- (3.75,-.2);
\draw [gray] (3.5,0.2) -- (3.5,-.2);

\draw [gray] (-4.5,0.2) -- (-4.5,-.2);
\draw [gray] (-4.25,0.2) -- (-4.25,-.2);
\draw [gray] (-4,0.2) -- (-4,-.2);
\draw [gray] (-3.75,0.2) -- (-3.75,-.2);
\draw [gray] (-3.5,0.2) -- (-3.5,-.2);

\draw [gray] (-1,0.2) -- (-1,-.2);
\draw [gray] (-.75,0.2) -- (-.75,-.2);
\draw [gray] (-0.5,0.2) -- (-0.5,-.2);
\draw [gray] (-.25,0.2) -- (-.25,-.2);
\draw [gray] (0,0.2) -- (0,-.2);
\draw [gray] (1,0.2) -- (1,-.2);
\draw [gray] (.75,0.2) -- (.75,-.2);
\draw [gray] (0.5,0.2) -- (0.5,-.2);
\draw [gray] (.25,0.2) -- (.25,-.2);
\draw [decorate,decoration={brace,amplitude=3pt,mirror},xshift=0pt,yshift=0pt](3.5,-0.2) -- (4.5,-0.2) node [black,midway,xshift=0cm,yshift=-0.4cm] {$K_u^i$};
\draw [decorate,decoration={brace,amplitude=3pt},xshift=0pt,yshift=0pt](0,-0.2) -- (-1,-0.2) node [black,midway,xshift=0cm,yshift=-0.4cm] {$J_u^i$};
\draw [decorate,decoration={brace,amplitude=3pt,mirror},xshift=0pt,yshift=0pt](0,-0.2) -- (1,-0.2) node [black,midway,xshift=0cm,yshift=-0.4cm] {$M_u^i$};

\draw [decorate,decoration={brace,amplitude=3pt},xshift=0pt,yshift=0pt](-3.5,-0.2) -- (-4.5,-0.2) node [black,midway,xshift=0cm,yshift=-0.4cm] {$I_u^i$};

\draw [blue, line width=1pt] (-1,0) -- (-0.75,0);
\draw [blue, line width=1pt] (-0.75,0) -- (-0.5,0);
\draw [blue, line width=1pt] (-0.5,0) -- (-0.25,0);
\draw [blue, line width=1pt] (-0.25,0) -- (0,0);

\draw [blue, line width=1pt] (1,0) -- (0.75,0);
\draw [blue, line width=1pt] (0.75,0) -- (0.5,0);
\draw [blue, line width=1pt] (0.5,0) -- (0.25,0);
\draw [blue, line width=1pt] (0.25,0) -- (0,0);

\draw [blue, line width=1pt] (4.5,0) -- (4.25,0);
\draw [blue, line width=1pt] (4.25,0) -- (4,0);
\draw [blue, line width=1pt] (4,0) -- (3.75,0);
\draw [blue, line width=1pt] (3.75,0) -- (3.5,0);

\draw [decorate,decoration={brace,amplitude=3pt},xshift=0pt,yshift=0pt](3.5,0.2) -- (3.75,0.2);
\draw [decorate,decoration={brace,amplitude=3pt},xshift=0pt,yshift=0pt](3.75,0.2) -- (4,0.2);
\draw [decorate,decoration={brace,amplitude=3pt},xshift=0pt,yshift=0pt](4,0.2) -- (4.25,0.2);
\draw [decorate,decoration={brace,amplitude=3pt},xshift=0pt,yshift=0pt](-0.75,0.2) -- (-0.5,0.2);
\draw [decorate,decoration={brace,amplitude=3pt},xshift=0pt,yshift=0pt](-0.5,0.2) -- (-0.25,0.2);
\draw [decorate,decoration={brace,amplitude=3pt},xshift=0pt,yshift=0pt](-0.25,0.2) -- (0,0.2);

\draw[blue]    (-0.125,0.3) to[out=60,in=120,distance=1cm](3.625,0.3) ;
\draw[blue]    (-0.125,0.3) to[out=60,in=120,distance=1cm](3.875,0.3) ;
\draw[blue]    (-0.125,0.3) to[out=60,in=120,distance=1cm](4.125,0.3) ;

\draw[red]    (-0.375,0.3) to[out=60,in=120,distance=2cm](3.875,0.3) ;
\draw[red]    (-0.375,0.3) to[out=60,in=120,distance=2cm](3.625,0.3) ;

\draw[violet]    (-0.625,0.3) to[out=60,in=120,distance=3cm](3.625,0.3) ;

\end{tikzpicture}
\caption{Pictorial representation of $I_u^i, J_u^i, K_u^i, M_u^i$ and their connectivity as described in Lemma \ref{lem:lemma1}. The colored lines show the regions that are connected to each other.\label{fig:patch}}
\end{figure}

Consider a $\{0,1\}$-indicator random variable $X_u$ that is $1$ if and only if there {does not exist any node in a region formed by union of any $2L-1$ patches amongst the ones described above.} Notice that { when $a<2b$, the patches do not overlap and the total size of $2L-1$ patches is $\frac{2^{c+1}-1}{2^c}\frac{(a-b)\log n}{n}$ and when $a\geq 2b$, the patches can overlap and the total size of the $2L-1$ patches is going to be more than $\min\{\frac{2^{c+1}-1}{2^c}\frac{(a-b)\log n}{n},\frac{a \log n}{n}\}$. }
Since there are ${4L \choose 2L-1}\le n^{\frac{4L}{\log n}}$ possible regions that consists of $2L-1$ patches, 
\begin{align*}
 \sum_u \avg X_u &\leq n {4L \choose 2L-1}\Big(1-\min\{\frac{2^{c+1}-1}{2^c}\frac{(a-b)\log n}{n},\frac{a \log n}{n}\}\Big)^{n-1}\\
 & \le \max\{n^{1-\frac{2^{c+1}-1}{2^c}(a-b) +\frac{4L}{\log n}}, n^{1-a+\frac{4L}{\log n}}\}.
\end{align*}
At this point we can choose $c= c_n = o(\log n)$ such that $\lim_n c_n = \infty$.
 Hence when $a-b > \frac12$ 
 and $a > 1$, for every vertex $u$ there exists at least one patch amongst every $2L-1$ patches in $\cup I_u^i \cup J_u^j \cup K_u^k, i,j,k =1,2, \dots, L$ that contains a vertex.

Consider a collection of patches $\cup_iI_u^i \cup_j J_u^j,  i,j=1, 2, \dots, L$. We know that there exist two patches amongst these $I_u^i$s and $K_u^j$s that contain at least one vertices. If one of $I_u^i$s and one of $K_u^j$s contain two vertices, we found one neighbor of $u$ on both left and right directions (see, Figure~\ref{fig:patch}).

We consider the other case now.
Without loss of generality assume that there are no vertex in all $I_u^i$s and there exist at least two patches in $K_u^i$s that contain at least one vertex each. Hence, there exists at least one of $\{K_u^i \mid i\in \{1,2, \ldots, L-1\}\}$ that contains a vertex. Similarly, we can also conclude in this case that there exists at least one of $\{J_u^i \mid i\in \{2,3 \ldots, L\}\}$ which contain a node. Assume $J_u^{\phi}$ to be the left most patch in $\cup J_u^i \mid i \in \{1, 2, \dots, L\}$ that contains a vertex (see, Figure~\ref{fig:patch}) . From our previous observation, we can conclude that $\phi \ge 2$.

We can observe that any vertex in $J_u^j$ is connected to the vertices in patches $K_u^k,  \forall k<j$. 
This is because for two vertices $v\in J_u^j$ and $w\in K_u^k$, we have
\begin{align*}
d(v,w) &\ge \frac{(b + (k-1)\theta )\log n}{n} - \frac{(-(a-b) + j\theta)  \log n}{n}
= \frac{(a + (k-j-1)\theta )\log n}{n}; \\
d(v,w) &\le \frac{(b + k\theta)  \log n}{n} - \frac{(-(a-b) + (j-1)\theta )\log n}{n}
= \frac{(a + (k-j+1)\theta )\log n}{n}. 
\end{align*}


Consider a collection of $2L-1$ patches $\{\cup I_u^i  \cup J_u^j \cup K_u^k \mid i,j,k\in \{1,\dots, L\}, j>\phi, k\le \phi-1\}$ where $\phi \ge 2$. This is a collection of $2L-1$ patches out of which one must have a vertex and since none of $\{J_u^j \mid j>\phi\}$ and $I_u^i$ can contain a vertex, one of $\{K_u^k \mid k \le \phi -1\}$  must contain the vertex. Recall that the vertex in $J_u^{\phi}$ is connected to any node in $K_u^k$ for any $k\le \phi-1$ and therefore $u$ has a node to the right direction and left direction that are connected to $u$. Therefore every vertex is part of a cycle and each of the circles  covers $[0,1]$.
 \end{proof}

\begin{lemma*}[Restatement of Lemma \ref{lem:lemma2}]
Set two real numbers $k\equiv \lceil b/(a-b)\rceil+1$ and $\epsilon < \frac1{2k}$. In an ${\rm RAG}^\ast_1(n, [\frac{b\log n}{n}, \frac{a\log n}{n}]), 0 <b <a$, with  probability $1-o(1)$ there exists a vertex $u_0$ and  $k$ nodes $\{u_1, u_2 ,\ldots, u_k\}$ to the right of $u_0$ such that $d(u_0,u_i) \in [\frac{(i(a-b)-2i\epsilon)\log n}{n},\frac{(i(a-b)-(2i-1)\epsilon) \log n}{n}]$ and  $k$ nodes $\{v_1, v_2 ,\ldots, v_k\}$ to the right of $u_0$ such that $d(u_0,v_i) \in [\frac{((i(a-b)+b-(2i-1)\epsilon)\log n}{n},\frac{(i(a-b)+b -(2i-2)\epsilon)\log n}{n}]$, for $i =1,2,\ldots,k$. The arrangement of the vertices is shown  in Figure~\ref{fig:arr}.
\end{lemma*}
\begin{proof}[Proof of Lemma~\ref{lem:lemma2}]
Recall that we want to show that
 there exists a node $u_0$ and  $k$ nodes $\{u_1, u_2 ,\ldots, u_k\}$ to the right of $u_0$ such that $d(u_0,u_i) \in [\frac{(i(a-b)-2i\epsilon)\log n}{n},\frac{(i(a-b)-(2i-1)\epsilon) \log n}{n}]$ and exactly $k$ nodes $\{v_1, v_2 ,\ldots, v_k\}$ to the right of $u_0$ such that $d(u_0,v_i) \in [\frac{((i(a-b)+b-(2i-1)\epsilon)\log n}{n},\frac{(i(a-b)+b -(2i-2)\epsilon)\log n}{n}]$, for $i =1,2,\ldots,k$ and $\epsilon$ is a constant less than $\frac1{2k}$ (see Figure~\ref{fig:arr} for a depiction). 
 Let $A_u$ be an indicator $\{0,1\}$-random variable for every node $u$ which is $1$ if $u$ satisfies the above conditions and $0$ otherwise. We will show $\sum_{u} A_u \ge 1$ with high probability. 
 
 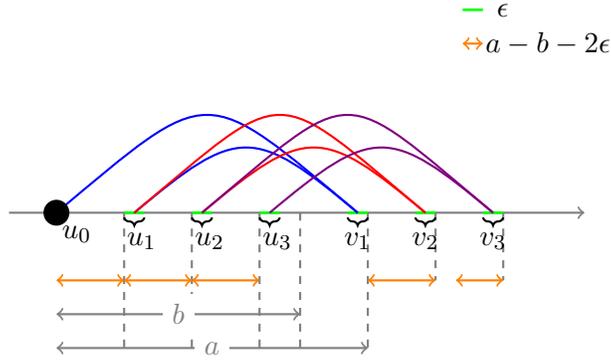
\begin{figure}
 \centering
\begin{tikzpicture}[thick, scale=0.9]

\draw [gray,->] (-1,0) -- (7.5,0);
\draw [gray,dashed] (0.7,0) -- (0.7,-2);
\draw [gray,dashed] (1.7,0) -- (1.7,-2);
\draw [gray,dashed] (2.7,0) -- (2.7,-2);
\draw [gray,dashed] (3.3,0) -- (3.3,-2);
\draw [gray,dashed] (4.3,0) -- (4.3,-2);
\draw [gray,dashed] (5.3,0) -- (5.3,-1);
\draw [gray,dashed] (6.3,0) -- (6.3,-1);

\draw [green, line width=1pt] (4,0) -- (4.3,0);
\draw [decorate,decoration={brace,amplitude=3pt,mirror},xshift=0pt,yshift=-2pt](4.,0) -- (4.3,0) node [black,midway,xshift=0cm,yshift=-0.31cm] {$v_1$};

\draw [green, line width=1pt] (5,0) -- (5.3,0);
\draw [decorate,decoration={brace,amplitude=3pt,mirror},xshift=0pt,yshift=-2pt](5,0) -- (5.3,0) node [black,midway,xshift=0cm,yshift=-0.31cm] {$v_2$};

\draw [green, line width=1pt] (6,0) -- (6.3,0);
\draw [decorate,decoration={brace,amplitude=3pt,mirror},xshift=0pt,yshift=-2pt](6,0) -- (6.3,0) node [black,midway,xshift=0cm,yshift=-0.31cm] {$v_3$};

\draw [green, line width=1pt] (0.7,0) -- (1,0);
\draw [decorate,decoration={brace,amplitude=3pt,mirror},xshift=0pt,yshift=-2pt](0.7,0) -- (1,0) node [black,midway,xshift=0.1cm,yshift=-0.31cm] {$u_1$};

\draw [green, line width=1pt] (1.7,0) -- (2,0);
\draw [decorate,decoration={brace,amplitude=3pt,mirror},xshift=0pt,yshift=-2pt](1.7,0) -- (2,0) node [black,midway,xshift=0.1cm,yshift=-0.31cm] {$u_2$};

\draw [green, line width=1pt] (2.7,0) -- (3,0);
\draw [decorate,decoration={brace,amplitude=3pt,mirror},xshift=0pt,yshift=-2pt](2.7,0) -- (3,0) node [black,midway,xshift=0.1cm,yshift=-0.31cm] {$u_3$};

\draw [orange,<->] (-0.3,-1) -- (0.7,-1);
\draw [orange,<->] (0.7,-1) -- (1.7,-1);
\draw [orange,<->] (1.7,-1) -- (2.7,-1);
\draw [orange,<->] (4.3,-1) -- (5.3,-1);
\draw [orange,<->] (5.6,-1) -- (6.3,-1);

\draw[gray,<->] (-0.3,-2) -- node[midway,fill=white] {${a}$} (4.3,-2);
\draw[gray,<->] (-0.3,-1.5) -- node[midway,fill=white] {${b}$} (3.3,-1.5);

\draw [green, line width=1pt] (5.7,3) -- (6,3)node [black,midway,xshift=0.4cm,yshift=0cm] {$\epsilon$};
\draw [orange,<->] (5.7,2.5) -- (6,2.5)node [black,midway,xshift=1cm,yshift=0cm] {$a-b-2\epsilon$};


\draw[blue]    (-0.3,0) to[out=40,in=140,distance=3cm](4.15,0) ;
\draw[blue]    (0.85,0) to[out=40,in=140,distance=2cm](4.15,0) ;
\draw[red]    (0.85,0) to[out=40,in=140,distance=3cm](5.15,0) ;
\draw[red]    (1.85,0) to[out=40,in=140,distance=2cm](5.15,0) ;
\draw[violet]    (1.85,0) to[out=40,in=140,distance=3cm](6.15,0) ;
\draw[violet]    (2.85,0) to[out=40,in=140,distance=2cm](6.15,0) ;

\filldraw [black] (-0.3,0) circle (5pt)node[anchor=south] at (0,-0.6) {$u_0$};
\end{tikzpicture}
\caption{The location of $u_i$ and $v_i$ relative to $u$ scaled by $\frac{\log n}{n}$ in Lemma \ref{lem:lemma2}. Edges stemming put of  $v_1,v_2, v_3$ are shown as blue, red and violet respectively. \label{fig:arr}}
\end{figure}

 We  have,
\begin{align*}
\Pr(A_u=1) & = n(n-1)\dots (n-(2k-1))\Big(\frac{\epsilon \log n}{n}\Big)^{2k} \Big(1-2k\epsilon \frac{\log n}{n}\Big)^{n-2k}\\
& = c_0 n^{-2k\epsilon} (\epsilon \log n)^{2k}  \prod_{i=0}^{2k-1} (1-i/n)\\
&= c_1 n^{-2k\epsilon} (\epsilon \log n)^{2k}
\end{align*}
where $c_0,c_1$ are just absolute constants independent of $n$ (recall $k$ is a constant). 
Hence,
\begin{align*}
\sum_{u} \avg A_u= c_1 n^{1-2k\epsilon} (\epsilon \log n)^{2k} \ge 1
\end{align*} 
as long as $\epsilon \leq \frac{1}{2k}$.  Now, in order to prove $\sum_{u}  A_u\ge 1$ with high probability, we will show that the variance of $\sum_{u}  A_u$  is bounded from above.  This calculation is very similar to the one in the proof of the connectivity lower bound in  Theorem~\ref{thm:rag}. 
Recall that if $A =\sum_{u}  A_u$ is a sum of indicator random variables, we must have 
$${\rm Var}(A) \le \avg[A]+\sum_{u \neq v} {\rm Cov}(A_u,A_v)=\avg[A]+\sum_{u \neq v} \Pr(A_u=1 \cap A_v=1)-\Pr(A_u=1)\Pr(A_v=1).$$ 
Now first consider the case when vertices $u$ and $v$ are at a distance of at least $\frac{2(a+b) \log n}{n}$ apart (happens with probability $1- \frac{4(a+b)\log n}{n}$). 
Then the region in $[0,1]$ that is within distance $\frac{(a+b)\log n}{n}$ from both $u$ and $v$ is the empty-set. In this case, $\Pr(A_u=1 \cap A_v=1) = n(n-1)\dots (n-(4k-1))\Big(\frac{\epsilon \log n}{n}\Big)^{4k} \Big(1-4k\epsilon \frac{\log n}{n}\Big)^{n-4k} = c_2 n^{-4k\epsilon} (\epsilon \log n)^{4k},$ where $c_2$ is a constant.

In all other cases, $\Pr(A_u=1 \cap A_v=1) \le \Pr(A_u =1)$.
Therefore,
\begin{align*}
\Pr(A_u=1 \cap A_v=1)\leq \Big(1-\frac{4(a+b) \log n}{n}\Big) c_2 n^{-4k\epsilon} (\epsilon \log n)^{4k}+ \frac{4(a+b) \log n}{n} c_1n^{-2k\epsilon} (\epsilon \log n)^{2k} 
\end{align*}
and
\begin{align*}
{\rm Var(A)} &\le c_1n^{1-2k\epsilon} (\epsilon \log n)^{2k} +{n \choose 2}\Big(\Pr(A_u=1 \cap A_v=1)-\Pr(A_u=1)\Pr(A_v=1)\Big) \\
&\le c_1n^{1-2k\epsilon} (\epsilon \log n)^{2k}+ c_3 n^{1-2k\epsilon} (\log n)^{2k+1}\\
& \le c_4  n^{1-2k\epsilon} (\log n)^{2k+1}
\end{align*}
where $c_3,c_4$ are constants. Again invoking Chebyshev's inequality, with probability at least $1-\frac{1}{\log n}$ 
$$
A > c_1n^{1-2k\epsilon} (\epsilon \log n)^{2k} - \sqrt{c_4  n^{1-2k\epsilon} (\log n)^{2k+2}}.
$$  

\end{proof}

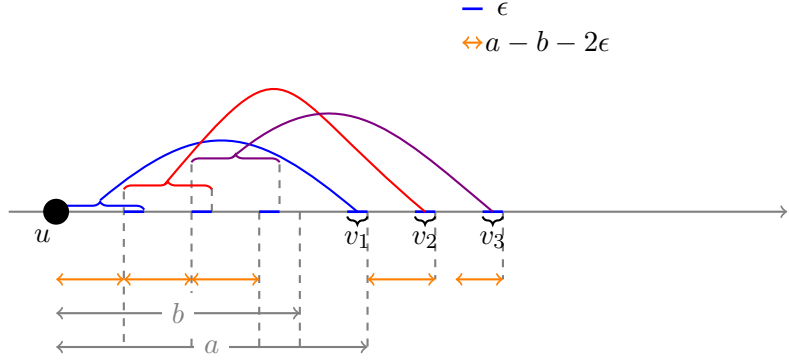
\begin{figure}
\centering
\begin{tikzpicture}[thick, scale=0.9]

\draw [gray,->] (-1,0) -- (10.5,0);
\draw [gray,dashed] (0.7,0.3) -- (0.7,-2);
\draw [gray,dashed] (1.7,0.7) -- (1.7,-2);
\draw [gray,dashed] (2,0.35) -- (2,0);
\draw [gray,dashed] (3,0.7) -- (3,0);
\draw [gray,dashed] (2.7,0) -- (2.7,-2);
\draw [gray,dashed] (3.3,0) -- (3.3,-2);
\draw [gray,dashed] (4.3,0) -- (4.3,-2);
\draw [gray,dashed] (5.3,0) -- (5.3,-1);
\draw [gray,dashed] (6.3,0) -- (6.3,-1);

\draw [blue, line width=1pt] (4,0) -- (4.3,0);
\draw [decorate,decoration={brace,amplitude=3pt,mirror},xshift=0pt,yshift=-2pt](4.,0) -- (4.3,0) node [black,midway,xshift=0cm,yshift=-0.31cm] {$v_1$};

\draw [blue, line width=1pt] (5,0) -- (5.3,0);
\draw [decorate,decoration={brace,amplitude=3pt,mirror},xshift=0pt,yshift=-2pt](5,0) -- (5.3,0) node [black,midway,xshift=0cm,yshift=-0.31cm] {$v_2$};

\draw [blue, line width=1pt] (6,0) -- (6.3,0);
\draw [decorate,decoration={brace,amplitude=3pt,mirror},xshift=0pt,yshift=-2pt](6,0) -- (6.3,0) node [black,midway,xshift=0cm,yshift=-0.31cm] {$v_3$};

\draw [blue, line width=1pt] (0.7,0) -- (1,0);
\draw [blue,decorate,decoration={brace,amplitude=3pt},xshift=0pt,yshift=-2pt](-0.3,0.1) -- (1,0.1);
\draw [red,decorate,decoration={brace,amplitude=3pt},xshift=0pt,yshift=-2pt](.7,0.4) -- (2,0.4);
\draw [violet,decorate,decoration={brace,amplitude=3pt},xshift=0pt,yshift=-2pt](1.7,0.8) -- (3,0.8);

\draw [blue, line width=1pt] (1.7,0) -- (2,0);

\draw [blue, line width=1pt] (2.7,0) -- (3,0);

\draw [orange,<->] (-0.3,-1) -- (0.7,-1);
\draw [orange,<->] (0.7,-1) -- (1.7,-1);
\draw [orange,<->] (1.7,-1) -- (2.7,-1);
\draw [orange,<->] (4.3,-1) -- (5.3,-1);
\draw [orange,<->] (5.6,-1) -- (6.3,-1);

\draw[gray,<->] (-0.3,-2) -- node[midway,fill=white] {${a}$} (4.3,-2);
\draw[gray,<->] (-0.3,-1.5) -- node[midway,fill=white] {${b}$} (3.3,-1.5);

\draw [blue, line width=1pt] (5.7,3) -- (6,3)node [black,midway,xshift=0.4cm,yshift=0cm] {$\epsilon$};
\draw [orange,<->] (5.7,2.5) -- (6,2.5)node [black,midway,xshift=1cm,yshift=0cm] {$a-b-2\epsilon$};


\draw[blue]    (0.35,0.17) to[out=40,in=140,distance=2cm](4.15,0) ;
\draw[red]    (1.3,0.4) to[out=50,in=140,distance=3cm](5.15,0) ;

\draw[violet]    (2.3,0.8) to[out=40,in=140,distance=2cm](6.15,0) ;

\filldraw [black] (-0.3,0) circle (5pt)node[anchor=south] at (-0.5,-0.6) {$u$};
\end{tikzpicture}
\caption{The line segments where $v_1,v_2,v_3$ can have neighbors (scaled by $\frac{\log n}{n}$) in the proof of Lemma \ref{lem:lemma2}. The point $t$ has to lie in one of these regions.\label{fig:arr1d}}
\end{figure}


\begin{proof}[Proof of Corollary \ref{cor:extra}]
Consider a node $u$ and assume that the position of $u$ is $0$. Associate a random variable $A_u^i$ for $i \in \{1,2,3,4\}$ which takes the value of $1$ when there does not exist any node $x$ such that 

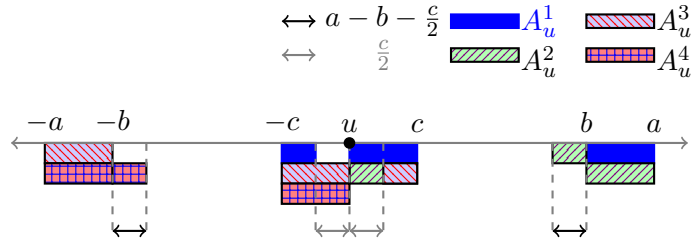
\begin{figure}[h]
\centering
\begin{tikzpicture}[thick, scale=0.9]
\draw [black,<->] (-0.5,1.8) -- (-1,1.8)node at (0.5,1.8) {$a-b-\frac{c}{2}$};;
\draw [gray,<->] (-0.5,1.3) -- (-1,1.3)node at (0.5,1.3) {$\frac{c}{2}$};;
\draw[pattern=north west lines,pattern color=red][preaction={fill=blue!20}] (3.5,1.9) -- (4.5,1.9) -- (4.5,1.7) -- (3.5,1.7) -- cycle node at (4.8,1.8) {$A_u^3$};;
\draw[pattern=grid,pattern color=blue][preaction={fill=red!50}]   (3.5,1.4) -- (4.5,1.4) -- (4.5,1.2) -- (3.5,1.2) -- cycle node at (4.8,1.2) {$A_u^4$};;
\draw[fill=blue,blue]   (1.5,1.9) -- (2.5,1.9) -- (2.5,1.7) -- (1.5,1.7) -- cycle node at (2.8,1.8) {$A_u^1$};;
\draw[pattern=north east lines,pattern color=violet][preaction={fill=green!30}] (1.5,1.4) -- (2.5,1.4) -- (2.5,1.2) -- (1.5,1.2) -- cycle node at (2.8,1.2) {$A_u^2$};;

\draw[fill=blue,blue]  (3.5,0) -- (4.5,0) -- (4.5,-0.3) -- (3.5,-0.3) -- cycle;
\draw[fill=blue,blue]  (0,0) -- (1,0) -- (1,-0.3) -- (0,-0.3) -- cycle;
\draw[fill=blue,blue]  (-0.5,0) -- (-1,0) -- (-1,-0.3) -- (-0.5,-0.3) -- cycle;

\draw[pattern=north east lines,pattern color=violet][preaction={fill=green!30}] (3.5,-0.6) -- (4.5,-0.60) -- (4.5,-0.3) -- (3.5,-0.3) -- cycle;
\draw[pattern=north east lines,pattern color=violet][preaction={fill=green!30}]  (0,-0.6) -- (1,-0.6) -- (1,-0.3) -- (0,-0.3) -- cycle;
\draw[pattern=north east lines,pattern color=violet][preaction={fill=green!30}]  (3,0) -- (3.5,0) -- (3.5,-0.3) -- (3,-0.3) -- cycle;

\draw[pattern=north west lines,pattern color=red][preaction={fill=blue!20}] (-3.5,0) -- (-4.5,0) -- (-4.5,-0.3) -- (-3.5,-0.3) -- cycle;
\draw[pattern=north west lines,pattern color=red][preaction={fill=blue!20}] (0,-0.60) -- (-1,-0.60) -- (-1,-0.3) -- (0,-0.3) -- cycle;
\draw[pattern=north west lines,pattern color=red][preaction={fill=blue!20}] (0.5,-0.6) -- (1,-0.6) -- (1,-0.3) -- (0.5,-0.3) -- cycle;

\draw[pattern=grid,pattern color=blue][preaction={fill=red!50}](-3.5,-0.6) -- (-4.5,-0.6) -- (-4.5,-0.3) -- (-3.5,-0.3) -- cycle;
\draw[pattern=grid,pattern color=blue][preaction={fill=red!50}] (0,-0.6) -- (-1,-0.6) -- (-1,-0.9) -- (0,-0.9) -- cycle;
\draw[pattern=grid,pattern color=blue][preaction={fill=red!50}]  (-3,-0.3) -- (-3.5,-0.3) -- (-3.5,-0.6) -- (-3,-0.6) -- cycle;

\draw [gray,<->] (0,-1.3) -- (0.5,-1.3);
\draw [gray,<->] (0,-1.3) -- (-.5,-1.3);
\draw [black,<->] (-3.50,-1.3) -- (-3,-1.3);
\draw [black,<->] (3.50,-1.3) -- (3,-1.3);

\draw [gray,<->] (-5,0) -- (5,0);
\draw [gray,dashed] (0.5,0) -- (0.5,-1.4);
\draw [gray,dashed] (0,0) -- (0,-1.4);
\draw [gray,dashed] (-0.5,0) -- (-0.5,-1.4);
\draw [gray,dashed] (-3,0) -- (-3,-1.4);
\draw [gray,dashed] (-3.5,0) -- (-3.5,-1.4);
\draw [gray,dashed] (3.5,0) -- (3.5,-1.4);


\draw [gray,dashed] (3,0) -- (3,-1.4);
\draw node[anchor=south] at (3.5,0) {$b$};
\draw node[anchor=south] at (4.5,0) {$a$};
\draw node[anchor=south] at (-3.5,0) {$-b$};
\draw node[anchor=south] at (-4.5,0) {$-a$};
\draw node[anchor=south] at (-1,0) {$-c$};

\draw node[anchor=south] at (1,0) {$c$};

%
%
%
%
\filldraw [black] (0,0) circle (2pt)node[anchor=south] at (0,0) {$u$};
\end{tikzpicture}
\caption{The representation of different intervals corresponding to each random variable as described in Corollary \ref{cor:extra}}
\end{figure}

\begin{enumerate}
\item $d(u,x) \in [b\frac{\log n}{n},a\frac{\log n}{n}] \cup [0,c\frac{\log n}{n}]  \cup [\frac{-c\log n}{n},\frac{-c/2\log n}{n}]\} \text{ for } i=1$
\item $d(u,x) \in [b\frac{\log n}{n},a\frac{\log n}{n}] \cup [0,c\frac{\log n}{n}]  \cup [\frac{b-c/2\log n}{n},\frac{(a-c)\log n}{n}]\} \text{  for  } i=2$
\item $d(u,x) \in [-a\frac{\log n}{n},-b\frac{\log n}{n}] \cup [-c\frac{\log n}{n},0]  \cup  [\frac{c/2\log n}{n},\frac{c\log n}{n}]\} \text{  for  } i=3$
\item $d(u,x) \in [-a\frac{\log n}{n},-b\frac{\log n}{n}] \cup [-c\frac{\log n}{n},0]  \cup [\frac{(c-a)\log n}{n},\frac{(c/2-b)\log n}{n}]\} \text{  for  } i=4$
\end{enumerate}

\begin{align*}
\Pr(A_u^i=1)&=
\begin{cases}
\Big(1-(c+a-b+(a-b-c/2))\frac{\log n}{n}\Big)^{n} \text{ when }a-c<b \text{ and } b-c/2 > c\\
\Big(1-(b-c)\frac{\log n}{n}\Big)^{n} \text{ when }a-c<b \text{ and } b-c/2 < c\\
\Big(1-(a)\frac{\log n}{n}\Big)^{n} \text{ when }a-c\ge b \text{ and } b-c/2 < c\\
\Big(1-(c+a-b+c/2)\frac{\log n}{n}\Big)^{n} \text{ when }a-c \ge b \text{ and } b-c/2 \ge c\\
\end{cases}
\end{align*}
Notice that $A_u^1$ and $A_u^2$ being zero implies that either there is a node in $\{x \mid d(u,x) \in [b\frac{\log n}{n},a\frac{\log n}{n}]\cup [0,c\frac{\log n}{n}] \}$ or there exists nodes $(v_1,v_2)$ in $\{x \mid d(u,x) \in [\frac{-c\log n}{n},\frac{-c/2\log n}{n}]\}$ and $\{x \mid d(u,x) \in  [\frac{b-c/2\log n}{n},\frac{(a-c)\log n}{n}]\}$. In the second case, $u$ is connected to $v_1$ and $v_1$ is connected to $v_2$. Therefore $u$ has nodes on left and right and $u$ is connected to both of them although not directly. Similarly $A_u^3$ and $A_u^4$ being zero implies that there exist nodes in $\{x \mid d(u,x) \in  [-a\frac{\log n}{n},-b\frac{\log n}{n}]\cup [-c\frac{\log n}{n},0]\}$ or again $u$ will have nodes on left and right and will be connected to them. So , when all the $4$ events happen together, the only exceptional case is when there are nodes in $\{x \mid d(u,x) \in [b\frac{\log n}{n},a\frac{\log n}{n}]\cup [0,c\frac{\log n}{n}] \}$ and $\{x \mid d(u,x) \in  [-a\frac{\log n}{n},b\frac{\log n}{n}]\cup [-c\frac{\log n}{n},0] \}$. But in that case $u$ has direct neighbors on both its left and right. So, we can conclude that for every node $u$, there exists a node $v$ such that $d(u,v) \in[0,\frac{a\log n}{n}]$ and a node $w$ such that $d(u,w) \in[\frac{-a\log n}{n},0]$ such that $u$ is connected to both $v$ and $w$. This implies that every node $u$ has neighbors on both its left and right  and therefore every node is part of a cycle that covers $[0,1]$. 
\end{proof}

%% file: rgg-high.tex
\section{Connectivity of High Dimensional Random Annulus Graphs: Detailed Proofs of Theorems \ref{th:lb} and \ref{thm:highdem1}}
\label{sec:hrag-detail}

In this section we first prove an impossibility result on the connectivity of random annulus graphs in $t$ dimensions by showing a sufficient condition of existence of isolated nodes.
 Next, we show that if the  gap between $r_1$ and $r_2$ is large enough then the RAG is fully connected. We will start by introducing a few notations. Let us define the regions $B_t(u,r)$ and $B_t(u,[r_1,r_2])$ for the any  $u\in S^t$ in the following way:
\begin{align*}
B_t(u,r)&=\{x \in S^t \mid \|u-x\|_2 \le r \} \\
B_t(u,[r_1,r_2])&=\{x \in S^t \mid  r_1 \le \|u-x\|_2 \le r_2 \}.
\end{align*}
First, we calculate  $|B_t(u,r)|$ and show that it is proportional to $r^t$. 
\begin{lemma}
$|B_t(u,r)|=c_tr^t$ for $r=o(1)$ where $c_t \approx \frac{\pi^{t/2}}{\Gamma (\frac{t}{2}+1) }$.
\label{lem:area}
\end{lemma}
\begin{proof}
We use the following fact from (\cite{larsen2017improved,li2011concise}) for the proof. For a $t$-dimensional unit sphere, the hyperspherical cap of angular radius $\theta= \max_{x \in B_t(u,r)} \arccos \langle x,u \rangle$ has a surface area $C_t(\theta)$ given by 
\begin{align*}
C_t(\theta)=\int_{0}^{\tan \theta} \frac{S_{t-1}(r)}{(1+r^{2})^{2}}dr 
\end{align*}
where $S_{t-1}(\theta)=\frac{t\pi^{t/2}}{\Gamma (\frac{t}{2}+1) }\theta^{t-1}$. Note that $C_t(\theta)$ is nothing but $|B_t(u,r)|$ where $\cos \theta=1-\frac{r^{2}}{2}$ and therefore $\tan \theta=\frac{r \sqrt{4-r^{2}}}{2-r^{2}} \approx r$ for small $r$. Now since $r=o(1)$ and $1+r^{2}$ is an increasing function of $r$, we must have that 
\begin{align*}
\int_{0}^{r} \frac{t\pi^{t/2}}{(1+o(1))\Gamma (\frac{t}{2}+1) }\theta^{t-1} d\theta  < C_t(\theta) <  \int_{0}^{r} \frac{t\pi^{t/2}}{\Gamma (\frac{t}{2}+1) }\theta^{t-1} d\theta
\end{align*}
and therefore $C_t(\theta)$ can be expressed as $c_t r^t$ where $c_t$ lies in $\Big(\frac{\pi^{t/2}}{(1+o(1))\Gamma (\frac{t}{2}+1) },\frac{\pi^{t/2}}{\Gamma (\frac{t}{2}+1) } \Big)$.
\end{proof}

\subsection{Impossibility Result}
The following theorem proves the impossibility result for the connectivity of a random annulus graph by proving a tight threshold for the presence of an isolated node with high probability.
\begin{theorem*}(Restatement of Theorem \ref{th:lb})
For a random annulus graph ${\rm RAG}_t(n,[r_1,r_2])$ where $r_1=b \Big(\frac{ \log n}{n}\Big)^{\frac{1}{t}}$ and $r_2=a\Big(\frac{\log n}{n}\Big)^{\frac{1}{t}}$, there exists isolated nodes with high probability if and only if 
$a^t -b^t < \frac{\sqrt{\pi}(t+1)\Gamma(\frac{t+2}{2})}{\Gamma(\frac{t+3}{2})}$.
\end{theorem*}
\begin{proof}
Consider the random annulus graph ${\rm RAG}_t(n,[r_1,r_2])$ in $t$ dimensions. In this graph, a node $u$ is isolated if there are no nodes $v$ such that $ r_1 \le \|u-v\|_2 \le r_2$. Since all nodes are uniformly and randomly distributed on $S^t$, the probability of a node $v$ being connected to a node $u$ is the volume of $B_t(u,[r_1,r_2])$. Define the indicator random variable $A_u \in \{0,1\}$ which is $1$ if and only if the node $u$ is isolated. Also define the random variable $A=\sum_u A_u$ which denotes the total number of isolated nodes. Since $|B_t(u,[r_1,r_2])|=c_t(r_2^t-r_1^t)$, we must have
\begin{align*}
\Pr(A_u=1)=\Bigg(1-\frac{c_t\Big(r_2^t-r_1^t \Big)}{|S^t|}\Bigg)^{n-1}.
\end{align*}
Now, we know from (\cite{surfacensphere}) that  
$|S^t|=\frac{(t+1)\pi^{\frac{t+1}{2}}}{\Gamma(\frac{t+3}{2})}$. Plugging in, we get that $\frac{c_t}{|S^t|}=\frac{\Gamma(\frac{t+3}{2})}{\sqrt{\pi}(t+1)\Gamma(\frac{t+2}{2})}$.
Hence, the expected number of isolated nodes $\avg A$ is going to be 
\begin{align*}
n\Bigg(1-\Big(a^t-b^t\Big)\frac{c_t}{|S^t|}\frac{\log n}{n}\Bigg)^{n-1} \approx n^{1- \frac{c_t(a^t-b^t)}{|S^t|}}.
\end{align*}
Therefore $\avg [A] \ge 1$ if $a^t-b^t <\frac{|S^t|}{c_t}$. In order to show that $A =\omega(1)$ with high probability we are going to show that the variance of $A$ is bounded from above. Since $A$ is a sum of indicator random variables, we have that 

$${\rm Var}(A) \le \avg[A]+\sum_{u \neq v} {\rm Cov}(A_u,A_v)=\avg[A]+\sum_{u \neq v} (\Pr(A_u=1 \cap A_v=1)-\Pr(A_u=1)\Pr(A_v=1)).$$ 
Now, consider the scenario when the vertices $u$ and $v$ are at a distance more than $2r_2$ apart which happens with probability $1-\frac{c_t(2r_2)^t}{|S^t|}$. Then the region  in which every point is within a distance of $r_2$ and $r_1$ from both $u,v$ is empty and therefore $\Pr(A_u=1 \cap A_v=1) =  
\Bigg(1-\frac{2c_t}{|S^t|}\Big(a^t-b^t\Big)\frac{\log n}{n}\Bigg)^{n-2}$.
When the vertices  are within distance $2r_2$ of one another, then
$
\Pr(A_u=1 \cap A_v=1) \le \Pr(A_u=1).
$
Therefore,
\begin{align*}
\Pr(A_u=1 \cap A_v=1) \le &\Big(1-\frac{c_t(2r_2)^t}{|S^t|} \Big) \Big(1- \frac{2c_t}{|S^t|}\Big(a^t-b^t\Big)\frac{\log n}{n}\Big)^{n-2} + \frac{c_t(2r_2)^t}{|S^t|} \Pr(A_u=1)\\ 
&\le  (1-\frac{c_t(2r_2)^t}{|S^t|}) n^{-\frac{2c_t}{|S^t|}(a^t-b^t)+o(1)}+ \frac{c_t(2r_2)^t}{|S^t|} n^{-\frac{c_t(a^t-b^t)}{|S^t|}+o(1)}.
\end{align*}
Consequently for large enough $n$,
\begin{align*}
\Pr(A_u=1 \cap A_v=1)-\Pr(A_u=1)\Pr(A_v=1) &\le (1-\frac{c_t(2r_2)^t}{|S^t|}) n^{-\frac{2c_t(a^t-b^t)}{|S^t|}+o(1)} \\+ \frac{c_t(2r_2)^t}{|S^t|} n^{-\frac{c_t(a^t-b^t)}{|S^t|}+o(1)}  -& n^{-\frac{2c_t(a^t-b^t)}{|S^t|} +o(1)} 
\le  \frac{2c_t(2r_2)^t}{|S^t|}\Pr(A_u=1).
\end{align*}
Now,
$$
{\rm Var}(A) \le \avg[A] + 2\binom{n}{2}\frac{c_t(2r_2)^t}{|S^t|}\Pr(A_u=1) \le \avg[A](1+ \frac{c_t(2a)^t}{|S^t|}  \log n).
$$
By using Chebyshev bound, with probability at least $1-O\Big(\frac{1}{\log n}\Big)$, 
$$A >n^{1-\frac{c_t(a^t-b^t)}{|S^t|}}-\sqrt{n^{1-\frac{c_t(a^t-b^t)}{|S^t|}}(1+\frac{c_t(2a)^t}{|S^t|} \log n)\log n},$$
which implies that for $a^t-b^t < \frac{|S^t|}{c_t}$, $A>1$ and hence there will exist isolated nodes with high probability.  
\end{proof}

 \subsection{Connectivity Bound}
We show the upper bound for connectivity of a Random Annulus Graphs in $t$ as per Theorem \ref{thm:highdem1}, rewritten below.

\begin{theorem*}(Restatement of Theorem \ref{thm:highdem1})
For $t$ dimensional random annulus graph ${\rm RAG}_t(n,[r_1,r_2])$ where $r_2=a\Big(\frac{\log n}{n}\Big)^{t}$ and $r_1=b\Big(\frac{\log n}{n}\Big)^{t}$ with $a\ge b$ and $t$ is a constant, the graph is connected with high probability if 
\begin{align*}
a^t-b^t \ge \frac{8|S^t|(t+1)}{c_t\Big(1 -  \frac{1}{{2^{1+1/t} - 1}}\Big)} \text{  and  }  a>2^{{1}+\frac{1}{t}}b.
\end{align*} 
\end{theorem*}
Let us define a \emph{pole} to be a vertex  which is connected to all vertices within a distance of $r_2$ from itself. 
In order to prove  Theorem \ref{thm:highdem1}, we first show the existence of a pole with high probability in Lemma \ref{lem:pole}. Next, Lemma \ref{lem:high_stuff1} shows that for every vertex $u$ and every hyperplane $L$ passing through $u$ and not too close to the tangential hyperplane at $u$, there will be a neighbor of $u$ on either side of the plane. In order to formalize this, let us define a few regions associated with a node $u$ and a hyperplane $L:w^{T}x=\beta$ passing through $u$.

\begin{align*}
\mathcal{R}_{L}^1 &\equiv \{x \in S^t \mid b\Big(\frac{\log n}{n}\Big)^{1/t} \le  \|u-x\|_2 \le a\Big(\frac{\log n}{n}\Big)^{1/t}, w^{T}x \le \beta \} \\
\mathcal{R}_{L}^2 &\equiv \{x \in S^t \mid b\Big(\frac{\log n}{n}\Big)^{1/t} \le  \|u-x\|_2 \le a\Big(\frac{\log n}{n}\Big)^{1/t}, w^{T}x \ge \beta \} \\
\mathcal{A}_{L} & \equiv \{x \mid x \in \mathcal{S}^t, \quad w^{T}x=\beta \}.
\end{align*}

Informally, $\mathcal{R}_{L}^1$ and $\mathcal{R}_{L}^2$ represents the partition of the region $B_t(u,[r_1,r_2])$ on either side of the hyperplane $L$ and $\mathcal{A}_L$ represents the region on the sphere lying on $L$. 

\begin{lemma*}(Restatement of Lemma \ref{lem:pole})
In  ${\rm RAG}_t\left(n, \left[b\left(\frac{\log n}{n}\right)^{1/t},a\left(\frac{\log n}{n}\right)^{1/t}\right]\right), 0 <b <a$, with  probability $1-o(1)$ there exists a pole.
\end{lemma*}

The proof of Lemma \ref{lem:pole} is delegated to Appendix \ref{sec:helper}.

\begin{lemma*}(Restatement of Lemma \ref{lem:high_stuff1})
If we sample $n$ nodes from $S^t$ according to ${\rm RAG}_t\left(n,\left[b\left(\frac{\log n}{n}\right)^{1/t},a\left(\frac{\log n}{n}\right)^{1/t}\right]\right)$, then for every node $u$ and every hyperplane $L$ passing through $u$ such that $\mathcal{A}_L \not \subset B_t(u,a\left(\frac{\log n}{n}\right)^{1/t})$, node $u$ has a neighbor on both sides of the hyperplane $L$ with probability at least $1-\frac{1}{n}$ provided 
\begin{align*}
a^t-b^t \ge \frac{8|S^t|(t+1)}{c_t\Big(1 -  \frac{1}{{2^{1+1/t} - 1}}\Big)}
\end{align*}  and $a>2^{1+\frac{1}{t}}b$.
\end{lemma*}

For a node $u \equiv (u_1,u_2,\dots,u_{t+1})$, define the particular hyperplane $L^{\star}_u : x_1=u_1$ which is normal to the line joining $u_0 \equiv (1,0,\dots,0)$ and the origin and passes through $u$. We now have the following lemma.
\begin{lemma*}(Restatement of Lemma \ref{lem:cut_twice})
For a particular node $u$ and corresponding hyperplane $L^{\star}_u$, if $\mathcal{A}_{L^{\star}_u}  \subseteq B_t(u,r_2)$ then $u$ must be within $r_2$ of $u_0$.
\end{lemma*}

For now, we assume that the Lemmas \ref{lem:pole}, \ref{lem:high_stuff1} and \ref{lem:cut_twice} are true and show why these lemmas together imply the proof of Theorem \ref{thm:highdem1}.
\begin{proof}[Proof of Theorem \ref{thm:highdem1}] 
We consider an alternate (rotated but not shifted) coordinate system by multiplying every vector by a orthonormal matrix $R$ such that the new position of the pole is the $t+1$-dimensional vector $(1,0,\dots,0)$ where only the first co-ordinate is non-zero. Let the $t+1$ dimensional vector describing a node $u$ in this new coordinate system be $\hat{u}$. Now consider the hyperplane $L: x_1=\hat{u}_1$ and if $u$ is not connected to the pole already, then by Lemma \ref{lem:high_stuff1} and Lemma \ref{lem:cut_twice} the node $u$ has a neighbor $u_2$ which has a higher first coordinate. The same analysis applies for $u_2$ and hence we have a path where the first coordinate of every node is higher than the previous node. Since the number of nodes is finite, this path cannot go on indefinitely and at some point, one of the nodes is going to be within $r_2$ of the pole and will be connected to the pole. Therefore every node is going to be connected to the pole and hence our theorem is proved. 

\end{proof}
We show the proofs of Lemma \ref{lem:pole}, \ref{lem:high_stuff1} and \ref{lem:cut_twice} in the following sections.

\subsection{Proof of Lemma \ref{lem:pole}}{\label{sec:helper}}
Lemma \ref{lem:pole_helper} is a helper lemma that shows the region of connectivity for a small ball of radius $\epsilon (\frac{\log n}{n})^{1/t}$. Lemma \ref{lem:lemma2pole} uses this  lemma  to show the existence of a point $u_0$ which is connected to various balls of radius $\epsilon (\frac{\log n}{n})^{1/t}$.
\begin{lemma}\label{lem:pole_helper}
For a $t$ dimensional random annulus graph ${\rm RAG}_t(n,[r_1,r_2])$ where $r_1=b\left(\frac{\log n}{n}\right)^{1/t}, r_2=a\left(\frac{\log n}{n}\right)^{1/t}$ and $a\ge b$, consider the region ${B_t}(O,\theta)$ centered at $O$ and radius $\theta = \epsilon \left(\frac{\log n}{n}\right)^{1/t}$. Then, every vertex in ${B_t}(O,\theta)$ is connected to all vertices present in  ${B_t}(O,[\theta_1,\theta_2])$ where $\theta_1 = (b+\epsilon)\left(\frac{\log n}{n}\right)^{1/t}$ and $\theta_2 = (a-\epsilon)\left(\frac{\log n}{n}\right)^{1/t}$.
\end{lemma}
\begin{proof}
For any point $A\in {B_t}(O,\theta)$, we have $0<\|A-O\|_2\leq \theta$ and for any point $X \in {B_t}(O,[\theta_1,\theta_2])$, we must have $\theta_1 \leq \|X-O\|_2\leq \theta_2$. Hence 
\begin{align*}
\|A-X\|_2&\leq \|A-O\|_2 + \|X-O\|_2   \\
&\leq \theta + \theta_2\\
&=  a \left(\frac{\log n}{n}\right)^{1/t},\\
\|A-X\|_2&\geq \|X-O\|_2 - \|A-O\|_2\\
&\geq \theta_1 -  \theta\\
&=  b \left(\frac{\log n}{n}\right)^{1/t},
\end{align*}
and therefore the claim of the lemma is proved.
\end{proof}

\begin{figure}
\vspace{-20pt}
\centering
\begin{tikzpicture}[thick, scale=0.5]
 \filldraw[color=black, fill=blue!20,  line width=2pt](0,0) circle (4.5);
 \filldraw[color=black, fill=red!0,  line width=2pt](0,0) circle (3.5);
 \filldraw[color=gray, fill=red!10,  line width=2pt](0,0) circle (1);
  \filldraw [gray] (0,0) circle (2pt)node[anchor=north ] at (0,0){O};
  \draw[line width=1pt,gray,->] (0,0)--(-1,0)node[anchor=north ] at (-0.5,0.2){$\epsilon$};
    \draw[gray,->] (0,0)--(-3.3,1)node[anchor=north ] at (-2.3,0.9){$b+\epsilon$};
      \draw[gray,->] (0,0)--(-4,2)node[anchor=north ] at (-2.1,2.2){$a-\epsilon$};
\end{tikzpicture}
\caption{Any node in the red region is connected to any node in the blue region.\label{fig:arr0high}}
\end{figure}
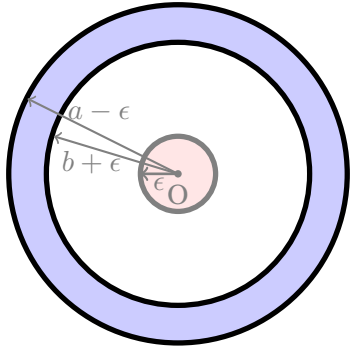


\begin{lemma}
Set two real numbers $k\equiv \lceil b/(a-b)\rceil+1$ and $\epsilon < \left( \frac{|S_t|}{2kc_t} \right)^{1/t}$. In an ${\rm RAG}_t\left(n, \left[b\left(\frac{\log n}{n}\right)^{1/t},a\left(\frac{\log n}{n}\right]\right)^{1/t}\right), 0 <b <a$, with  probability $1-o(1)$ there exists a vertex $u_0 \in S^t$  with the following property. 
Consider a homogeneous hyperplane $L$ in $\reals^{t+1}$ that pass through $u_0$.
There are $k$ nodes $\cA = \{u_1, u_2 ,\ldots, u_k\}$  with $u_i \in {B_t}\left(O_{u_i}, \epsilon \left(\frac{\log n}{n}\right)^{1/t}\right)$ for some $O_{u_i}\in L \cap S^t$ such that $\|O_{u_i}-u_0\|_2 = (i(a-b)-(4i-1)\epsilon)\left(\frac{\log n}{n}\right)^{1/t}$ and  $k$ nodes $\cB = \{v_1, v_2 ,\ldots, v_k\}$ with $v_i \in {B_t}\left(O_{v_i}, \epsilon \left(\frac{\log n}{n}\right)^{1/t}\right)$ for some $O_{v_i}\in L\cap S^t$ such that $\|O_{v_i}-u_0\|_2= (i(a-b)+b-(4i-3)\epsilon)\left(\frac{\log n}{n}\right)^{1/t}$, for $i =1,2,\ldots,k$
with $\cA$ and $\cB$ separated by $L$.
~\label{lem:lemma2pole}
\end{lemma}

\begin{proof}[Proof of Lemma~\ref{lem:lemma2pole}]
 Let $A_u$ be an indicator $\{0,1\}$-random variable for every node $u$ which is $1$ if $u$ satisfies the conditions stated in the lemma and $0$ otherwise. We will show $\sum_{u} A_u \ge 1$ with high probability.

We  have,
\begin{align*}
\Pr(A_u=1) & = \frac{1}{2k!} n(n-1)\dots (n-(2k-1))\Big(\frac{\epsilon^tc_t \log n}{n|S^t|}\Big)^{2k} \Big(1-2kc_t\epsilon^t\frac{\log n}{n|S^t|}\Big)^{n-2k}\\
& = c_1 n^{-2k\epsilon^t c_t/|S^t|} (\epsilon^t \log n)^{2k}  \prod_{i=0}^{2k-1} (1-i/n)\\
&= c_2 n^{-2k\epsilon^t c_t/|S^t|} (\epsilon^t \log n)^{2k}
\end{align*}
where $c_1=\frac{c_t^{2k}}{2k!|S^t|^{2k}},c_2$ are just absolute constants independent of $n$ (recall $k$ is a constant). 
Hence,
\begin{align*}
\sum_{u} \avg A_u= c_2 n^{1-2k\epsilon^t c_t/|S^t|} (\epsilon^t \log n)^{2k} \ge 1
\end{align*} 
as long as $\epsilon \leq \left( \frac{|S^t|}{2kc_t} \right)^{1/t}$.  Now, in order to prove $\sum_{u}  A_u\ge 1$ with high probability, we will show that the variance of $\sum_{u}  A_u$  is bounded from above.  
Recall that if $A =\sum_{u}  A_u$ is a sum of indicator random variables, we must have 
$${\rm Var}(A) \le \avg[A]+\sum_{u \neq v} {\rm Cov}(A_u,A_v)=\avg[A]+\sum_{u \neq v} \Pr(A_u=1 \cap A_v=1)-\Pr(A_u=1)\Pr(A_v=1).$$ 
Now first consider the case when vertices $u$ and $v$ are at a distance of at least $2(a+b) \left(\frac{ \log n}{n}\right)^{1/t}$ apart (happens with probability $1- 4^t(a+b)^t c_{t} \left(\frac{\log n}{n|S^t|}\right)$). 
Then the region that is within distance $(a+b) \left(\frac{\log n}{n}\right)^{1/t}$ from both $u$ and $v$ is the empty-set. In that case, $\Pr(A_u=1 \cap A_v=1) = n(n-1)\dots (n-(4k-1)) c_3 \Big(\frac{\epsilon^t c_t \log n}{n|S^t|}\Big)^{4k} \Big(1-4k\epsilon^t \frac{c_t\log n}{n|S^t|}\Big)^{n-4k} = c_4 n^{-4k\epsilon^tc_t/|S^t|} (\epsilon^t \log n)^{4k},$ where $c_3,c_4$ are constants.

In all other cases, $\Pr(A_u=1 \cap A_v=1) \le \Pr(A_u =1)$.
Therefore,\\
\begin{align*}
\Pr(A_u=1 \cap A_v=1)&\leq &\Big(1- 4^t(a+b)^t c_t \left(\frac{\log n}{n|S^t|}\right)\Big) c_4 n^{-4k\epsilon^tc_t/|S^t|} (\epsilon^t \log n)^{4k}+\\ 
&&\frac{4^t(a+b)^tc_t \log n}{n|S^t|} c_2 n^{-2k\epsilon^t c_t/|S^t|} (\epsilon^t \log n)^{2k}
\end{align*}
and
\begin{align*}
{\rm Var(A)} &\le c_2n^{1-2k\epsilon^tc_t/|S^t|} (\epsilon^t \log n)^{2k} +{n \choose 2}\Big(\Pr(A_u=1 \cap A_v=1)-\Pr(A_u=1)\Pr(A_v=1)\Big) \\
&\le c_2n^{1-2k\epsilon^tc_t/|S^t|} (\epsilon^t \log n)^{2k}+ c_5 n^{1-2k\epsilon^tc_t/|S^t|} (\log n)^{2k+1}\\
& \le c_6  n^{1-2k\epsilon^tc_t/|S^t|} (\log n)^{2k+1}
\end{align*}
where $c_5,c_6$ are constants. Again invoking Chebyshev's inequality, with probability at least $1-O\Big(\frac{1}{\log n}\Big)$ 
$$
A > c_2n^{1-2k\epsilon^tc_t/|S^t|} (\epsilon^t \log n)^{2k} - \sqrt{c_6  n^{1-2k\epsilon^tc_t/|S^t|} (\log n)^{2k+2}}
$$  
which implies that $A>1$ with high probability.
\end{proof}

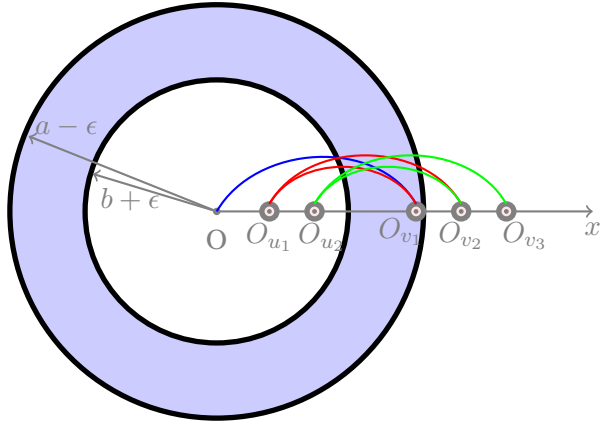
\begin{figure}
\vspace{-20pt}
\centering
\begin{tikzpicture}[thick, scale=0.5]

 \filldraw[color=black, fill=blue!20,  line width=2pt](0,0) circle (5.5);
 \filldraw[color=black, fill=red!0,  line width=2pt](0,0) circle (3.5);
    \draw[gray,->] (0,0)--(-3.3,1)node[anchor=north ] at (-2.3,0.9){$b+\epsilon$};
      \draw[gray,->] (0,0)--(-5,2)node[anchor=north ] at (-4,2.8){$a-\epsilon$};
   \draw[gray,->] (0,0)--(10,0)node[anchor=north ] at (10,0){$x$};

   \filldraw[color=gray, fill=red!10,  line width=2pt](5.3,0) circle (0.2)node[anchor=north ] at (4.9,0.15){$O_{v_1}$};
      \filldraw[color=gray, fill=red!10,  line width=2pt](6.5,0) circle (0.2)node[anchor=north ] at (6.5,0.05){$O_{v_2}$};
            \filldraw[color=gray, fill=red!10,  line width=2pt](7.7,0) circle (0.2)node[anchor=north ] at (8.2,0.05){$O_{v_3}$};
      \filldraw[color=gray, fill=red!10,  line width=2pt](1.4,0) circle (0.2)node[anchor=north ] at (1.4,0){$O_{u_1}$};
	\filldraw[color=gray, fill=red!10,  line width=2pt](2.6,0) circle (0.2)node[anchor=north ] at (2.8,0){$O_{u_2}$};
    \filldraw [gray] (5.3,0) circle (1pt);
      \filldraw [gray] (0,0) circle (2pt)node[anchor=north ] at (0,-0.2){O};
      \filldraw [gray] (1.4,0) circle (1pt);
            \filldraw [gray] (2.6,0) circle (1pt);
                        \filldraw [gray] (6.5,0) circle (1pt);
                                    \filldraw [gray] (7.7,0) circle (1pt);
              \draw[blue]    (0,0) to[out=60,in=120] (5.3,0.2);
  \draw[red]    (1.4,0.2) to[out=60,in=120] (5.3,0.2);
    \draw[red]    (1.4,0.2) to[out=60,in=120] (6.5,0.2);
        \draw[green]    (2.6,0.2) to[out=60,in=120] (6.5,0.2);
                \draw[green]    (2.6,0.2) to[out=60,in=120] (7.7,0.2);
  

\end{tikzpicture}
\caption{Representation of $u_i$ and $v_i$ in the $t+1$-dimensional sphere with respect to $u_0$.\label{fig:arr1}}
\end{figure}

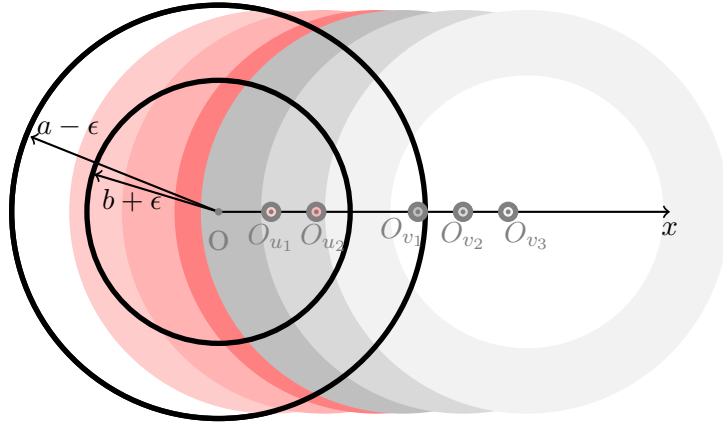
\begin{figure}
\vspace{-10pt}
\centering
\begin{tikzpicture}[thick, scale=0.5]

 \filldraw[color=black, fill=none,  line width=2pt](0,0) circle (5.5);
 
  \filldraw[color=red!20, fill=red!20,  line width=2pt] (1.4,0) circle (5.3);
 \filldraw[color=red!20, fill=red!0,  line width=2pt] (1.4,0) circle (3.7);
  
  \filldraw[color=red!30, fill=red!30,  line width=2pt] (2.8,0) circle (5.3);
    \filldraw[color=red!30, fill=red!0,  line width=2pt] (2.8,0) circle (3.7);
    
        \filldraw[color=red!50, fill=red!50,  line width=2pt] (4.2,0) circle (5.3);    
        \filldraw[color=red!50, fill=red!0,  line width=2pt] (4.2,0) circle (3.7);

                \filldraw[color=gray!50, fill=gray!50,  line width=2pt] (4.9,0) circle (5.3);  
                \filldraw[color=gray!50, fill=gray!0,  line width=2pt] (4.9,0) circle (3.7);
                
                                                \filldraw[color=gray!30, fill=gray!30,  line width=2pt] (6.5,0) circle (5.3);
                                \filldraw[color=gray!30, fill=gray!0,  line width=2pt] (6.5,0) circle (3.7);
                                
                                                                                \filldraw[color=gray!10, fill=gray!10,  line width=2pt] (8.2,0) circle (5.3);
                                \filldraw[color=gray!10, fill=blue!0,  line width=2pt] (8.2,0) circle (3.7);
                                
 \filldraw[color=black, fill=none,  line width=2pt](0,0) circle (3.5);

 \filldraw[color=black, fill=none,  line width=2pt](0,0) circle (5.5);

    \draw[black,->] (0,0)--(-3.3,1)node[anchor=north ] at (-2.3,0.9){$b+\epsilon$};
      \draw[black,->] (0,0)--(-5,2)node[anchor=north ] at (-4,2.8){$a-\epsilon$};
   \draw[black,->] (0,0)--(12,0)node[anchor=north ] at (12,0){$x$};

   \filldraw[color=gray, fill=gray!50,  line width=2pt](5.3,0) circle (0.2)node[anchor=north ] at (4.9,0.15){$O_{v_1}$};
      \filldraw[color=gray, fill=gray!30,  line width=2pt](6.5,0) circle (0.2)node[anchor=north ] at (6.5,0.05){$O_{v_2}$};
            \filldraw[color=gray, fill=gray!10,  line width=2pt](7.7,0) circle (0.2)node[anchor=north ] at (8.2,0.05){$O_{v_3}$};
      \filldraw[color=gray, fill=red!20,  line width=2pt](1.4,0) circle (0.2)node[anchor=north ] at (1.4,0){$O_{u_1}$};
	\filldraw[color=gray, fill=red!30,  line width=2pt](2.6,0) circle (0.2)node[anchor=north ] at (2.8,0){$O_{u_2}$};
    \filldraw [gray] (5.3,0) circle (1pt);
      \filldraw [gray] (0,0) circle (2pt)node[anchor=north ] at (0,-0.2){O};
      \filldraw [gray] (1.4,0) circle (1pt);
            \filldraw [gray] (2.6,0) circle (1pt);
                        \filldraw [gray] (6.5,0) circle (1pt);
                                    \filldraw [gray] (7.7,0) circle (1pt);
\end{tikzpicture}
\caption{Shaded regions represent the region of connectivity with $u_i$ and $v_i$ (red for $u_i$'s and gray for $v_i$'s).\label{fig:arr2}}
\end{figure}
\begin{lemma*}(Restatement of Lemma \ref{lem:pole})
In a ${\rm RAG}_t\left(n, \left[b\left(\frac{\log n}{n}\right)^{1/t},a\left(\frac{\log n}{n}\right)^{1/t}\right]\right), 0 <b <a$, with  probability $1-o(1)$ there exists a vertex $u_0$ such that any node $v$ that satisfies $\|u-v\|_2 \le a\left(\frac{\log n}{n}\right)^{1/t}$ is connected to $u_0$.
\end{lemma*}
\begin{proof}
Consider the vertices $u_0,\{u_1, u_2, \ldots, u_k\}$ and $\{v_1, v_2, \ldots, v_k\}$ that satisfy the conditions of Lemma \ref{lem:lemma2pole} as shown in Fig \ref{fig:arr1}. We can observe that each vertex $v_i$ has an edge with $u_i$ and $u_{i-1}$, $i =1, \ldots,k$. 
\begin{align*}
\|u_i-v_i\|_2 &\ge \|u_i-O_{v_i}\|_2 - \|v_i-O_{v_i}\|_2\\
&\ge \|O_{v_i}-O_{u_i}\|_2  - \|u_i-O_{u_i}\|_2 - \|v_i-O_{v_i}\|_2\\
& \ge (b+2\epsilon)\left(\frac{\log n}{n}\right)^{1/t} - 2\epsilon \left( \frac{\log n}{n}\right)^{1/t}= b\left(\frac{ \log n}{n}\right)^{1/t} \quad \text{and}
\end{align*}
\begin{align*}
 \|u_i-v_i\|_2 &\le  \|O_{v_i}-O_{u_i}\|_2 + \|u_i-O_{u_i}\|_2 + \|v_i-O_{v_i}\|_2\\
  & =  (b+2\epsilon)\left(\frac{\log n}{n}\right)^{1/t} + 2\epsilon  \left(\frac{\log n}{n}\right)^{1/t}= (b+4\epsilon)\left(\frac{ \log n}{n} \right)^{1/t}
\end{align*}

Similarly, 
\vspace{-10pt}
\begin{align*}
\|u_{i-1}-v_i\|_2& \ge \|u_{i-1}-O_{v_i}\|_2 - \|v_i-O_{v_i}\|_2\\
&\ge \|O_{v_i}-O_{u_{i-1}}\|_2 - \|u_{i-1}-O_{u_{i-1}}\|_2 - \|v_i-O_{v_i}\|_2\\
&\ge (a-2\epsilon)\left( \frac{\log n}{n}\right)^{1/t} - 2\epsilon\left(\frac{ \log n}{n}\right)^{1/t}\\
& = (a-4\epsilon)\left(\frac{\log n}{n}\right)^{1/t} \quad \text{and}
\end{align*}
\begin{align*}
\|u_{i-1}-v_i\|_2 &\le \|O_{v_i}-O_{u_{i-1}}\|_2  + \|u_{i-1}-O_{u_{i-1}}\|_2 + \|v_i-O_{v_i}\|\\
&\le (a-2\epsilon)\left(\frac{\log n}{n}\right)^{1/t} - 2\epsilon\left(\frac{ \log n}{n}\right)^{1/t}= a\left(\frac{\log n}{n}\right)^{1/t}.
\end{align*}
 This implies that $u_0$ is connected to $u_i$ and $v_i$ for all $i=1,\dots,k$.  Next, we show that any point in the region ${B_t}\Big(u_0,r_s=a(\frac{\log n}{n})^{1/t}\Big)$ is connected to $u_0$. Now recall that any point in the region ${B_t}\Big(x,\left[(b+\epsilon)\Big(\frac{\log n}{n}\Big)^{t},(a-\epsilon)\Big(\frac{\log n}{n}\Big)^{t}\right]\Big)$ is connected to any point in the region ${B_t}(x,\epsilon)$.
 We can observe that the nodes $u_1,\ldots, u_k,v_1$,\ldots, $v_k$ form a cover of ${B_t}\Big(u_0,r_s=a(\frac{\log n}{n})^{1/t}\Big)$ in the form of these annulus regions (A region corresponding to a particular node implies the portion of the hypersphere such that any other node in that region is connected to it) translated by $(a-b-4\epsilon)\left(\frac{\log n}{n}\right)^{1/t}$. This is because any node in ${B_t}\Big(O_{u_i},\left[(b+\epsilon)\Big(\frac{\log n}{n}\Big)^{t},(a-\epsilon)\Big(\frac{\log n}{n}\Big)^{t}\right]\Big)$ is connected to $u_i$ and any node in ${B_t}\Big(O_{v_i},\left[(b+\epsilon)\Big(\frac{\log n}{n}\Big)^{t},(a-\epsilon)\Big(\frac{\log n}{n}\Big)^{t}\right]\Big)$  is connected to $v_i$ respectively (Figure\ref{fig:arr2}). Therefore any node falling in any of the aforementioned regions is connected with $u_0$. Since the width of each region is $(a-b-2\epsilon)\left(\frac{\log n}{n}\right)^{1/t}$, the regions overlap with each other. Additionally, the inner radius of a particular region is $(b+\epsilon)\left(\frac{\log n}{n}\right)^{1/t}$ which is greater than $b\left(\frac{\log n}{n}\right)^{1/t}$. Hence, there can not exist any point in ${B_t}\Big(u_0,r_2=a(\frac{\log n}{n})^{1/t}\Big)$ which is not covered by the union of these regions. 
\end{proof}

\subsection{Proofs of Lemma \ref{lem:high_stuff1} and Lemma \ref{lem:cut_twice}} 
Assume that the $(t+1)$-dimensional space is described by a coordinate system whose center coincides with the center of the sphere. In this coordinate system, let us denote the point $(1,0,0,\dots,0)$ by $u_0$. 

Lemma \ref{conjec:ratio} shows that for any plane $L$ with  $\mathcal{A}_{L} \not \subset B_t(u,r_2)$, the region of connectivity of $u$ on both sides differ by a constant fraction.
\begin{lemma}\label{conjec:ratio}
For a particular node $u$ in ${\rm RAG}_t\left(n,\left[r_1,r_2\right]\right)$ where $r_1=b\Big(\frac{\log n}{n}\Big)^{1/t}$, $r_2=a\Big(\frac{\log n}{n}\Big)^{1/t}$, consider a hyperplane $L$ passing through $u$ such that $\mathcal{A}_{L} \not \subset B_t(u,a\Big(\frac{\log n}{n}\Big)^{1/t})$, then $ \frac{\min (|R_{L}^{1}|,|R_{L}^{2}|)}{|R_{L}^{1}|+|R_{L}^{2}|} \ge \delta$ if $a > 2b$, where $\delta$ is a constant.
\end{lemma}
\begin{proof}

First, for a node $u$ and a given hyperplane $L:w^{T}x=\beta$ passing through $u$, we try to evaluate the surface area of the region corresponding to
\begin{align*}
\{x \in S^t \mid  \|u-x\|_2 \le r_2=a\Big(\frac{\log n}{n}\Big)^{1/t}, w^{T}x \ge \beta \}
\end{align*}
such that  the farthest point from $u$ on the plane $L$ and $S^t$ is at distance $r_2$. This region is a spherical cap corresponding to $B_t(u',r')$ where $u'$ is the intersection of $S^t$ with the normal from the origin to the plane and $r'=\|u-u'\|_2$. Suppose $h$ is the height of this cap (perpendicular distance from $u'$ to the hyperplane $L$). Using pythagoras theorem, we can see that $r'^2=h^2+\left(r_2/2\right)^2$ and $(1-h)^2+ r_2^2/4 = 1$. Simplifying this, we get $h=\frac{r'^2}{2}$ and hence $r'\approx  \frac{r_2}{2}$.
Hence the area of this region is $ c_t \left(r_2/{2}\right)^t$.



Without loss of generality, assume $|\mathcal{R}_{L}^1| \ge |\mathcal{R}_{L}^2|$. Now,
\begin{eqnarray*}
|\mathcal{R}_{L}^1| &=& |\{x \in S^t \mid b\Big(\frac{\log n}{n}\Big)^{1/t} \le  \|u-x\|_2 \le a\Big(\frac{\log n}{n}\Big)^{1/t}, w^{T}x \le \beta \}|\\
&=&  |\{x \in S^t \mid b\Big(\frac{\log n}{n}\Big)^{1/t} \le  \|u-x\|_2 \le a\Big(\frac{\log n}{n}\Big)^{1/t} \}| - |\mathcal{R}_{L}^2|\\
&=& c_t(r_2^t-r_1^t) -  |\mathcal{R}_{L}^2|.\\\\
|\mathcal{R}_{L}^2| &\equiv& |\{x \in S^t \mid b\Big(\frac{\log n}{n}\Big)^{1/t} \le  \|u-x\|_2 \le a\Big(\frac{\log n}{n}\Big)^{1/t}, w^{T}x \ge \beta \}|\\
&\equiv& |\{x \in S^t \mid  \|u-x\|_2 \le a\Big(\frac{\log n}{n}\Big)^{1/t}, w^{T}x \ge \beta \}| - |\{x \in S^t \mid   \|u-x\|_2 \le b\Big(\frac{\log n}{n}\Big)^{1/t}, w^{T}x \ge \beta \}|\\
 &\ge&   |\{x \in S^t \mid  \|u-x\|_2 \le a\Big(\frac{\log n}{n}\Big)^{1/t}, w^{T}x \ge \beta \}| -\\
&& \left[ |\{x \in S^t \mid   \|u-x\|_2 \le b\Big(\frac{\log n}{n}\Big)^{1/t},w^{T}x < \beta \}| + |\{x \in S^t \mid   \|u-x\|_2 \le b\Big(\frac{\log n}{n}\Big)^{1/t},w^{T}x \ge \beta \}|\right]\\
&\equiv& |\{x \in S^t \mid  \|u-x\|_2 \le a\Big(\frac{\log n}{n}\Big)^{1/t}, w^{T}x \ge \beta \}| - |\{x \in S^t \mid   \|u-x\|_2 \le b\Big(\frac{\log n}{n}\Big)^{1/t} \}|\\
&\equiv& c_t \left(r_2/2\right)^t -c_tr_1^t
\end{eqnarray*}
If $c_t \left(r_2/{2}\right)^t -c_tr_1^t > 0$, then,
\begin{eqnarray*}
1\le \frac{|\mathcal{R}_{L}^1|}{|\mathcal{R}_{L}^2|} &\le&  \frac{c_t(r_2^t-r_1^t)}{c_t \left(r_2/{2}\right)^t -c_tr_1^t}-1\\
&=& \frac{ a^t -b^t}{ \left(a/{2}\right)^t -b^t} -1 = \delta'\\
\end{eqnarray*}
Hence, 
\begin{eqnarray*}
2\le \frac{|\mathcal{R}_{L}^1|+ |\mathcal{R}_L^2|}{|\mathcal{R}_{L}^2|} &\le&  1+\delta'\\
\end{eqnarray*}
This gives us that $\frac{\min (|\mathcal{R}_{L}^{1}|,|\mathcal{R}_{L}^{2}|)}{|\mathcal{R}_{L}^{1}|+|\mathcal{R}_{L}^{2}|} =\frac{|\mathcal{R}_{L}^2|}{|\mathcal{R}_{L}^1|+|\mathcal{R}_{L}^2|} \geq \frac{1}{1+\delta'} = \delta$. 
Hence, the claim of the lemma is satisfied   if $\left(a/{2}\right)^t -b^t > 0$ i.e. $a>2b.$
\end{proof}

\begin{corollary}\label{cor:ratio}
For a particular node $u$ in ${\rm RAG}_t\left(n,\left[b\Big(\frac{\log n}{n}\Big)^{1/t},a\Big(\frac{\log n}{n}\Big)^{1/t}\right]\right)$, consider a hyperplane $L$ passing through $u$ such that $\mathcal{A}_{L} \not \subset B_t(u,a\Big(\frac{\log n}{n}\Big)^{1/t})$, then $ \frac{\min (|\mathcal{R}_{L}^{1}|,|\mathcal{R}_{L}^{2}|)}{|\mathcal{R}_{L}^{1}|+|\mathcal{R}_{L}^{2}|} \ge  \frac{ \left(a/2\right)^t -b^t}{ a^t -b^t}$.
\end{corollary}
\begin{proof}
Using Lemma \ref{conjec:ratio}, $ \frac{\min (|\mathcal{R}_{L}^{1}|,|\mathcal{R}_{L}^{2}|)}{|\mathcal{R}_{L}^{1}|+|\mathcal{R}_{L}^{2}|} \ge \delta = \frac{1}{1+ \delta'}$ where $\delta' = \frac{ \left(a\right)^t -b^t}{ \left(a/{2}\right)^t -b^t} -1 $.
\end{proof}

For a node $u$, recall that the hyperplane $L^{\star}_u : x_1=u_1$ is normal to the line joining $u_0$ and the origin and passes through $u$. We now have the following lemma, which tries to show that if the plane satisfies $\mathcal{A}_{L^{\star}_u}  \subseteq B_t(u,r_2)$ then the node $u$ must be within $r_2$  distance of $u_0$. 
\begin{lemma*}(Restatement of Lemma \ref{lem:cut_twice})
For a particular node $u$ and corresponding hyperplane $L^{\star}_u$, if $\mathcal{A}_{L^{\star}_u}  \subseteq B_t(u,r_2)$  then $u$ must be within $r_2$ of $u_0$.
\end{lemma*}
\begin{proof}
The reflection of  $u$ in x-axis (say $v$) is the farthest from $u$ that lies on both $\mathcal{A}_{L^{\star}_u}$ and $\mathcal{S}^t$.  Now we want to show that $v = (u_1,-u_2,\ldots,-u_{t+1})$ has the following property:  if $\|u-v\|_2\le r_2$ then $\|u-u_0\|_2\le r_2$.
We are given that
\begin{eqnarray*}
u_1^2+u_2^2+\dots+u_{t+1}^2 = 1\\
d(u,v) = \sqrt{4(u_2^2+\ldots+u_{t+1}^2)}\le r_2\\
4(1-u_1^2)\leq r_2^2
\end{eqnarray*}
We need to show that, 
\begin{eqnarray*}
d(u,u_0)^2 &=& {(1-u_1)^2 + (u_2^2+\ldots+u_{t+1}^2) }\\
&=&(1-u_1)^2 + (1-u_1^2) \\
&=&  {2-2u_1}\\
&\leq & r_2^2
\end{eqnarray*}
which holds if $\sqrt{\frac{4-r_2^2}{4}}\ge \frac{2-r_2^2}{2}$. Notice that, 
\begin{eqnarray*}
&\sqrt{\frac{4-r_2^2}{4}}\ge \frac{2-r_2^2}{2}\\
 &\implies 4-r_2^2 \geq4-4r_2^2+r_2^4\\
 & \implies r_2^4-3r_2^2\leq 0 \\
 & \implies r_2^2(r_2^2-3) \le 0
\end{eqnarray*}
which is true since $0 \le  r_2 \le 1$.
\end{proof}
Since, we do not know the location of the pole, we need to show that every point has a neighbor on both sides of the plane $L$ no matter what the orientation of the plane given that $\mathcal{A}_{L} \not \subset B_t(u,r_2)$. For this we need to introduce the concept of VC Dimension. Define $(X,R)$ to be a range space if $X$ is a set (possibly infinite) and $R$ is a family of subsets of $X$. For any set $A \subseteq X$, we define $P_{R}(A)=\{r \cap A \mid r \in R\}$ to be the projection of $R$ on $A$. Finally we define the VC dimension $d$ of a range space $(X,R)$ to be $d=\sup_{A \subseteq X} \{|A| \mid |P_R(A)|=2^{A} \}$. Next we give a modified version of a well-known theorem about VC-dimension (\cite{haussler1987}).

\begin{lemma}\label{thm:VC_dim}
Let $(X,R)$ be a range space of VC dimension $d$ and let $U$ be a uniform probability measure defined on $X$. In that case, if we sample a set $\mathcal{M}$ of $m$ points according to $U$ such that
\begin{align*}
m \ge \max \Big(\frac{8d}{\epsilon}\log \frac{8d}{\epsilon}, \frac{4}{\epsilon}\log \frac{2}{\eta} \Big)
\end{align*}
then with probability $1-\eta$ for any set $r \in R$ such that $\Pr_{x \sim_U X}(x \in r) \ge \epsilon $, we have $|r \cap \mathcal{M}| \neq \Phi$.  \\
\end{lemma}
\begin{proof}
Define a set $r \in R$ to be \emph{heavy} if $\Pr_{x \sim_{U} X} (x \in r) \ge \epsilon$. We pick two random samples $N$ and $T$ each of size $m$ according to the uniform distribution defined on $X$. Consider the event $E_1$ (bad event) for which there exists a heavy $r \in R$ such that $r \cap N=\Phi$. Consider another event 
$E_2$ for which there exists a heavy $r \in R$ such that $r \cap N=\Phi$ and $|r \cap T| \ge \frac{\epsilon m}{2}$. Now, since $r$ is heavy, assume that $\Pr_{x \sim_{U} X} (x \in r)=\alpha$ such that $\alpha>\epsilon$. In that case, $|r \cap T|$ is a Binomial random variable with mean $\alpha m$ and variance at most $\alpha m$ as well. Hence, we have that 
\begin{align*}
\Pr (E_2 \mid E_1)&=\Pr(|r \cap T| \ge \frac{\epsilon m}{2})=1- \Pr(|r \cap T| \le \frac{\epsilon m}{2})  \\
&\ge 1- \Pr(|r \cap T| \le \frac{\alpha m}{2})  \ge 1- \frac{\alpha m}{ (\frac{\alpha m}{2})^{2}} \ge 1-\frac{4}{m \alpha} 
\end{align*} 
Now for $m \ge \frac{8}{\epsilon} \ge \frac{8}{\alpha}$, we conclude that $\Pr (E_2 \mid E_1) \ge \frac{1}{2}$. Now consider the same experiment in a different way. Consider picking $2m$ samples according to the uniform distribution from $X$ and then equally partition them randomly between $N$ and $T$. Consider the following event for a particular set $r$.
\begin{align*}
E_r: r \cap N= \Phi \text{ and } |r \cap T| \ge \frac{\epsilon m}{2} 
\end{align*}
and therefore
\begin{align*}
E_2=\bigcup _{r:\textup{ heavy}} E_r 
\end{align*}
Let us fix $N \cup T$ and define $p=|r \cap (N \cup T)|$. In that case, we have 
\begin{align*}
\Pr(r \cap N= \Phi \mid |r \cap (N \cup T)| \ge \frac{\epsilon m}{2})=\frac{(2m-p)(2m-p-1)\dots(m-p+1)}{2m(2m-1) \dots (2m-p+1)} \le 2^{-p} \le 2^{-\frac{\epsilon m}{2}}
\end{align*}
The last inequality holds since $p \ge \frac{\epsilon m}{2}$. Now, since the VC dimension of the range space $(X,R)$ is $d$, the cardinality of the set $\{ r \cap (N \cup T) \mid r \in R\}$ is at most $\sum_{i \le d} {2m \choose i} \le (2m)^{d}$ ( see \cite{shalev2014understanding}). Notice that
\begin{align*}
\Pr(E_r)=\Pr(r \cap N= \Phi \mid |r \cap (N \cup T)| \ge \frac{\epsilon m}{2}) \Pr(|r \cap (N \cup T)| \ge \frac{\epsilon m}{2}) \le 2^{-\frac{\epsilon m}{2}}
\end{align*}
Therefore by using the union bound over the possible number of distinct events $E_r$, we have
\begin{align*}
\Pr(E_2) \le (2m)^{d} 2^{-\frac{\epsilon m}{2}}
\end{align*}
Since $\Pr(E_2 \mid E_1) \ge \frac{1}{2}$ and $\Pr( E_1 \mid E_2)=1$, we must have 
\begin{align*}
\Pr(E_1) \le 2(2m)^{d} 2^{-\frac{\epsilon m}{2}} \le \delta
\end{align*}
which is ensured by the statement of the theorem.
\end{proof}

In order to use this theorem consider the range space $(X,\mathcal{R}_{u})$ where $X$ is the set of points in $\mathcal{S}^t$ and $\mathcal{R}_u$ be the family of sets $\{x \in S^t \mid b\Big(\frac{\log n}{n}\Big)^{1/t} \le  \|u-x\|_2 \le a\Big(\frac{\log n}{n}\Big)^{1/t}, w^{T}x \ge \beta , \mathcal{A}_{L:w^{T}x=\beta} \not \subset B_t(u,r_2)\}$. We now have the following lemma about the VC Dimension of the above range space which is a straightforward extension of VC dimension of half-spaces  (\cite{shalev2014understanding}):
\begin{lemma}\label{lem:VC_dimension}
VC dimension of the range space $(X,\mathcal{R}_u) \le t+1$. 
\end{lemma}
\begin{proof}
In order to show this, consider a set $\mathcal{S}$ of $t+2$ points. Recall that the convex hull of a set $S$ of points $\{x_i\}_{i=1}^{n}$ is the set 
\begin{align*}
C(S)=\{ \sum \lambda_{i}x_{i} \mid \sum \lambda_i=1, \lambda_i \ge 0  \}.
\end{align*}
By Radon's lemma (\cite{shalev2014understanding}) we have that the set of points $\mathcal{S}$ can be partitioned into two sets $\mathcal{S}_1$ and $\mathcal{S}_2$ such that their convex hulls intersect.
Let $p \in \mathcal{S}_1$ be a point in that intersection. Assume there exist a hyperplane such that
\begin{align*}
w^{T}x_i  \le w_0,  \forall x_i  \in \mathcal{S}_1 \\
 w^{T}x_i \ge w_0, \forall x_i  \in \mathcal{S}_2.
\end{align*}
Since $p$ is in the convex hull of $\mathcal{S}_1$ we must have that $w ^{T}p \le w_0$. But then,
\begin{align*}
w^{T}p= \sum_{i:x_i \in \mathcal{S}_2} \lambda_i w^{T} x_i  >( \sum_{i \in \mathcal{S}_2} \lambda_i ) \min_{i: x_i \in \mathcal{S}_2} w^{T}x_i= \min (w^{T} x_i)>w_0.
\end{align*}
which is a contradiction. Hence it is not possible to shatter $t+2$ elements and therefore the VC dimension of this range space is at most $t+1$.
\end{proof}

Using the results shown above, we are ready to prove the following Lemma.

\begin{lemma*}(Restatement of Lemma \ref{lem:high_stuff1})
If we sample $n$ nodes from $S^t$ according to ${\rm RAG}_t(n,[r_1,r_2])$ with $r_1=b\Big(\frac{\log n}{n}\Big)^{1/t}$, $r_2=a\Big(\frac{\log n}{n}\Big)^{1/t}$ , then for every node $u$ and every hyperplane $L$ passing through $u$ such that $\mathcal{A}_L \not \subset B(u,r_2)$, node $u$ has a neighbor on both sides of the hyperplane $L$ with probability at least $1-\frac{1}{n}$ provided 
\begin{align*}
(a/2)^t-b^t \ge \frac{8|S^t|(t+1)}{c_t}
\end{align*} and $a > 2b.$
\end{lemma*}

\begin{proof}
Recall that the volume of $B(u,r_1,r_2)$ is $c_t(r_2^t-r_1^t)=\frac{c_t \log n}{n}(a^t-b^t)$. According to Corollary \ref{cor:ratio}, whenever $a>2b$, $ \frac{\min (|\mathcal{R}_{L}^{1},\mathcal{R}_{L}^{2}|)}{|\mathcal{R}_{L}^{1}|+|\mathcal{R}_{L}^{2}|} \ge \delta$ where $\delta =  \frac{ \left(a/2\right)^t -b^t}{ a^t -b^t}$ when the hyperplane $L$ satisfies the conditions of the lemma. In that case, we have that for all $r \in R_u$,
\begin{align*}
 \Pr_{x \sim_U X}(x \in r) \ge \frac{\delta c_t \log n}{ n |S^t|}(a^t-b^t).
\end{align*}
Since VC Dimension of $(X,R_u) \le t+1$ and $n$ points are sampled from $X$, the conditions of Lemma \ref{thm:VC_dim} is satisfied for $\eta=\frac{2}{n^{2}}$ if
\begin{align*}
n \ge \max \Big ( \frac{8n|S^t|(t+1)}{\delta c_t \log n (a^t-b^t)}\log \frac{8 n|S^t|(t+1)}{\delta c_t\log n (a^t-b^t)},\frac{8n|S^t|}{\delta c_t \log n (a^t-b^t)} \log n  \Big)
\end{align*} 
Since $\lim_{n \rightarrow \infty} \frac{1}{\log n} \log \frac{8|S^t| n(t+1)}{\delta c_t\log n (a^t-b^t)} \rightarrow 1$ for constant $t$, hence we have that 
\begin{align*}
a^t-b^t \ge \frac{8|S^t|(t+1)}{c_t\delta}.
\end{align*} 
By taking a union bound over all the $n$ range spaces $(X,R_u)$ corresponding to the $n$ nodes and applying the statement of Lemma \ref{thm:VC_dim}, we have proved the lemma.
\end{proof}

